\documentclass{article} 
\usepackage{kr}

\usepackage{times}

\usepackage{todonotes}

\usepackage{amsthm}
\usepackage{thmtools} 

\newcommand{\citeappendix}[1]{Appendix~\ref{#1}\xspace}
\newcommand{\citetheappendix}[0]{the appendix\xspace}

\usepackage[utf8]{inputenc} 
\usepackage[T1]{fontenc}    
\usepackage{hyperref}       
\usepackage{url}            
\usepackage{booktabs}       
\usepackage{amsfonts}       
\usepackage{nicefrac}       
\usepackage{microtype}      
\usepackage{xcolor}         

\usepackage{cuted}
\usepackage{scalerel}
\usepackage{xspace}
\usepackage{booktabs}
\usepackage{amsmath}
\usepackage{xspace}

\usepackage[bb=stixtwo]{mathalpha}

\usepackage{amsthm}
\usepackage{amssymb}
\usepackage{stmaryrd}
\usepackage{tikz}
\usepackage{pgfplots}
\pgfplotsset{compat=1.18}
\usepackage{todonotes}
\usepackage{subcaption}
\usepackage{multirow}

\usetikzlibrary{arrows.meta,positioning}

\newcommand{\quantlogic}[0]{$q\mathcal{L}$\xspace}

\newcommand{\set}[1]{{\{#1\}}}

\renewcommand{\phi}{\varphi}

\newcommand\setvertices V
\newcommand\setedges E
\newcommand\labeling \ell

\newcommand{\NTexpression}{\xi}

\newcommand{\QFBAPAcons}{\chi}
\newcommand{\QFBAPAsat}{\mathsf{sat}}

\newcommand{\aGNN}{\mathcal{A}}
\newcommand{\semanticsof}[1]{[[#1]]}

\newcommand{\statetv}[2]{x_{#1}(#2)}

\newcommand{\multiset}[1]{\{\{#1\}\}}

\newcommand{\lequiv}{\leftrightarrow}
\newcommand{\limp}{\rightarrow}

\newcommand{\hintikkaset}{H}

\newcommand{\Ap}{Ap}
\newcommand{\modalitynumber}{\sharp}

\newcommand{\globmodalitynumber}
{{\sharp_{\scaleto{{g}\mathstrut}{6pt}}}}
\newcommand{\istrue}[1]{\mathbb{1}{#1}}

\newcommand{\localmodality}{\Box}
\newcommand{\globalmodality}
{\Box_g}

\newcommand{\logicKsharpone}{\ensuremath{K^{\modalitynumber}}\xspace}

\newcommand{\logicKsharpglobone}{\ensuremath{K^{\modalitynumber, \globmodalitynumber}}\xspace}

\newcommand{\ALC}{\ensuremath{\mathcal{ALC}}\xspace}
\newcommand{\ALCQ}{\ensuremath{\mathcal{ALCQ}}\xspace}
\newcommand{\ALCSCC}{\ensuremath{\mathcal{ALCSCC}}\xspace}
\newcommand{\ALCSCCpp}{\ensuremath{\ALCSCC^{++}}\xspace}

\newcommand\sigmabold{\vec{\sigma}}

\newcommand{\universe}{\mathcal U}
\newcommand{\QFBAPA}{QFBAPA\xspace}

\newcommand{\quantQFBAPA}{\ensuremath{\text{QFBAPA}_\setnumbers}\xspace}

\newcommand{\num}[1]{#1}
\newcommand{\transKA}[0]{\overline{\tau}}
\newcommand{\transTK}[0]{\underline{\tau}}

\newcommand{\trFE}{mod2expr}

\newcommand{\verificationtask}[1]{\textsc{vt#1}\xspace}
\newcommand{\verificationtaskbounded}[1]{\textsc{vt'#1}\xspace}

\newcommand{\featuresset}{F}
\newcommand{\setnumbers}{\mathbb{K}}
\newcommand{\expression}{\vartheta} 
\newcommand{\activationfunction}{\alpha}
\newcommand{\agreggationfunction}{agg}
\newcommand{\agreggationfunctionglobalreadout}{agg_{\scaleto{{\forall}\mathstrut}{6pt}}}
\newcommand{\sem}[2]{[[#1]]_{#2}}
\newcommand{\semone}[1]{[[#1]]}

\declaretheorem[]{theorem}
\declaretheorem[sibling=theorem]{lemma}
\declaretheorem[sibling=theorem]{remark}

\declaretheorem[sibling=theorem]{definition}
\declaretheorem[sibling=theorem]{example}
\declaretheorem[sibling=theorem]{claim}
\declaretheorem[sibling=theorem]{corollary}

\newcommand{\nblayers}{L}
\newcommand{\layer}{\mathcal L}
\newcommand{\AGG}{\mathit{agg}}
\newcommand{\GAGG}{agg_{\scaleto{{\forall}\mathstrut}{6pt}}}
\newcommand{\COMB}{\mathit{comb}}
\newcommand{\CLS}{\mathit{cls}}

\newcommand{\semanticsQFBAPA}{\sigma}

\newcommand{\Erdos}{\text{Erd$\ddot{\text{o}}$s–R$\acute{\text{e}}$nyi}\xspace}

\newcommand\suchthat\mid

\newcommand{\tile}[6]{
	\draw[fill=#3] (#1,#2) -- (#1,#2+1) -- (#1+0.5, #2+0.5) -- (#1, #2);
	\draw[fill=#4] (#1,#2+1)-- (#1+1,#2+1) -- (#1 +0.5, #2+0.5) -- (#1, #2+1);
	\draw[fill=#5] (#1+1,#2)-- (#1+1,#2+1) -- (#1 +0.5, #2+0.5) -- (#1 +1, #2);
	\draw[fill=#6] (#1,#2)-- (#1 +1,#2) -- (#1 +0.5, #2+0.5) -- (#1, #2);
	\draw[fill=black] (#1+0.5, #2+0.5) -- (#1+0.6, #2+0.4) -- (#1+0.4, #2+0.4) -- cycle;
}

\newcommand{\tilered}{red!80!white}
\newcommand{\tileyellow}{yellow!50!white}

\usepackage{cleveref}
\usepackage[inline]{enumitem}

\title{
Verifying Quantized GNNs With Readout Is Decidable But Highly Intractable
}

\author{%
Artem Chernobrovkin$^1$\and
Marco Sälzer$^2$\and
François Schwarzentruber$^{3}$\and
Nicolas Troquard$^1$ \\
\affiliations
$^1$Gran Sasso Science Institute (GSSI), Viale F.~Crispi, 7 -- 67100 L'Aquila, Italy\\
$^2$Rheinland-Pfälzische Technische Universität (RPTU), Kaiserslautern, Germany\\
$^3$ENS de Lyon, CNRS, Université Claude Bernard Lyon 1, Inria, LIP, UMR 5668, 69342, Lyon, France\\
\emails
\{artem.chernobrovkin, nicolas.troquard\}@gssi.it,
marco.saelzer@rptu.de,
francois.schwarzentruber@ens-lyon.fr
}

\usepackage{fontawesome5}

\begin{document}

\maketitle

\begin{abstract}
We introduce a logical language for reasoning about quantized aggregate-combine graph neural 
networks with global readout (ACR-GNNs). We provide a logical characterization and use it to prove 
that verification tasks for quantized GNNs with readout are (co)NEXPTIME-complete. This result 
implies that the verification of quantized GNNs is computationally intractable, prompting substantial 
research efforts toward ensuring the safety of GNN-based systems. We also experimentally demonstrate 
that quantized ACR-GNN models are lightweight while maintaining good accuracy and generalization 
capabilities with respect to non-quantized models.
\end{abstract}


\section{Introduction}
\label{section:introduction}

Graph neural networks (GNNs) are models used for classification and regression tasks on graph-structured data, 
including node-level and graph-level tasks. GNNs find applications in recommendation systems in social networks~\cite{DBLP:journals/kbs/SalamatLJ21},
knowledge graphs~\cite{Zi22KG}, chemistry~\cite{Reiser22GGNmaterialchemistry}, drug discovery~\cite{XIONG20211382}, etc. 
Like several other machine learning models, GNNs are difficult to interpret, understand, and verify.
This poses a significant issue for their adoption, morally and legally, with the enforcement of regulatory policies
such as the EU AI Act~\cite{AI-act}. Previous works lay the foundation for analyzing them using formal logic,
as seen in~\cite{DBLP:conf/iclr/BarceloKM0RS20},~\cite{NunnSST24}, or~\cite{benedikt2025decidabilitygraphneuralnetworks_V4}.
However, many of these approaches consider \emph{idealized} GNNs in which numbers are either arbitrarily large integers
or rationals, whereas in real-world implementations, GNNs are \emph{quantized}; numbers are represented as Standard IEEE 754
64-bit floats, INT8, or FP8~\cite{DBLP:journals/corr/abs-2209-05433}.
Verification of quantized GNNs without global readout has been addressed by~\cite{Saelzer-etal-IJCAI25}. However,
global readout is a crucial component of GNNs, particularly for graph classification~\cite{DBLP:conf/iclr/XuHLJ19}.
%

In this paper, we consider three verification tasks for GNNs to check critical properties of graphs. Examples of such verification tasks are: ``Is it the case that...
\begin{itemize}
\item 
... any power plant serving over a thousand users is classified by the GNN as essential?''
\item 
... any power plant classified by the GNN as essential has at least 3 substations?''
\item 
... a power plant classified by the GNN as essential can be a wind farm whose wind turbines have a total capacity of less than 12 MW?''
\end{itemize}
They are respectively instances of \emph{sufficiency}, \emph{necessity}, and \emph{consistency} tasks that will be formally defined in due time.
Verifying critical properties of GNNs helps ensure the safety of GNN-based systems, analogous to how correctness proofs are used to guarantee the safety of traditional computer programs~\cite{hoare1969}.

\paragraph{Contribution. } 
The contribution is threefold. First, we show that verifying the three properties above for Aggregate-Combine Graph Neural Networks 
with global Readout (ACR-GNNs) is decidable and (co)NEXPTIME-complete.
This contrasts with the PSPACE-completeness without global readout as demonstrated by~\cite{Saelzer-etal-IJCAI25}.
This implies that global readout makes quantized GNN verification highly intractable\footnote{As pointed out in~\cite{DBLP:conf/nips/JoglT023}, any inference of a ACR-GNN can be simulated by a AC-GNN (Aggregate-Combine Graph Neural Networks 
\emph{without} global Readout) where the global readout mechanism is replaced by an extra complete relation $V \times V$ where $V$ is the set of vertices of the input graph. So global readout has no significant impact on the complexity of inference. On the contrary, verification tasks involve quantifications over the set of pointed graphs where this extra relation \emph{must} be~$V \times V$.}.
To prove this,
we define the logic \quantlogic, extending the one introduced by~\cite{Saelzer-etal-IJCAI25} to capture global readout.
It is expressive enough to capture quantized ACR-GNNs with arbitrary activation functions.
Moreover, \quantlogic can serve as a flexible graph property specification language reminiscent of modal 
logics~\cite{DBLP:books/cu/BlackburnRV01}. The following example illustrates the use of \quantlogic for
expressing graph properties.

\begin{example}
\label{ex:animals-humans}
    Assume a class of knowledge graphs (KGs) representing communities of people and animals, 
    where each node corresponds to an individual.
    Each individual can be Animal, Human, Leg, Fur, White, Black, etc. These concepts can be encoded 
    with features $x_0, x_1, \ldots, x_5, \ldots$ respectively, taking values $0$ or $1$. Edges in a KG 
    represent a generic `has' relationship: a human can have an animal (pet); an animal can have a human
    (owner), a leg, a fur; a fur can have a color; etc.
    Suppose that $\aGNN$ is a GNN processing those KGs and is trained to supposedly recognize dogs. 
    We can verify that the nodes recognized by $\aGNN$ are animals—arguably a critical property of the domain—by checking
    the validity (i.e., the non-satisfiability of the negation) of $\phi_{\aGNN} \rightarrow x_0 = 1$
    where $\phi_{\aGNN}$ is a \quantlogic-formula corresponding to $\aGNN$'s computation, true in exactly
    the pointed graphs accepted by $\aGNN$.
    Ideally, $\aGNN$ should not overfit the concept of dog as a perfect prototypical animal. For instance,
    three-legged dogs do exist. We can verify that
    $\aGNN$ lets it be a possibility by checking the satisfiability of the formula $\phi_{\aGNN} \land
    \Diamond^{\leq 3} (x_2 = 1)$.
    More complex \quantlogic{} formulas can be written to express graph properties
    to be evaluated against an ACR-GNN, which will be formalized later in
    Example~\ref{ex:modalities}:
    \begin{enumerate*}
        \item Has a human owner, whose pets are all two-legged. 
        \item \label{item:not-fol} A human in a community that has more than twice as many animals as humans, and more than five animals without an owner.
            \footnote{Interestingly, \quantlogic{} goes beyond the capabilities of graded modal logic and even first-order logic (FOL). The property of Item~\ref{item:not-fol} cannot be expressed in FOL.}
        \item An animal in a community where some animals have both white and black fur.
    \end{enumerate*}
\end{example}

Firstly, to prove the (co)NEXPTIME upper bound, we reuse a concept from mathematical logic called
Hintikka sets~\cite{DBLP:books/cu/BlackburnRV01}, which are complete sets of subformulas that can be true
at a given vertex of a graph. We then introduce a quantized variant of Quantifier-Free Boolean Algebra
with Presburger Arithmetic (QFBAPA) logic~\cite{kuncak-rinard-QFBAPA}, denoted by $\quantQFBAPA$, and prove
that it is in NP, similar to the original QFBAPA on integers. We reduce the satisfiability problem of \quantlogic
to that of $\quantQFBAPA$. On the other hand, (co)NEXPTIME-hardness is proven by reduction from a suitable
tiling problem. Similarly, we also add global counting to the logic $\logicKsharpone$ previously introduced
by~\cite{NunnSST24}. We show that it corresponds to AC-GNNs over $\mathbb{Z}$ with global readout and truncated ReLU 
activation functions. We prove that the satisfiability problem is NEXPTIME-complete, partially addressing a problem 
left open in the literature—that is, for the case of integer values and truncated ReLU activation functions~\cite{DBLP:conf/icalp/BenediktLMT24,benedikt2025decidabilitygraphneuralnetworks_V4}.
Details are in \citetheappendix to keep the presentation concise.

Secondly, since NEXPTIME is highly intractable—it is provably distinct from NP~\cite{nondeterministic-complexity}—we relax
the satisfiability problem of \quantlogic and ACR-GNNs, focusing on finding graph counterexamples
where the number of vertices is bounded. This problem is NP-complete, and an implementation is provided as
a proof of concept and a baseline for future research.

Finally, we experimentally demonstrate that the quantization of GNNs results in minimal
accuracy degradation. Our results confirm that quantized models retain robust predictive performance while achieving 
substantial reductions in model size and inference cost. These findings demonstrate the practical viability of quantized 
ACR-GNNs for deployment in resource-constrained environments.

\paragraph{Outline.}
We introduce quantized ACR-GNNs in \Cref{section:background} and present experimental results justifying their practical utility 
in \Cref{section:experiment}.
\Cref{section:quantlogic} defines the logic \quantlogic, discusses its expressivity,
and defines ACR-GNN verification tasks.
\Cref{section:upperbound} provides the (co)NEXPTIME membership for the satisfiability problem of \quantlogic and the ACR-GNN verification tasks.
In \Cref{section:lowerbound}, we show that these problems are (co)NEXPTIME-complete.
\Cref{section:verificationtool} discusses the relaxation making the problems (co)NP-complete. We complement these results by reporting on the results of a prototype 
implementation (\citeappendix{appendixsection:toolperformance}) for verifying said 
GNNs in a bounded setting, highlighting the hardness of verifying quantized GNNs.

\paragraph{Related work.}
\cite{DBLP:conf/iclr/BarceloKM0RS20} showed that ACR-GNNs are at least expressive as $\text{FOC}_2$, two-variable first-order logic with counting. 
\cite{DBLP:journals/corr/abs-2508-06091} showed that there are ACR-GNNs which can be captured by first-order logic but not $\text{FOC}_2$.
Further work has explored the logical expressiveness of GNN variants in more detail. 
Notably,~\cite{NunnSST24} and~\cite{DBLP:conf/icalp/BenediktLMT24} introduced 
logics to exactly characterize the capabilities of different forms of GNNs. 
Similarly,~\cite{DBLP:conf/kr/CucalaG24} analyzed Max-Sum-GNNs through the lens of Datalog. 
~\cite{Saelzer-etal-IJCAI25} considered the expressivity of GNN with quantized parameters but without global readout. 
~\cite{AhvonenHKL24} offered several logical characterizations of recurrent GNNs over floats and real numbers.

On the verification side,~\cite{HenzingerLZ21} 
studied the complexity of verification of quantized feedforward neural networks (FNNs), while~\cite{SalzerL21,SalzerL23} investigated 
reachability and reasoning problems for general FNNs and GNNs. 
Some verification methods have been proposed, via integer linear programming (ILP) by~\cite{DBLP:conf/aaai/000200DWZB24}, 
and~\cite{10.1145/3551349.3556916}, and via model checking by~\cite{DBLP:journals/corr/abs-2106-05997}.

From a logical perspective, reasoning over structures involving arithmetic constraints is closely tied to several well-studied logics.
Relevant work includes Kuncak and Rinard’s decision procedures for QFBAPA (~\cite{kuncak-rinard-QFBAPA}), as well as further developments by 
~\cite{DBLP:journals/japll/DemriL10},~\cite{DBLP:conf/ecai/BaaderBR20},~\cite{DBLP:conf/fsttcs/BednarczykOPT21}, and~\cite{GallianiKT23}. 
These logics form the basis for the characterizations established in~\cite{NunnSST24,DBLP:conf/icalp/BenediktLMT24}.

Quantization techniques in neural networks exist, with surveys such as~\cite{gholami2022survey,Nagel2021AWP} providing comprehensive overviews focused on maintaining model accuracy. Although most practical advancements target convolutional neural networks (CNNs), many of the underlying principles extend to GNNs as well~\cite{zhou2020graph}. NVIDIA has demonstrated hardware-ready quantization strategies~\cite{DBLP:journals/corr/abs-2004-09602?nvidia}, and frameworks like PyTorch~\cite{ansel2024pytorch2} support both post-training quantization and quantization-aware training (QAT), the latter simulating quantization effects during training to improve low-precision performance. QAT has been particularly effective in closing the gap between quantized and full-precision models, especially for highly compressed or edge-deployed systems~\cite{jacob2018quantization}. In the context of GNNs, \cite{tailor2021degreequant} proposed Degree-Quant, incorporating node degree information to mitigate quantization-related issues. 
Based on this,
~\cite{zhu2023rm} introduced $A^2Q$, a mixed-precision framework that adapts bitwidths 
on graph topology to achieve high compression with minimal performance loss.






\newcommand{\setZ}{\mathbb Z}


\section{Background}
\label{section:background}

\paragraph{Quantized numbers} Let $\setnumbers$ be a set of quantized numbers, and let $n$ denote the \emph{bitwidth} of $\setnumbers$,
which is the number of bits required to represent a number in $\setnumbers$. The bitwidth $n$ is written in unary;
this is motivated by the fact that $n$ is small and that we would, in any case, need to allocate $n$-bit consecutive memory for storing a number.
Formally, we consider a sequence $\setnumbers_1, \setnumbers_2, \ldots$ corresponding to bitwidths $1$, $2$, etc.,
but we retain the notation $\setnumbers$ for simplicity.
We suppose that $\setnumbers$ saturates; for example, if $x \geq 0$ and $y \geq 0$, then $x+y \geq 0$ (i.e., there is no modulo behavior like in \texttt{int} in C, for instance).
We assume that $1 \in \setnumbers$.

We consider Aggregate-Combine Graph Neural Networks with global Readout (ACR-GNNs),
a standard class of message-passing GNNs \cite{DBLP:conf/iclr/BarceloKM0RS20,DBLP:conf/icml/GilmerSRVD17}. 

\paragraph{Graphs}
A \emph{(labeled directed) graph} $G$ is a tuple $(\setvertices, \setedges, \labeling)$ such that $\setvertices$ is 
a finite set of vertices, $\setedges \subseteq \setvertices \times \setvertices$ a set of directed edges and $\labeling$ 
is a mapping from~$\setvertices$ to $\setnumbers^m$, for some 
integer $m$, which is the dimension of the node embedding.
Given a graph $G$ and vertex $u \in \setvertices$, we call $(G,u)$ a \emph{pointed graph}.

\paragraph{ACR-GNNs}

We define quantized Aggregate-Combine GNNs with Readout (ACR-GNN).

An \emph{ACR-GNN layer} $\layer = (\COMB, \AGG, \GAGG)$ is a tuple where $\COMB : \setnumbers^{3m} \to \setnumbers^{m'}$ is a so-called \emph{combination function}, $\AGG$ is a so-called \emph{local aggregation function}, mapping multisets of vectors from $\setnumbers^m$ to a single vector from $\setnumbers^{m}$, $\GAGG$ is a so-called \emph{global aggregation function}, also mapping multisets of vectors from $\setnumbers^m$ to a single vector from $\setnumbers^m$. We call $m$ the \emph{input dimension} of layer $\layer$ and $m'$ the \emph{output dimension} of layer $\layer$.
Then, a \emph{ACR-GNN} is a tuple $(\layer^{(1)}, \dotsc, \layer^{(\nblayers)}, \CLS)$ where $\layer^{(1)}, \dotsc, \layer^{(\nblayers)}$ are 
$\nblayers$ ACR-GNN layers and $\CLS : \setnumbers^m \rightarrow \set{0, 1}$ is a \emph{classification function}. We assume that all GNNs are well-formed in the sense that
output dimension of layer $\layer^{(i)}$ matches input dimension of layer $\layer^{(i+1)}$ as well
as output dimension of $\layer^{(L)}$ matches input dimension of $\CLS$.

Let $G = (\setvertices, \setedges, x_0)$ be a graph 
with initial vertex labeling
$x_0 : \setvertices \rightarrow \setnumbers^m$,
Then, the $i$-th layer of $\aGNN$ computes an updated state $x_i(u)$ for all $u \in V$ by
\[
\begin{aligned}
\statetv{i}u :=\;&
\COMB\Bigl(
  \statetv{i-1}u,\;
  \AGG\bigl(\multiset{\statetv{i-1}v \mid (u,v) \in \setedges}\bigr), \\
&\qquad
  \GAGG\bigl(\multiset{\statetv{i-1}v \mid v \in \setvertices}\bigr)
\Bigr)
\end{aligned}
\]
\noindent
where $\AGG$, $\GAGG$, and $\COMB$ are respectively the local aggregation, global aggregation and combination function of the $i$-th layer. We call each $x_i$ a \emph{state of $G$}. Let $(G,u)$ be a pointed graph. We write $\aGNN(G,u)$ to denote the application of $\aGNN$ to $(G,u)$, which is formally defined as $\aGNN(G,u) = \CLS(x_L(u))$ where $x_L$ is the state of $G$ computed by $\aGNN$ after layer $L$. This corresponds informally to a binary classification of vertex $u$.

We concentrate on a specific subclass where both $\AGG$ and $\GAGG$ perform summation over vectors, and where $\COMB(x,y,z) = N(x,y,z)$ is a classical feedforward neural network (FNN) using activation function $\activationfunction$.
The classification function is a linear threshold: $\CLS(x) = \sum_i a_i x_i \geq 1$ with weights $a_i \in \setnumbers$.
We assume that all arithmetic operations are executed according to the arithmetic related to $\setnumbers$.
It is assumed that the context makes clear the $\setnumbers$ and arithmetic being used.
We use $\sem{\aGNN}{}$ to denote the set of pointed graphs $(G,u)$ such that $\aGNN(G,u) = 1$.
An ACR-GNN $\aGNN$ is satisfiable if $\sem{\aGNN}{}$ is non-empty.
The \emph{satisfiability problem} for ACR-GNNs is:
    Given a ACR-GNN $\aGNN$, decide whether $\aGNN$ is satisfiable.

\section{Quantization Effects on Accuracy, Performance and Model Size}
\label{section:experiment}
To begin on the right foot, we demonstrate that the quantized ACR-GNNs defined before are indeed of practical interest.
%
To this aim, we investigate the application of dynamic Post-Training Quantization (PTQ) to ACR-GNNs.
As a reference, we used the models described and analyzed in~\cite{DBLP:conf/iclr/BarceloKM0RS20}, using their implementation~\cite{BarceloGit2021} as the baseline. 
Experiments used two datasets: one synthetic (\Erdos model) and one real-world dataset (Protein-Protein Interactions (PPI) benchmark by~\cite{zitnik2017predicting}). 
In the original work, experiments were made with two activation functions: Rectified Linear Unit (ReLU) and truncated ReLU. Since the models of \Cref{section:background} can handle arbitrary activation functions, in our experiments, we used several types of activation function: Piecewise linear (ReLU, ReLU6, and trReLU), Smooth unbounded (GELU and SiLU), Smooth bounded (Sigmoid), and Smooth ReLU-like (Softplus and ELU). 
The quantization method is implemented in PyTorch~\cite{ansel2024pytorch2,pytorch_quantization}, dynamic PTQ transforms a pre-trained floating-point model into a quantized version without requiring retraining. In this approach, model weights are statically quantized to INT8, while activations remain in floating-point format until they are dynamically quantized at compute time. This hybrid representation enables efficient low-precision computation using INT8-based matrix operations, thereby reducing the memory footprint and improving the inference speed. The PyTorch implementation applies per-tensor quantization to the weights and stores the activations as floating-point values between operations to strike a balance between accuracy and performance.

Synthetic graphs (Table~\ref{tab:dataset_summary} of Appendix~\ref{sec:Experiments}) were generated using the dense \Erdos model, a classical approach to constructing random graphs. Each graph includes five initial node colors, encoded as one-hot feature vectors. Following~\cite{DBLP:conf/iclr/BarceloKM0RS20}, labels were assigned using formulas from the logic fragment $\text{FOC}_2$. Specifically, a hierarchy of classifiers $\alpha_i(x)$ was defined as follows: $\alpha_0(x):= \text{Blue}(x)$ and $\alpha_{i+1}(x):= \exists^{[N,M]}y\left(\alpha_i(y) \wedge \neg E(x,y)\right),$
where $\exists^{[N,M]}$ denotes the quantifier ``there exist between $N$ and $M$ nodes" satisfying a given condition. Each classifier $\alpha_i(x)$ can be expressed within $\text{FOC}_2$, as the bounded quantifier can be rewritten using $\exists^{\geq N}$ and $\neg \exists^{\geq M + 1}$. Each property $p_i$ corresponds to a classifier $\alpha_i$ with $i \in \{1,2,3\}$.
For the analysis, we collected training time, model size, and accuracy for both datasets. We list the principal findings of the analysis. More detailed statistics can be found in the Appendix~\ref{sec:Experiments}. According to our experimental flow, we first examine the training time. For both datasets, we found that piecewise-linear activation functions consistently achieve the shortest training time (Table~\ref{tab:training-times} and Table~\ref{tab:ppi:nonqua:training-time} of Appendix~\ref{sec:Experiments}).
Moreover, for computational efficiency, we found that the model with the activation function (Smooth) was consistently the slowest, regardless of the datasets (Table~\ref{tab:worst:training-times} and Table~\ref{tab:ppi:nonqua:minmax} of Appendix~\ref{sec:Experiments}).
We performed an analysis of the accuracy of both data across all activation families, the observed $\Delta_{acc}$ values are generally within $\pm 1\%$. In addition,we observed a drop in the accuracy of the baseline models after two layers. 

These findings highlight the advantages of quantized ACR-GNN models of \Cref{section:background} with respect to non-quantized models, striking a remarkable balance between model size and accuracy.
From the evaluation of 9,000 models on \Erdos and 80 models on real-world data, we selected the following hyperparameters: 1–3 layers, hidden dimension 32 for \Erdos experiments and 256 for PPI, and a learning rate of 0.01. 

We evaluated ACR-GNNs (PPI data) across a range of bit-width configurations (32, 16, 8, 7, 6, 5, 4, 2) and also under finer-grained reductions from 8 to 4 bits. Across both datasets, 16- and 8-bit models preserved accuracy within a small margin of the full-precision baseline. For the PPI dataset (Tables~\ref{tab:ppi-bitwidth-drops-RELU}--\ref{tab:ppi-bitwidth-drops-ELU} in Appendix~\ref{sec:Experiments}) exhibits strong robustness to dynamic post-training quantization: performance is preserved down to 6 bits across all activation functions, with only minor deviations (less than 3\%)  observed in the from 7 to 6 and from 6 to 5 bit transitions. The only marked degradation appears when moving into the very low-precision range (5 bits and below).

\section{Logic \quantlogic}
\label{section:quantlogic}

We set up a logical framework called \quantlogic extending the logic introduced in \cite{Saelzer-etal-IJCAI25} with global aggregation: it is a \emph{lingua franca} to represent GNN computations and properties on graphs.

\paragraph{Syntax.}
Let $\featuresset$ be a finite set of features and $\setnumbers$ be some finite-width arithmetic. We consider a set of \emph{expressions} defined by the following grammar in Backus-Naur form:
$$\expression ::= c \mid x_i \mid \activationfunction (\expression) \mid \agreggationfunction (\expression) \mid \agreggationfunctionglobalreadout(\expression) \mid
\expression + \expression \mid c\times\expression$$
where $c$ is a number in  $\setnumbers$, $x_i$ is a feature in $\featuresset$, $\alpha$ is a symbol for denoting the activation function, and $\agreggationfunction$ and $\agreggationfunctionglobalreadout$ denote the aggregation function for local and global readout respectively.
We define the set of subexpression of $\expression$, denoted by $\mathit{sub}(\expression)$ inductively by $\mathit{sub}(c) = \{c\}$, $\mathit{sub}(x_i) = \{x_i\}$, $\mathit{sub}(\activationfunction(\expression)) = \{\activationfunction(\expression)\} \cup \mathit{sub}(\expression)$, 
$\mathit{sub}(\AGG(\expression)) = \{\AGG(\expression)\} \cup \mathit{sub}(\expression)$, $\mathit{sub}(\GAGG(\expression)) = \{\GAGG(\expression)\} \cup \mathit{sub}(\expression)$, $\mathit{sub}(\expression + \expression') = \{\expression + \expression\} \cup \mathit{sub}(\expression) \cup \mathit{sub}(\expression')$, and $\mathit{sub}(c \times \expression) = \{c \times \expression\} \cup \mathit{sub}(c) \cup \mathit{sub}(\expression')$.   
A \emph{formula} is defined as $\expression \geq k$
where $\expression$ is an expression and $k$ is an element of $\setnumbers$.
If $-1 \in \setnumbers$ we write $-\expression$ instead of $(-1) \times \expression$ for any expression $\expression$. Other standard abbreviations can be used as well.
%
%
%
%

%

Formulas are represented as directed acyclic graphs, aka circuits, meaning that we do not repeat the same expressions several times.
For instance, the formula $\agreggationfunction(x_1 + x_2) + (x_1 + x_2) \geq 3$ can be represented as the DAG given in Figure~\ref{figure:dag}. Formulas can also be represented by a sequence of assignments via new fresh intermediate variables. For instance:
$
    y := x_1 + x_2, 
    z := \agreggationfunction(y) + y, 
    res := z \geq 3
$.

\begin{figure}[t]
\newcommand{\subtermdrawarrow}[4]{
\draw[-latex] (#1) edge[bend left=#3] node [above] {#4} (#2);
}
    \begin{tikzpicture}[xscale=1,yscale=0.3,every node/.style={inner sep=0.5mm}]
     \node (geq) at (-3, 0) {$\geq$};
     \node (3) at (0, 0) {3};
     \node (+1) at (0, 2) {+};
     \node (agg) at (2, 2) {$\agreggationfunction$};
     \node (+2) at (2, 0) {+};
     \node (x1) at (4, 2) {$x_1$};
     \node (x2) at (4, 0) {$x_2$};

     \subtermdrawarrow{geq} {3} {-5} {}
     \subtermdrawarrow{geq} {+1} {10} {$\cdot$}
     \subtermdrawarrow{+1} {+2} {-15} {}
     \subtermdrawarrow{agg} {+2} {0} {}    
   
    \subtermdrawarrow{+1} {agg} {5} {}    
    \subtermdrawarrow{+2} {x1} {15} {}
    \subtermdrawarrow{+2} {x2} {-5} {}
    \end{tikzpicture}
    
    \caption{DAG representation of the formula 
    $\agreggationfunction(x_1 + x_2) + (x_1 + x_2) \geq 3$.}
    \label{figure:dag}
\end{figure}

\paragraph{Semantics.}

Consider a graph $G = (\setvertices, \setedges, \labeling)$, where vertices in $\setvertices$ are labeled via a labeling function $\labeling: V \rightarrow \setnumbers^{|F|}$ with feature values.
The value of an expression $\expression$ in a vertex
$u \in V$ is denoted by 
$\sem{\expression}{G, u}$ 
and is defined by induction on~$\expression$:
\begin{equation*}
\begin{aligned}
\sem{c}{G,u} &= c, \\
\sem{x_i}{G,u} &= \labeling(u)_i, \\
\sem{\expression + \expression'}{G,u} &=
  \sem{\expression}{G,u} +_{\setnumbers} \sem{\expression'}{G,u}, \\
\sem{c \times \expression}{G,u} &=
  c \times_{\setnumbers} \sem{\expression}{G,u}, \\
\sem{\activationfunction(\expression)}{G,u} &=
  \semone{\activationfunction}\!\bigl(\sem{\expression}{G,u}\bigr), \\
\sem{\agreggationfunction(\expression)}{G,u} &=
  \Sigma_{\substack{v \mid u \setedges v}}\sem{\expression}{G,v}, \\
\sem{\agreggationfunctionglobalreadout(\expression)}{G,u} &=
  \Sigma_{v \in \setvertices}\sem{\expression}{G,v}.
\end{aligned}
\end{equation*}

We define
$\sem{\expression \geq k}{} = \{ G,u \mid \sem{\expression}{G,u} \geq_{\setnumbers} \sem{k}{G,u} \}$ (we write $\geq$ for the syntactic symbol and $\geq_\setnumbers$ for the relation defined in $\setnumbers$). A formula $\phi$ is satisfiable if $\sem{\phi}{}$ is non-empty.
The \emph{satisfiability problem} for \quantlogic is:
    Given a \quantlogic-formula $\phi$, decide whether $\phi$ is satisfiable.

\paragraph{Simulating a modal logic in the logic \quantlogic.}
\label{subsection:logicinquantlogic}

We show that extending \quantlogic with modal operators \cite{DBLP:books/cu/BlackburnRV01} does not increase the expressivity. We can even compute an equivalent \quantlogic without Boolean connectives and without modal operators in poly-time. It means that formulas like  $\phi_{\aGNN_1} \rightarrow x_0 = 1$ 
or
$\phi_{\aGNN_1} \land \Diamond^{\leq 3} (x_2 = 1)$ have poly-size equivalent formulas in \quantlogic.

%
%
Let $Atm_0$ be the set of atomic formulas of \quantlogic of the form $\expression \geq 0$. We suppose that $\expression$ takes integer values. In general, $\expression \geq k$ is an atomic formula equivalent to $\expression - k \geq 0$.
Without loss of generality, we thus assume that formulas of \quantlogic are over $Atm_0$.
Let modal \quantlogic be the
propositional logic on $Atm_0$ extended with modalities and a restricted variant of graded modalities where number $k$ in $\setnumbers$ in the following way:
\begin{align*}
  &\sem{\Box\phi}{} &&=&& \{ G,u \mid {G,v} \in \sem{\phi}{} \text{ for every } v \text{ s.t. } u\setedges v \}\\
  &\sem{\Box_g\phi}{} &&=&& \{ G,u \mid  {G,v} \in \sem{\phi}{} \text{ for every } v \text{ in } \setvertices \} \\
  &\sem{\Diamond^{\geq k}\phi}{}  &&=&& \{ G,u \mid  |\{ G,v \mid  u\setedges v,  G,v \in \sem{\phi}{} \}| \geq_{\setnumbers} k \} \\
  &\sem{\Diamond_g^{\geq k}\phi}{}  &&=&& \{ G,u \mid  |\sem{\phi}{}| \geq_{\setnumbers} k \}
\end{align*}
and modalities $\Diamond^{\leq k}\phi$ and 
$\Diamond_g^{\leq k}\phi$ defined the same way but with $\leq_{\setnumbers}$.
\begin{example}[continuation of Example~\ref{ex:animals-humans}]
\label{ex:modalities}
    We first define a few simple formulas to characterize the concepts of the domain.
    Let $\phi_A := x_0 = 1$ (Animal),
    $\phi_H := x_1 = 1$ (Human),
    $\phi_L := x_2 = 1$ (Leg),
    $\phi_F := x_3 = 1$ (Fur),
    $\phi_W := x_4 = 1$ (White), and
    $\phi_B := x_5 = 1$ (Black).:
    \begin{enumerate}
        \item Has a human owner, whose pets are all two-legged:
        $\Diamond(\phi_H \land \Box(\phi_A \rightarrow \Diamond^{=2}\phi_L))$.
        \item \label{item:not-fol} A human in a community that has more than twice as many animals as humans, and more than five animals without an owner:
        $
            \phi_H \land (\agreggationfunctionglobalreadout(x_0) - 2 \times \agreggationfunctionglobalreadout(x_1) \geq 0) \land \Diamond^{\geq 5}_g(\phi_A \land \Box(\lnot \phi_H)).$
        \item An animal in a community where some animals have white and black fur: $
            \phi_A \land \Diamond_g(\Diamond (\phi_F \land \Diamond \phi_W) \land \Diamond (\phi_F \land \Diamond \phi_B)).$
    \end{enumerate}
\end{example}
We can see the boolean operator $\lnot$, and the various modalities as functions from $\mathrm{Atm}_0$ to $\mathrm{Atm}_0$, and the boolean operator $\lor$ as a function from $\mathrm{Atm}_0 \times \mathrm{Atm}_0$ to $\mathrm{Atm}_0$ defined as follows, where $\alpha$ is ReLU ($\alpha(\expression) = \expression$ if $\expression \geq_\setnumbers 0$ and $\alpha(\expression) = 0$ otherwise). 
%
\begin{equation*}
\scalebox{0.9}{$
\begin{aligned}
&f_\lnot(\expression \geq 0) &&:=&& - \expression - 1 \geq 0 \\
&f_\lor(\expression_1 \geq 0, \expression_2 \geq 0)  &&:=&& \expression_1 + \alpha(\expression_2 - \expression_1) \geq 0\\
&f_{\Box}(\expression \geq 0) &&:=&& \agreggationfunction(-\alpha(-\expression)) \geq 0\\
&f_{\Diamond^{\geq k}}(\expression \geq 0) &&:=&& \agreggationfunction(\alpha(\expression+1) - \alpha(\expression)) - k \geq 0\\
&f_{\Diamond^{\leq k}}(\expression \geq 0) &&:=&& k - \agreggationfunction(\alpha(\expression + 1) - \alpha(\expression)) \geq 0
\end{aligned}
$}
\end{equation*}
%
%
For the corresponding global modalities ($f_{\Box_g}(\expression \geq 0)$, $f_{\Diamond_g^{\geq k}}(\expression \geq 0)$, and $f_{\Diamond_g^{\leq k}}(\expression \geq 0)$), it suffices to use $\agreggationfunctionglobalreadout$ in place of $\agreggationfunction$. 
%
%
The previous transformations can be generalized to arbitrary formulas of a modal version of \quantlogic as follows. 
\begin{equation*}
\scalebox{0.9}{$
\begin{aligned}
    &\trFE(\expression \geq 0) &&:=&& \expression \geq 0 \\
    &\trFE(\lnot \phi) &&:=&& f_\lnot (\trFE(\phi))\\
    &\trFE(\phi_1 \lor \phi_2) &&:=&& f_\lor(\trFE(\phi_1), \trFE(\phi_2))\\
     &\trFE(\boxplus \phi) &&:=&& f_{\boxplus}(\trFE(\phi)),
\end{aligned}
$}
\end{equation*}
where $\boxplus \in \{\Box, \Box_g, \Diamond^{\geq k}, \Diamond_g^{\geq k}, \Diamond^{\leq k}, \Diamond_g^{\leq k}\}$.

We can show that formulas of modal \quantlogic can be captured by a single expression $\expression \geq 0$. This is a consequence of the following lemma
\footnote{For simplicity, we do not present how to handle $\expression \geq 0$ when $\expression$ is not an integer. We could introduce several activation functions $\alpha$ in \quantlogic, one of them could be interpreted as the Heavyside step function. In \Cref{def:hintikkaset}, Point~4 is just repeated for each $\activationfunction$.}, proven in \citetheappendix.

\begin{restatable}{lemma}{lemmamodallogicintoquantlogic}
    \label{lem:trFE-corr}
    Let $\phi$ be a formula of modal \quantlogic. The formulas $\phi$ and $\trFE(\phi)$ are equivalent.
\end{restatable}

\subsection{Link between \quantlogic and ACR-GNN verification}
We focus on the following decision problems in context of the verification of ACR-GNNs: 
\begin{itemize}
    \item (\verificationtask{1}, sufficiency) Given a GNN $\aGNN$ and a \quantlogic{} formula $\phi$, decide whether $\sem{\phi}{} \subseteq \sem{\aGNN}{}$.
    \item (\verificationtask{2}, necessity) Given a GNN $\aGNN$ and a \quantlogic{} formula $\phi$,  decide whether $\sem{\aGNN}{} \subseteq \sem{\phi}{}$.
    \item (\verificationtask{3}, consistency) Given a GNN $\aGNN$ and a \quantlogic{} formula $\phi$, decide whether $\sem{\phi}{} \cap \sem{\aGNN}{} \neq \emptyset$.
\end{itemize}

Informally, these problems are described as follows: problem \verificationtask{1} consists of verifying that pointed graphs satisfying the specification $\phi$ are classified positively by the GNN $\aGNN$, while \verificationtask{2} is the reverse, namely verifying that any pointed graphs positively classified by $\aGNN$ satisfies $\phi$. Problem \verificationtask{3} consists in verifying whether $\aGNN$ can classify some pointed graph satisfying $\phi$.

Essential to our results is the straightforward insight that \quantlogic where $\alpha$ is given by ReLU 
and ACR-GNN are equally expressive.
\begin{theorem}
    \label{th:formulatognn}
    Given an ACR-GNN $\aGNN$, we can compute \quantlogic-formula in poly-time in the size of $\aGNN$ with $\sem{\aGNN}{} = \sem{\phi_\aGNN}{}$. Vice-versa, with $\alpha =$ ReLU, given a \quantlogic formula $\phi$, we can compute an ACR-GNN
    $\aGNN_\phi$ in poly-time, with $\sem{\aGNN_\phi}{} = \sem{\phi}{}$.
    
\end{theorem}
\begin{proof}
    We begin by arguing that for each $\aGNN$ there is an equivalent \quantlogic-formula. 
    The key insight needed here is that \quantlogic is tailored to capture the computation of an ACR-GNN. 
    Let $\aGNN = (\layer^{(1)}, \dots, \layer^{(\nblayers)}, \CLS)$
    be an ACR-GNN. Let $\mathcal{L}^{(1)}=(\COMB^{(1)}, \AGG,\GAGG)$ be the first layer of $\aGNN$ computing 
    $\COMB^{(1)}((x_1, \dots, x_m), (y_1, \dots, y_m), (z_1, \dots, z_m))$ where $\COMB^{(1)}$ is represented
    by a feedforward neural network $N_1$ using activation function $\alpha$. The vectors
    $(y_1, \dots, y_m)$ and $(z_1, \dots, z_m)$ correspond to the vectors aggregated by $\AGG$
    and $\GAGG$, respectively.
    Let $v_{11}(x_1', \dots, x_n') = \alpha(b_1 + \sum_{i=1}^n w_i x_i')$ be a node of $N_1$
    in the first layer where $x_i'$ corresponds to either some $x_i$, $y_i$ or $z_i$ as determined by $\COMB^{(1)}$.
    We build an equivalent \quantlogic-expression $\expression_{v_{11}}$ combining $c$,
    $\expression + \expression$, $c \times \expression$ and $\alpha(\expression)$ to represent
    the arithmetic and activation function, and we use $x_i$ in the case that $x_i' = x_i$,
    we use $\AGG(x_i)$ in the case that $x_i' = y_i$, and we use $\GAGG(x_i)$ in the case that $x_i' = z_i$.
    We do so for all nodes $v_{1j}$ in the first layer of $N_1$. This results in a sequence
    of expressions $\expression_{\mathcal{L}^{(1)}11}, \dots, \expression_{\mathcal{L}^{(1)}1k_1}$
    where $k_1$ is the layer size of the first layer of $N_1$. For subsequent layers, we use the same
    construction, but the expressions $\expression_{\mathcal{L}^{(1)}1j}$ as atomic. This results in a sequence
    of expressions $\expression_{\mathcal{L}^{(1)}m1}, \dots, \expression_{\mathcal{L}^{(1)}m_1k}$
    where $k$ is the size of the last layer $m_1$ of $N_1$, and also the output dimensionality of $\mathcal{L}^{(1)}$. We continue as before for GNN layer $\mathcal{L}^{(i)}$,
    but use $\varphi_{\mathcal{L}^{(i-1)}m_jj}$ as atomic expressions.
    The final classification condition $\CLS$ is then given by the formula $\expression \geq 1$ where 
    $\expression$ uses $\varphi_{\mathcal{L}^{(\nblayers)}m_\nblayers j}$ as atomic as well as multiplicative parameters
    $a_i$ determined by $\CLS$.

    For the other direction, namely showing that for each \quantlogic-formula there is an equivalent ACR-GNN,
    let $\varphi$ be a \quantlogic-formula $\expression \geq k$ over variables $x_1, \dots x_m$.
    We remark that the argument follows the same line of reasoning as in expressivity results of previous works
    such as \cite{Saelzer-etal-IJCAI25,NunnSST24,DBLP:conf/iclr/BarceloKM0RS20}. We note $\mathit{sub}(\expression)$ the set of subexpressions of $\expression$.
    Let $\expression_1, \dots, \expression_n$
    be an enumeration of $\mathit{sub}(\expression)$ such that if $\expression_i \in \mathit{sub}(\expression_j)$ then we have $i \leq j$.
    We construct an ACR-GNN $\aGNN_\phi$ with $n$ layers as follows. First, the input and output dimensionality
    of each layer of $\aGNN$ is $n$, meaning that we have one dimension per subexpression. W.l.o.g.\
    we assume that the first $m$ dimensions correspond to the values of the basic variables $x_1, \dots, x_m$
    of $\expression$. Furthermore, we assume for all layers $\mathcal{L}^{(i)}$ that they do not adjust any other dimension
    than the $i$th one of a state. Note that this can easily be realised using FNN with ReLU activation. Then, if $\expression_i=x_j$ layer $\mathcal{L}^{(i)}$ effectively does nothing in the
    sense that the FNN $N_i$ representing the combination function simply represents the identity. In the cases
    $\expression_i=c$, $\expression_i=\expression+\expression$, $\expression_i=x \times \expression$,
    or $\expression_i = \activationfunction(\expression)$ we also only utilise FNN $N_i$ to do the arithmetic operation
    or applying activation function $\activationfunction$. Note that this involves the assumption that previous layers
    evaluated all necessary subexpressions already. Finally, for the cases $\expression_i=\AGG(\expression)$
    and $\expression_i = \GAGG(\expression)$ we simply use $\AGG$ and $\GAGG$ functions of layer $\mathcal{L}^{(i)}$
    to resolve it. Then, in the $\CLS$ function of $\aGNN_\phi$ we evaluate $\varphi = \expression \geq k$.
\end{proof}

The following example highlights the construction idea from ACR-GNNs to \quantlogic of Theorem~\ref{th:formulatognn}.
\begin{example}
To reason formally about ACR-GNNs, we represent their computations using \quantlogic.
Consider an ACR-GNN $\mathcal{A}$ with two layers of input and output dimension~2, using summation for aggregation, truncated ReLU as activation $\activationfunction(x) = \max(0, \min(1,x)) = \sem{\alpha}{}(x)$, and a classification function $2x_1 - x_2 \geq 1$. The combination functions are:
\begin{equation*}
\scalebox{1.0}{$
\begin{aligned}
\COMB_1(\boldsymbol{x}, \boldsymbol{y}, \boldsymbol{z}) &:=
\begin{pmatrix}
\activationfunction(2x_1 + x_2 + 5y_1 - 3y_2 + 1) \\
\activationfunction(-x_1 + 4x_2 + 2y_1 + 6y_2 - 2)
\end{pmatrix}^\intercal, \\
\COMB_2(\boldsymbol{x}, \boldsymbol{y}, \boldsymbol{z}) &:=
\begin{pmatrix}
\activationfunction(3x_1 - y_1 + 2z_2) \\
\activationfunction(-2x_1 + 5y_2 + 4z_1)
\end{pmatrix}^\intercal.
\end{aligned}
$}
\end{equation*}
Note that this assumes that $\mathcal{A}$ operates over $\setnumbers$ with at least three bits.
Then, the corresponding \quantlogic formula $\varphi_\mathcal{A}$ is given by:
\begin{align*}
    \psi_1 &:= \alpha(2x_1 + x_2 + 5\agreggationfunction(x_1) - 3\agreggationfunction(x_2) + 1) \\
    \psi_2 &:= \alpha(-x_1 + 4x_2 + 2\agreggationfunction(x_1) + 6\agreggationfunction(x_2) - 2) \\
    \chi_1 &:= \alpha(3\psi_1 - \agreggationfunction(\psi_1) + 2(\agreggationfunctionglobalreadout(\psi_2))) \\
    \chi_2 &:= \alpha(-2\psi_1 + 5(\agreggationfunction(\psi_2)) + 4\agreggationfunctionglobalreadout(\psi_1)) \\
    \varphi_A &:= 2(\chi_1) - \chi_2 \geq 1. 
\end{align*}

To sum up, given a GNN $\aGNN$, we compute \quantlogic-formula in poly-time in the size of $\aGNN$ with $\sem{\aGNN}{} = \sem{\phi_\aGNN}{}$. The vice-versa direction works analogously.
\end{example}

Finally, ACR-GNN verification tasks can be solved by reduction to the satisfiability problem of \quantlogic:
\begin{itemize}
    \item \verificationtask{1} by checking that $\phi \land \lnot \phi_\aGNN$ is not satisfiable
    \item  \verificationtask{2} by checking that $\lnot \phi \land \phi_\aGNN$ is not satisfiable
    \item \verificationtask{3} by checking that $\phi \land \phi_\aGNN$ is satisfiable
\end{itemize} 

\section{Complexity Upper Bound}
\label{section:upperbound}

In this section, we prove the NEXPTIME membership of reasoning in modal quantized logic, and also of solving of ACR-GNN verification tasks (by reduction to the former).

\begin{theorem}
Suppose that $\sem{\activationfunction}{}$ is computable in exponential-time in the bit-width $n$ of $\setnumbers$.
The satisfiability problem of \quantlogic is decidable and in NEXPTIME, and so is \verificationtask{3}. Consequently, \verificationtask{1} and \verificationtask{2} are in coNEXPTIME.
\label{th:modalquantizedlogicNEXPTIME}
\end{theorem}

In order to prove Theorem~\ref{th:modalquantizedlogicNEXPTIME}, 
we borrow from the ideas behind the proof of the NEXPTIME membership of concept satisfiability in the description logic $\mathcal{ALCSCC}^{++}$ 
from \cite{DBLP:conf/ecai/BaaderBR20} (which adapts a proof in \cite{DBLP:conf/gcai/BaaderE17} of the complexity of \ALC ECBoxes consistency). We adapt it to the logic \quantlogic. 
The difference resides in the definition of Hintikka sets and the treatment of quantization. The idea is to encode the constraints of a \quantlogic-formula $\phi$ in a formula of exponential length of a quantized version of QFBAPA, that we prove to be in NP.

\subsection{Hintikka Sets}

Consider \quantlogic-formula $\phi$. Let $E(\phi)$ be the set of subexpressions in $\phi$. For instance, if $\phi$ is
$\agreggationfunction(\activationfunction(x_2 + \agreggationfunctionglobalreadout(x_1))) \geq 5$ 
then 
    $E(\phi) =  
\{
\agreggationfunction(\activationfunction(x_2 + \agreggationfunctionglobalreadout(x_1))), \activationfunction(x_2 + \agreggationfunctionglobalreadout(x_1)), x_2,$
    $\agreggationfunctionglobalreadout(x_1), x_1\}.$
From now on, we consider equality formulas that are of the form $\expression {=} k$ where $\expression$ is a subexpression of $\phi$ and $k \in \setnumbers$.

\begin{definition}
\label{def:hintikkaset}
    A Hintikka set $\hintikkaset$ for $\phi$ is a subset of $\{\expression {=} k \mid \expression \in E(\phi), k \in \setnumbers \}$, such that:
    \begin{enumerate}
        [noitemsep]
        \item\label{hintikka:unique} For all $\expression \in E(\phi)$, there is a unique value $k \in \setnumbers$ such that $\expression = k \in \hintikkaset$

        \item\label{hintikka:constant} For all $c \in E(\phi) \cap \setnumbers$, $c = c \in \hintikkaset$ 
        
        \item\label{hintikka:sum} For all $\expression_1 {+} \expression_2 \in E(\phi)$, if $\expression_1 {=} k_1, \expression_2 {=} k_2 \in \hintikkaset$ then $\expression_1 {+} \expression_2 {=} k' \in \hintikkaset$, where $k' = k_1 {+}_\setnumbers k_2$
       \item\label{hintikka:product} For all $c \times \expression \in E(\phi)$, if $\expression = k \in H$ then $ c\times\expression {=} k' \in H$ where $k' = c \times_\setnumbers k$
       
        \item\label{hintikka:activation} 
        For all $\activationfunction(\expression) \in E(\phi)$, $\expression {=} k \in \hintikkaset$ then 
        $\activationfunction(\expression) {=} k' \in \hintikkaset$ where 
        $k' = \sem{\activationfunction}{}(k)$
    \end{enumerate}
\end{definition}

Informally, a \emph{Hintikka set} contains equality subformulas obtained from a choice of a value for each subexpression of~$\phi$ (point \ref{hintikka:unique}), provided that the set is consistent \emph{at the current vertex} (points 1--5). The notion of Hintikka set does not take any constraints about $\agreggationfunction$ and $\agreggationfunctionglobalreadout$ into account since checking consistency of aggregation would require information about the neighbors or the whole graph.

\begin{example}
If $\phi$ is 
$\agreggationfunction(\activationfunction(x_2 + \agreggationfunctionglobalreadout(x_1))) \geq 5$,
then an example of Hintikka set is:
$
\{\agreggationfunction(\activationfunction(x_2 + \agreggationfunctionglobalreadout(x_1)) = 8,$
$ \activationfunction(x_2 + \agreggationfunctionglobalreadout(x_1)) = 9,
x_2 + \agreggationfunctionglobalreadout(x_1) = 9,  x_2 = 7, $
$\agreggationfunctionglobalreadout(x_1) = 2, x_1 = 5\}.$
\end{example}

\begin{restatable}
    {proposition}{propositionNumberHintikkasets}
    The number of Hintikka sets is bounded by $2^{n|\phi|}$ where $|\phi|$ is the size of $\phi$, and $n$ is the bitwidth of $\setnumbers$.
\end{restatable}

\subsection{Quantized Version of QFBAPA}
\newcommand{\setvariable}{S}
\newcommand{\setexpression}{S}

A QFBAPA formula is a propositional formula where each atom is either an inclusion of sets or equality of sets or linear constraints (\cite{kuncak-rinard-QFBAPA}, and \citeappendix{sec:qfbapa}). Sets are denoted by Boolean algebra expressions, e.g., $(\setvariable \cup \setvariable') \setminus \setvariable''$, or $\universe$ where $\universe$ denotes the set of all points in some domain. Here $\setvariable$, $\setvariable'$, etc., are set variables.
Linear constraints are over $|\setvariable|$ denoting the cardinality of the set denoted by the set expression $S$.
For instance, 
    the QFBAPA-formula $(windfarm \subseteq powerplant) \land (|powerplant| + |\universe \setminus windfarm| \geq 6) \land (|powerplant| < 2)$ is read as `all wind farms are power plants, and the number of power plants plus the number of entities that are not wind farms is greater than $6$ and the number of power plants is smaller than $2$'.

We now introduce a \emph{quantized} version $\quantQFBAPA$ of QFBAPA. It has the same syntax as QFBAPA except that 
numbers in expressions are in $\setnumbers$. 
Semantically, every arithmetic expression is interpreted in $\setnumbers$. For each set expression $\setexpression$, 
the interpretation of $|\setexpression|$ is not the cardinality $c$ of the interpretation of $\setexpression$, but the result of the computation
$1+1+\ldots + 1$
in $\setnumbers$ with $c$ occurrences of 1 in the sum.

We consider $\setnumbers$ that saturates, meaning that if $x + y$ exceeds the upper bound limit of $\setnumbers$, there is a special value denoted by $+\infty$ such that $x+y = +\infty$. 

\begin{restatable}{proposition}{quantQFBAPAinNP}
\label{proposition:quantQFBAPAinNP}
If the bitwidth $n$ is in unary, and if $\setnumbers$ saturates, then satisfiability in \quantQFBAPA is in NP.
\end{restatable}

\subsection{Reduction to \quantQFBAPA}

\newcommand{\subexpression}{\expression'}

Let
$\phi$ be a formula of \quantlogic. 
For each Hintikka set $H$, we introduce the set variable $X_H$ that intuitively represents the $H$-vertices, i.e., the vertices in which all subformulas in $H$ hold.
Moreover, for each formula $\subexpression=k$ in $\{\expression {=} k \mid \expression \in E(\phi), k \in \setnumbers \}$, we introduce the set variable $X_{\subexpression=k}$ that intuitively represents the vertices in which $\subexpression=k$ holds. 
We capture this with two \quantQFBAPA-formulas.
\Cref{form:reduc1} expresses that $\set{X_H}_H$ form a partition of the universe. \Cref{form:reduc2} makes the bridge between variables $X_{\subexpression=k}$ and $X_H$.
\begin{align}
    \bigwedge_{H \neq H'}
    (X_{H} {\cap} X_{H'}  {=} \emptyset) \land 
    (\bigcup_{H} X_{H} {=} \universe)\label{form:reduc1} \\ 
   \bigwedge_{\subexpression \in E(\phi)} \bigwedge_{k \in \setnumbers} (X_{\subexpression=k} = \!\!\bigcup_{H \suchthat \subexpression=k \in H} \!\!\! X_{H})\label{form:reduc2}
\end{align}
We introduce also a variable $S_H$ that denotes the set of all successors of some $H$-vertex. If there is no $H$-vertex then the variable $S_H$ is just irrelevant. 
The \quantQFBAPA-formula in \Cref{eq:qfbapaaggregation} encodes the semantics of $\agreggationfunction(\expression)$. More precisely, it says that for all subexpressions $\agreggationfunction(\expression)$, for all values $k$, for all Hintikka sets $H$ 
containing formula $\agreggationfunction(\expression) {=} k$, if there is some $H$-vertex (i.e., some vertex in $X_H$),
then the aggregation obtained by summing over the successors of some $H$-vertex is $k$:

\begin{align}
\bigwedge_{\agreggationfunction(\expression) \in E(\phi)}&\bigwedge_{k \in \setnumbers}
\bigwedge_{H  \suchthat \agreggationfunction(\expression) {=} k \in H
}   \label{eq:qfbapaaggregation} \\ 
&((X_H \neq \emptyset)
  \rightarrow
  \sum_{k' \in \setnumbers} |S_H \cap X_{\expression=k'}| \times k' = k) \nonumber 
\end{align}
In the previous sum, we partition $S_H$ into subsets $S_H \cap X_{\expression=k'}$ for all possible values $k'$. Each contribution for a successor in $S_H \cap X_{\expression=k'}$ is $k'$. We rely here on the fact\footnote{This is true for some fixed-point arithmetics but not for floating-point arithmetics. See \citeappendix{appendixsection:distributivity}.} that $(1+1+\ldots+1) \times k' = k' + k' + \ldots + k'$. We also fix a specific order over the values $k'$ in the summation 
(it means that $\agreggationfunction(\expression)$ is computed as follows: first order the successors according to the taken values of $\expression$ in that specific order, then perform the summation).
Finally, the semantics of $\agreggationfunctionglobalreadout$ is captured by the \quantQFBAPA-formula:
\begin{align}
\bigwedge_{\agreggationfunctionglobalreadout(\expression) \in E(\phi)}&\bigwedge_{k \in \setnumbers} \label{eq:qfbapaaggregationglob} \\ &X_{\agreggationfunctionglobalreadout(\expression) = k} \neq \emptyset
    \rightarrow 
    \sum_{k' \in \setnumbers} |X_{\expression=k'}| \times k' = k \nonumber
\end{align}
Note that intuitively \Cref{{eq:qfbapaaggregationglob}} implies that for $X_{\agreggationfunctionglobalreadout(\expression) = k}$ is interpreted as the universe, for the value $k$ which equals the semantics of $\sum_{k' \in \setnumbers} |X_{\expression=k'}| \times k'$.

\newcommand{\formularules}{\psi_{1\text{--}4}}
Finally, given $\phi = \expression \geq k$, we define
\[tr(\phi) := \formularules \land \bigvee_{k' \geq_\setnumbers k}X_{\expression=k'}\neq \emptyset\]
where $\formularules$ is the conjunction of Formulas 1--4. The function $tr$ requires to compute all the Hintikka sets. 
In particular, all of them need to satisfy point~4 of \Cref{def:hintikkaset}. Therefore, we get the following when $\sem{\activationfunction}{}$ is computable in time exponential in $n$.

\begin{restatable}{proposition}{NEXPTIMEupperboundtrcomputable}
 $tr(\phi)$ is of exponential size and is computable in time exponential in $|\phi|$ and $n$.
\end{restatable}
Crucially, the translation $tr(\ldots)$ preserves satisfiability.
\begin{restatable}{proposition}{NEXPTIMEupperboundtrcorrectness}
\label{prop:uppperbound-tr-correctness}
Let $\phi$ be a formula of \quantlogic.
    $\phi$ is satisfiable iff $tr(\phi)$ is \quantQFBAPA satisfiable.
\end{restatable}

\begin{proof}[Proof sketch.]
From left to right, suppose that the \quantlogic formula $\phi$ is satisfiable. Then there exist a graph $G = (\setvertices, \setedges, \labeling)$ and $u\in \setvertices$ such that $G,u\models\phi$. We define the \quantQFBAPA-substitution $\semanticsQFBAPA$ as follows:
\begin{itemize}
\item $\semanticsQFBAPA(\universe) := \setvertices$,
\item $\semanticsQFBAPA(X_{\expression = k}) := \{w\in \setvertices \mid \sem{\expression}{G,w}=k\}$,
\item $\semanticsQFBAPA(X_H) := \{w\in \setvertices \mid G,w\models \bigwedge H\}$,
\item $\semanticsQFBAPA(S_H) := \{w \mid \exists v\in \semanticsQFBAPA(X_H) \text{ and } (v,w)\in \setedges\}$.
\end{itemize}
We then show that $\semanticsQFBAPA$ satisfies $tr(\phi)$.

Every vertex induces a unique Hintikka set. Hence the sets $\semanticsQFBAPA(X_H)$ form a partition of $\setvertices$, yielding \Cref{form:reduc1}. By definition of $\semanticsQFBAPA(X_{\expression=k})$ and the construction of Hintikka sets, \Cref{form:reduc2} also holds.

Consider $\agreggationfunction(\expression)=k\in H$. If $\semanticsQFBAPA(X_H) = \emptyset$, then the implication in \Cref{eq:qfbapaaggregation} is trivial. Otherwise, let $v\in \semanticsQFBAPA(X_H)$. Then
$\semanticsQFBAPA( \sum_{k'\in\setnumbers} |S_H\cap X_{\expression=k'}|\times k' )$
counts the contribution of all successors of $v$, grouped by their $\expression$-value.
This coincides with $\sum_{w\mid v\setedges w}\sem{\expression}{G,w}$, hence equals $k$. So \Cref{eq:qfbapaaggregation} is satisfied by $\semanticsQFBAPA$.
\Cref{eq:qfbapaaggregationglob} is also satisfied by $\semanticsQFBAPA$ by the same reasoning.

Finally, as $G,u\models \phi$ and $\phi$ is of the form $\expression \geq k$, there exists $k' \geq_{\setnumbers} k$ with $u\in \semanticsQFBAPA(X_{\expression=k'})$. Thus the last disjunct of $tr(\phi)$ holds and $\semanticsQFBAPA$ satisfies $tr(\phi)$.

\smallskip
Conversely, from right to left, let $\semanticsQFBAPA$ be a \quantQFBAPA-substitution that satisfies $tr(\phi)$.
Using the definition of $tr(\phi)$ and \Cref{def:hintikkaset},
we can prove that $\semanticsQFBAPA$ also satisfies the following formulas:
\begin{equation}
\label{eq:partition-X}
    \bigwedge_{\expression} \bigwedge_{k \not =_\setnumbers k'} (X_{\expression = k} \cap X_{\expression = k'} = \emptyset) \land (\bigcup_{k} X_{\expression = k} = \universe)\\
\end{equation}
\begin{equation}
    \label{eq:constant-X}
    \bigwedge_c X_{c = c} = \universe
\end{equation}
\begin{equation}
    \label{eq:sum-X}
    \bigwedge_{\expression_1 + \expression_2 \in E(\phi)} \bigwedge_{k_1, k_2} X_{\expression_1 = k_1} \cap X_{\expression_2 = k_2} \subseteq X_{\expression_1 + \expression_2 = k_1 +_\setnumbers k_2}
\end{equation}
\begin{equation}
    \label{eq:prod-X}
    \bigwedge_{c \times \expression \in E(\phi)} \bigwedge_{k} X_{\expression = k} = X_{c \times \expression = c \times_\setnumbers k}
\end{equation}
\begin{equation}
    \label{eq:activ-X}
    \bigwedge_{\alpha(\expression) \in E(\phi)} \bigwedge_k X_{\expression = k} = X_{\activationfunction(\expression) = \sem{\activationfunction}{}(k)}
\end{equation}
(It is presented as \Cref{lemma:tr-entailed-equations} in \citetheappendix.)

We construct a graph $G = (\setvertices, \setedges, \labeling)$ such that:
\begin{itemize}
\item $\setvertices := \semanticsQFBAPA(\universe)$,
\item $\setedges := \set{(u, v) \mid \exists H, u \in \semanticsQFBAPA(X_H) \text{ and } v \in \semanticsQFBAPA(S_H)}$,
\item $\labeling(v)_i := k \text{ where } v \in \semanticsQFBAPA(X_{x_i = k})$.
\end{itemize}
The set of vertices is the (\quantQFBAPA) universe, and we add an edge between any $H$-vertex $u$ and a vertex $v \in \semanticsQFBAPA(S_H)$. The labeling is well defined because of
\Cref{eq:partition-X}.

The substitution $\semanticsQFBAPA$ satisfies $tr(\phi)$, and therefore the formula $\bigvee_{k' \geq_\setnumbers k} X_{\expression = k'} \not = \emptyset$.
So there is $k' \geq_\setnumbers k$ and $u \in \universe$ such that $u \in \semanticsQFBAPA(X_{\expression = k'})$. that $G, u \models \phi$.

We can then prove by structural induction on $\expression$ that
if $u\in \semanticsQFBAPA(X_{\expression=k})$
then $\sem{\expression}{G,u} = k$. For the base cases, we use \Cref{eq:constant-X} for constants, while the definition of $\labeling$ is used directly for variables. The cases of sum, product, and activation functions use \Cref{eq:sum-X}, \Cref{eq:prod-X}, and \Cref{eq:activ-X} respectively. The cases of aggregations exploit \Cref{eq:qfbapaaggregation} and \Cref{eq:qfbapaaggregationglob}.

\end{proof}

Finally, in order to check whether a \quantlogic-formula $\phi$ is satisfiable,
we construct a $\quantQFBAPA$-formula $tr(\phi)$ of exponential size in exponential time. 
As the satisfiability problem of $\quantQFBAPA$ is in NP, we obtain that the satisfiability problem of \quantlogic is in NEXPTIME. 
We proved~\Cref{th:modalquantizedlogicNEXPTIME}.
\begin{remark}
Our methodology can be generalized to reason in subclasses of graphs. For instance, we may tackle the problem of satisfiability in a graph where vertices are of bounded degree bounded by $d$. To do so, we add the constraint $\bigwedge_{H} |S_H| \leq d$.
\end{remark}

\section{Complexity Lower Bound}
\label{section:lowerbound}

The NEXPTIME upper-bound is tight.
Having defined modalities in \quantlogic and stated Lemma~\ref{lem:trFE-corr}, Theorem~\ref{th:hardess-quantlogic} is proven by adapting the proof of NEXPTIME-hardness of deciding the consistency of \ALCQ-$T_C$Boxes presented in \cite{DBLP:journals/jair/Tobies00}. 

NEXPTIME-hardness is proven via a reduction from the tiling problem by Wang tiles of a torus of size $2^n \times 2^n$. A Wang tile is a square with colors, e.g., 
\tikz[baseline=0mm, scale=0.3]{
\tile 0 0{\tilered}{\tileyellow}{\tilered}{white}
}, \tikz[baseline=0mm, scale=0.3]{
\tile 0 0{\tileyellow}{\tileyellow}{\tilered}{\tileyellow}
}, etc.
The problem takes as input a number $n$ in unary, and Wang tile types, and an initial condition---let us say the bottom row is already given. The objective is to decide whether the torus of $2^n \times 2^n$ can be tiled while colors of adjacent Wang tiles match.
A slight difficulty resides in adequately capturing a two-dimensional grid structure---as in Figure~\ref{fig:grid-hardness-quantlogic}---with only a single relation. To do that, we introduce special formulas $\phi_E$ and $\phi_N$ to indicate the direction (east or north). In the formula computed by the reduction, we also need to bound the number of vertices corresponding to tile locations by $2^n \times 2^n$. Thus $\setnumbers$ needs to encode $2^n \times 2^n$. We need a bit-width of at least $2n$.

\begin{restatable}{theorem}{theoremlowerbound}
    \label{th:hardess-quantlogic}
    The satisfiability problem in \quantlogic is NEXPTIME-hard, and so is the ACR-GNN verification task \verificationtask{3}. The ACR-GNN verification tasks \verificationtask{1} and \verificationtask{2} are coNEXPTIME-hard. The results already hold when $\alpha =$ ReLU. 
\end{restatable}

%
%

\newcommand{\nodecoordinate}[2]{u_{#1, #2}}
\newcommand{\nodeElabel}{u_E}
\newcommand{\nodeNlabel}{u_N}

\begin{figure}
    \centering
\usetikzlibrary{graphs,positioning}

\begin{tikzpicture}[
xscale=2.3, yscale=0.9, font=\tiny,
every node/.style = {inner sep=0.5mm}
]

\tikzset{nodeNE/.style={inner sep = 0.5mm}}

\def\maxcoord{2.5}
\newcommand{\twopowernminusone}{2^n{-}1}

\node (00) at (0, 0) {$\nodecoordinate00$};
\node (01) at (0, 1) {$\nodecoordinate01$};
\node (10) at (1, 0) {$\nodecoordinate10$};
\node (11) at (1, 1) {$\nodecoordinate11$};
\node (nn) at (\maxcoord, \maxcoord) {$\nodecoordinate{\twopowernminusone}{\twopowernminusone}$};
\node (0n) at (0, \maxcoord) {$\nodecoordinate0 {\twopowernminusone}$};
\node (n0) at (\maxcoord, 0) {$\nodecoordinate{\twopowernminusone} 0$};

\node[nodeNE] (N1) at (0, 0.5) {$\nodeNlabel$};
\node[nodeNE] (E1) at (0.5, 0) {$\nodeElabel$};
\node[nodeNE] (N2) at (1, 0.5) {$\nodeNlabel$};
\node[nodeNE] (E2) at (0.5, 1) {$\nodeElabel$};

\node[nodeNE] (Nn) at (0, \maxcoord + 0.5) {$\nodeNlabel$};
\node[nodeNE] (Nn2) at (\maxcoord, \maxcoord + 0.5) {$\nodeNlabel$};
\node[nodeNE] (En) at (\maxcoord + 0.5, 0) {$\nodeElabel$};
\node[nodeNE] (En2) at (\maxcoord + 0.5, \maxcoord) {$\nodeElabel$};

\draw[->] (00) edge (N1);
\draw[->] (00) edge (E1);
\draw [->](01) edge (E2);
\draw[->] (10) edge (N2);

\draw[->] (N1) edge (01);
\draw[->] (E1) edge (10);
\draw[->] (N2) edge (11);
\draw[->] (E2) edge (11);

\draw[dotted] (01) edge (0n);
\draw[dotted] (10) edge (n0);
\draw[dotted] (0n) edge (nn);
\draw[dotted] (n0) edge (nn);

\draw[dotted] (11) edge (1, \maxcoord);
\draw[dotted] (11) edge (\maxcoord, 1);

\draw[->] (0n) edge (Nn);
\draw[->] (nn) edge (Nn2);
\draw[->] (nn) edge (En2);
\draw[->] (n0) edge (En);

\draw[->] (Nn) edge [bend right] (00);
\draw[->] (En) edge [bend left] (00);
\draw[->] (Nn2) edge [bend right = 40] (n0);
\draw[->] (En2) edge [bend left] (0n);

\end{tikzpicture}

    \caption{Encoding a torus of exponential size with (modal) \quantlogic formulas. Vertices $\nodecoordinate{x}{y}$ correspond to locations $(x, y)$ in the torus while $\nodeNlabel$ and $\nodeElabel$ denote intermediate vertices indicating the direction (resp., north and east).}
    \label{fig:grid-hardness-quantlogic}
\end{figure}

\begin{remark}
    It turns out that the verification task only needs the fragment of \quantlogic 
where $\agreggationfunction$ is applied directly on an expression $\activationfunction(\expression)$. 
Indeed, this is the case when we represent a GNN in \quantlogic or when we translate logical 
formulas in \quantlogic (\Cref{lem:trFE-corr}). Reasoning 
about \quantlogic when $\setnumbers = \setZ$ and 
the activation function is truncated ReLU is also NEXPTIME-complete (see \Cref{appendixsection:NEXPTIMEcompletewithintegers}).
\end{remark}

\section{Bounding the Number of Vertices}
\label{section:verificationtool}

The satisfiability problem of \quantlogic is 
NEXPTIME-complete, thus far from tractable. The high complexity arises because counterexamples can be arbitrary large graphs. However, one usually looks for small counterexamples. 
\newcommand{\boundedgraphs}[1]{\mathcal{G}^{\leq #1}}
Let $\boundedgraphs{N}$ be the set of pointed graphs with at most $N$ vertices.
We consider the \quantlogic and ACR-GNN \emph{satisfiability problem with a bound on the number of vertices} and ACR-GNNs verification tasks: given a number $N$ given in unary, 
\begin{enumerate*} 
\item given a \quantlogic-formula $\phi$, is it the case that $\sem{\phi}{} \cap \boundedgraphs{N} \not = \emptyset$,
\item given an ACR-GNN $\aGNN$, is it the case that $\sem{\aGNN}{} \cap \boundedgraphs{N} \not = \emptyset$.
\end{enumerate*}
In the same way, we introduce the following verification tasks. Given a GNN $\aGNN$, a \quantlogic formula $\phi$, and a number $N$ in unary:
(\verificationtaskbounded{1}, sufficiency) Do we have $\sem{\phi}{} \cap \boundedgraphs{N} \subseteq \sem{\aGNN}{} \cap \boundedgraphs{N} $?
(\verificationtaskbounded{2}, necessity) Do we have $\sem{\aGNN}{} \cap \boundedgraphs{N} \subseteq \sem{\phi}{}\cap \boundedgraphs{N} $?
(\verificationtaskbounded{3}, consistency) Do we have $\sem{\phi}{} \cap \sem{\aGNN}{} \cap \boundedgraphs{N} \not = \emptyset$? 

\begin{restatable}{theorem}{theoremsatboundednpcomplete}
    The satisfiability problems with bounded number of vertices are NP-complete, so is ACR-GNN verification task $\verificationtaskbounded 3$, while the verification tasks $\verificationtaskbounded 1$ and $\verificationtaskbounded 2$ are coNP-complete.
\end{restatable}
The satisfiability NP upper bound is obtained by guessing a graph with at most $N$ vertices and then check that~$\phi$ holds. NP-hardness already holds for $\agreggationfunction$-free formulas by reduction from SAT for propositional logic using the reduction of Lemma~\ref{lem:trFE-corr}. The complexity of the bounded verification tasks follow from simple reductions.

It is possible to extend the methodology proposed in \cite{DBLP:journals/corr/abs-2106-05997} to the verification of Graph Neural Networks. However, an efficient SMT encoding for GNN verification tasks would constitute a substantial contribution in itself. In this work, we therefore restrict ourselves to a proof of concept that establishes a baseline for future research (see Appendix~\ref{appendixsection:toolperformance} for performance details).
For the verification task, we employ ESBMC (Efficient SMT-based Context-Bounded Model Checker)~\cite{esbmc2024}. The verification workflow (Figure~\ref{fig:model_to_esbmc} in Appendix~\ref{sec:Experiments}) translates a trained PyTorch model into C code that can be analyzed by ESBMC, embedding the corresponding preconditions and postconditions.

After evaluating the verification workflow, we identified two principal limitations affecting the formal verification phase. First, the process is constrained by SMT encodability: only activation functions that admit efficient logical representations can be supported, thereby restricting the architectural design space of GNNs (see Table~\ref{tab:activation-smt} in Appendix~\ref{sec:Experiments}). Second, scalability is governed by the size of the generated SMT formula, which increases with both the number of graph vertices and the model dimensionality. Our prototype evaluation confirms that even small increases in these parameters lead to substantial verification overhead due to the expansion of symbolic variables and constraints.

In particular, results on the \Erdos dataset (Table~\ref{tab:esbmc_time_relu_h2_l1} in Appendix~\ref{sec:Experiments}) empirically illustrate the state-space explosion inherent in SMT-based verification of GNNs. Increasing the graph size from one to two vertices causes a dramatic rise in verification time, highlighting the poor scalability of naive full symbolic encodings for graph-structured neural models. This behavior arises because each additional vertex introduces new symbolic variables for node features, aggregation outputs, linear transformations, and activation constraints; furthermore, ReLU activations induce case distinctions that significantly increase the number of feasible solver paths, thereby aggravating solver complexity.

\section{Conclusion and Future Work}

The main result is the (co)NEXPTIME-completeness of verification tasks for 
GNNs with global readout.
It helps to understand the inherent complexity, and demonstrates that the verification of ACR-GNNs is highly intractable. 
To further guide future research, we provide extensive experiments demonstrating that quantised GNNs are essential.
However, naive approaches to verifying GNNs, even over a set of graphs with a bounded number of vertices, quickly reach their 
limits, as expected, given the inherently high complexity.
This prompts significant efforts of the research community towards ensuring the safety of quantised GNN-based systems.



There are many directions to go. First, characterizing the modal flavor of \quantlogic---a powerful graph property specification language---for other activation functions than ReLU. 
New extensions of \quantlogic could also be proposed to tackle other classes of GNNs.
%
%
%
%
%
%
Verification of neural networks is challenging and is currently tackled by the verification community \cite{DBLP:journals/corr/abs-2501-05867}. So it will be for GNNs as well. Our verification tool with a bound on the number of vertices is still preliminary and a mere baseline for future research. One obvious path would be to improve it, to compare different approaches (bounded model checking vs.\ linear programming as in \cite{DBLP:conf/aaai/000200DWZB24}) and to apply it to real GNN verification scenarios.
Designing a practical verification procedure in the general case (without any bound on the number of vertices) and overcoming the high computational complexity is an exciting challenge for future research towards the verification of GNNs.


\paragraph{Limitations.} \label{par:limitations}
%
Section \ref{section:upperbound} and \ref{section:lowerbound} reflect theoretical results. Some practical implementations of GNNs may not fully align with them. 
In particular, the order in the (non-associative) summation over values in $\setnumbers$ is fixed in formulas (\ref{eq:qfbapaaggregation}) and (\ref{eq:qfbapaaggregationglob}). It means that we suppose that the aggregation $\agreggationfunction(\expression)$ is computed in that order too (we sort the successors of a vertex according to the values of $\expression$ and then perform the summation).





\paragraph{Reproducibility.}
The full proofs are provided in \citetheappendix. 

To ensure reproducibility, we provide a replication package containing the implementation, trained models, and verification scripts. The package is publicly available via a \href{https://github.com/francoisschwarzentruber/kr2026-Verifying-Quantized-GNNs-With-Readout-Is-Decidable-But-Highly-Intractable}{GitHub repository}~\cite{chernobrovkin2026replication}.

The implementation of the verification prototype is included in the supplementary material and can be found in the folder `\text{src\_verificationtool}'.

The replication package for the experimental evaluation of ACR-GNN quantization is provided in the folder \text{code\_notebooks\_csv}. The \text{Code} subfolder contains the Python implementation, accompanied by a \text{README.md} file that documents the project structure and provides detailed instructions for reproducing the experiments. The \text{Notebooks} subfolder includes the analysis scripts, the corresponding \text{.csv} files, and additional documentation describing their usage.

\section*{AI Declaration}
During the preparation of this work, some authors used Grammarly for grammar and spelling correction. No content was generated by artificial intelligence tools, and all scientific contributions, analyses, and conclusions are the sole responsibility of the authors.

\bibliography{biblio}
\bibliographystyle{kr}

\clearpage
\appendix

\section{Proofs of statements in the main text}
\label{app:proofs}

\lemmamodallogicintoquantlogic*

\begin{proof}
    We have to prove that for all $G, u$, we have $G, u \models \phi$ iff $G, u \models mod2expr(\phi)$.
We proceed by induction on $\phi$.

\begin{itemize}
    \item 
The base case is obvious: $G, u \models \phi$ iff $G, u \models mod2expr(\phi)$ is $G, u \models \phi$ iff $G, u \models mod2expr(\phi)$.

\item
$G, u \models \lnot \phi$ iff $G, u \not \models \phi$ 

iff (by induction) $G, u \not \models mod2expr(\phi)$

iff (by writing $mod2expr(\phi) = \expression \geq 0$) $G, u \not\models \expression \geq 0$ 

iff $G, u \models \expression < 0$ 

iff  $G, u \models \expression \leq -1$ (because we suppose that $\expression$ takes its value in the integers

iff  $G, u \models \expression + 1 \leq 0$ 

iff  $G, u \models -\expression - 1 \geq 0$.

\item $G, u \models (\phi_1 \lor \phi_2)$

iff $G, u \models \phi_1$ or $G, u \models \phi_2$

iff $G, u \models (\expression_1 \geq 0)$ or $G, u \models (\expression_2 \geq 0)$ 

iff $G, u \models \expression_1 + ReLU(\expression_2 - \expression_1) \geq 0$

Indeed, ($\Rightarrow$) if $G, u \models (\expression_1 \geq 0)$ then  $G, u \models \expression_1 + ReLU(\expression_2 - \expression_1) \geq \expression_1 \geq 0$.

If $G, u \models (\expression_2 \geq 0)$ and $G, u \models (\expression_1 < 0)$ then  $G, u \models \expression_1 + ReLU(\expression_2 - \expression_1) = \expression_1 + \expression_2 - \expression_1 = \expression_2  \geq 0$.

($\Leftarrow$) Conversely, by contraposition, if $G, u \models (\expression_2 < 0)$ and $G, u \models (\expression_1 < 0)$, then $G, u \models \expression_1 + ReLU(\expression_2 - \expression_1) = \expression_1 + \expression_2 - \expression_1 = \expression_2 < 0$ or $G, u \models \expression_1 + ReLU(\expression_2 - \expression_1) = \expression_1 + 0 = \expression_1 < 0$. In the two cases, $G, u \models \expression_1 + ReLU(\expression_2 - \expression_1)  < 0$.

\item
$G, u \models \Diamond^{\geq k} \phi$ 
iff the number of vertices $v$ that are successors of $u$ and with $G, v \models \phi$ is greater than $k$ 

iff the number of vertices $v$ that are successors of $u$ and with $G, v \models mod2expr(\phi)$ is greater than $k$ 

iff (written $\expression \geq 0$) iff the number of vertices $v$ that are successors of $u$ and with $G, v \models \expression \geq 0$ is greater than $k$ 


iff the number of vertices $v$ that are successors of $u$ and with $G, v \models ReLU(\expression+1) - ReLU(\expression) = 1$ is greater than $k$ (since we know by defining of modal \quantlogic that $\expression$ takes its value in integers) 

iff $G, u \models agg(ReLU(\expression+1) - ReLU(\expression) \geq k$

iff $G, u \models mod2expr(\Diamond^{\geq k} \phi)$

\item Other cases are similar.

\end{itemize}
\end{proof}

\propositionNumberHintikkasets*
\begin{proof}
For each expression $\expression$, we choose a number in $\setnumbers$. There are $2^n$ different numbers. There are $|\phi|$ number of expressions. So we get $(2^n)^{|\phi|} = 2^{n{|\phi|}}$ possible choices for a Hintikka set.
\end{proof}

\quantQFBAPAinNP*

\begin{proof}
\newcommand{\quantizedcardinal}[1]{k_{#1}}
Here is a non-deterministic algorithm for the satisfiability problem in \quantQFBAPA.
\begin{enumerate}
    \item 
Let $\chi$ be a \quantQFBAPA formula.

 \item 
For each set expression $B$ appearing in some $|B|$, guess a non-negative integer number $\quantizedcardinal B$ in $\setnumbers$.

 \item 
Let $\chi'$ be a (grounded) formula in which we replaced $|B|$ by $\quantizedcardinal B$.

\item 
Check that $\chi'$ is true (can be done in poly-time since $\chi'$ is a grounded formula, it is a Boolean formula on variable-free equations and inequations in $\setnumbers$). 

\item If not we reject.

\item We now build a standard QFBAPA formula $\delta = \bigwedge_{B} constraint(B)$ where:
$$constraint(B) = \begin{cases}
    |B| = \quantizedcardinal B \text{ if $\quantizedcardinal B < \infty_\setnumbers$ } \\
    |B| \geq limit \text{ if $\quantizedcardinal B = +\infty_\setnumbers$}
\end{cases}$$

where $limit$ is the maximum number that is considered as infinity in $\setnumbers$.

\item Run a non-deterministic poly-time algorithm for the QFBAPA satisfiability on $\delta$. Accepts if it accepts. Otherwise reject.

\end{enumerate}

The algorithm runs in poly-time. Guessing a number $n_B$ is in poly-time since it consists in guessing $n$ bits ($n$ in unary). Step 4 is just doing the computations in $\setnumbers$. In Step 6, $\delta$ can be computed in poly-time.

If $\chi$ is $\quantQFBAPA$ satisfiable, then there is a substitution $\semanticsQFBAPA$ such that $\semanticsQFBAPA$ satisfies $\chi$. At step 2, we guess $n_B = |\semanticsQFBAPA(B)|_\setnumbers$. The algorithm accepts the input.

Conversely, if the algorithm accepts its input, $\chi'$ is true for the chosen values $n_B$. $\delta$ is satisfiable. So there is a substitution $\semanticsQFBAPA$ such that $\semanticsQFBAPA$ satisfies $\delta$. By the definition of $constraint$, $\semanticsQFBAPA$ satisfies $\chi$.
\end{proof}

\begin{remark}
    If the number $n$ of bits to represent $\setnumbers$ is given in unary and if $\setnumbers$ is ``modulo'', then the satisfiability problem in $\quantQFBAPA$ is also in NP. The proof is similar except that now $constraint(B) = (|B| = k_B + L d_B)$ where $d_B$ is a new variable. 
\end{remark}

\NEXPTIMEupperboundtrcomputable*
\begin{proof}
In order to create $tr(\phi)$, we write an algorithm where 
each big conjunction, big disjunction, big union and big sum is replaced by a loop. For instance, $\bigwedge_{H \neq H'}$ is replaced by two inner loops over Hintikka sets. 
Note that we can check whether a candidate $H$ is a Hintikka set in exponential time in $n$ since Point 4 can be checked in exponential time in $n$ (thanks to our loose assumption on the computability of $\sem{\activationfunction}{}$ in exponential time in $n$). 
There are $2^{n|\phi|}$ many of them. In the same way, $\bigwedge_{k \in \setnumbers}$ is a loop over $2^n$ values.
There is a constant number of nested loops, each of them iterating over an exponential number (in $n$ and $|\phi|$ of elements). QED.
\end{proof}

The following is a technical lemma that will simplify the presentation of the proof of \Cref{prop:uppperbound-tr-correctness}.
\begin{lemma}\label{lemma:tr-entailed-equations}
Let $\phi$ be a \quantlogic formula, and
let $\semanticsQFBAPA$ be a substitution that satisfies $tr(\phi)$. It also satisfies the following \QFBAPA formulas.
\begin{enumerate}
    \item \label{eq:qfbapa-partition} $\bigwedge_{\expression} \bigwedge_{k \not =_\setnumbers k'} (X_{\expression = k} \cap X_{\expression = k'} = \emptyset) \land (\bigcup_{k} X_{\expression = k} = \universe)$

    \item \label{eq:qfbapa-constant} $\bigwedge_c X_{c = c} = \universe$

    \item \label{eq:qfbapa-sum} $\bigwedge_{\expression_1 + \expression_2 \in E(\phi)} \bigwedge_{k_1, k_2} X_{\expression_1 = k_1} \cap X_{\expression_2 = k_2} \subseteq X_{\expression_1 + \expression_2 = k_1 +_\setnumbers k_2}$

    \item \label{eq:qfbapa-prod} $\bigwedge_{c \times \expression \in E(\phi)} \bigwedge_{k} X_{\expression = k} = X_{c \times \expression = c \times_\setnumbers k}$

    \item \label{eq:qfbapa-activation} $\bigwedge_{\alpha(\expression) \in E(\phi)} \bigwedge_k X_{\expression = k} = X_{\activationfunction(\expression) = \sem{\activationfunction}{}(k)}$


\end{enumerate}
\end{lemma}
\begin{proof}
\Cref{eq:qfbapa-partition}:
For the first conjunct, suppose $k \not = k'$.
By \Cref{form:reduc2}, $\semanticsQFBAPA(X_{\expression = k}) = \bigcup_{H \mid \expression = k \in H} \semanticsQFBAPA(X_H)$
and
$\semanticsQFBAPA(X_{\expression = k'}) = \bigcup_{H \mid \expression = k' \in H} \semanticsQFBAPA(X_H)$.
By \Cref{def:hintikkaset}.\ref{hintikka:unique} 
$\semanticsQFBAPA(X_{\expression = k} \cap X_{\expression = k'}) = \emptyset$.
For the second conjunct. 
By 
\Cref{def:hintikkaset}.\ref{hintikka:unique}, for every Hintikka set $H$ there is a unique $k$ such that $\expression = k \in H$.
Therefore $\semanticsQFBAPA(\bigcup_H X_H) = \semanticsQFBAPA(\bigcup_k (\bigcup_{H \mid \expression = k \in H} X_H))$. By \Cref{form:reduc2}, 
we obtain $\semanticsQFBAPA(\bigcup_H X_H) = 
\semanticsQFBAPA(\bigcup_k X_{\expression = k})$.
But, by \Cref{form:reduc1}, we also have 
$\semanticsQFBAPA(\bigcup_H X_H) = \semanticsQFBAPA(\universe)$.

\Cref{eq:qfbapa-constant}:
By \Cref{form:reduc2}, $\semanticsQFBAPA(X_{c=c}) = \semanticsQFBAPA(\bigcup_{H | c= c} X_H)$.
By \Cref{def:hintikkaset}.\ref{hintikka:constant},
for every Hintikka set $H$ for $\phi$, and $c \in E(\phi) \cap \setnumbers$, we have $c = c \in H$. 
So $\semanticsQFBAPA(X_{c=c}) = \semanticsQFBAPA(\bigcup_{H} X_H)$.
By \Cref{form:reduc1}, $\semanticsQFBAPA(X_{c=c}) = \semanticsQFBAPA(\universe)$.

\Cref{eq:qfbapa-sum}:
Suppose that $u \in \semanticsQFBAPA(X_{\expression_1 = k_1} \cap X_{\expression_2 = k_2})$.
By \Cref{form:reduc2},
$\semanticsQFBAPA(X_{\expression_1 = k_1} \cap X_{\expression_2 = k_2})
= 
\semanticsQFBAPA((\bigcup_{H\mid \expression_1 = k_1 \in H} X_{H})
\cap
(\bigcup_{H\mid \expression_2 = k_2 \in H} X_{H}))
$.
With \Cref{form:reduc1}, 
this equals 
$\semanticsQFBAPA(\bigcup_{H\mid \expression_1 = k_1 \in H \land
\expression_2 = k_2 \in H}  X_{H})$.
By \Cref{def:hintikkaset}.\ref{hintikka:sum},
if $\expression_1 = k_1 \in H$ and 
$\expression_2 = k_2 \in H$, then 
$\expression_1 + \expression_2 = k_1 +_\setnumbers k_2 \in H$.

$\semanticsQFBAPA(\bigcup_{H\mid \expression_1 = k_1 \in H \land
\expression_2 = k_2 \in H}  X_{H})
\subseteq 
\semanticsQFBAPA(\bigcup_{H\mid
\expression_1 + \expression_2 = k_1 +_\setnumbers k_2 \in H}  X_{H})
$.
With \Cref{form:reduc2}
$
\semanticsQFBAPA(\bigcup_{H\mid
\expression_1 + \expression_2 = k_1 +_\setnumbers k_2 \in H}  X_{H})
=
\semanticsQFBAPA(X_{\expression_1 + \expression_2 = k_1 +_\setnumbers k_2})
$.
So $u \in 
\semanticsQFBAPA(X_{\expression_1 + \expression_2 = k_1 +_\setnumbers k_2})$.

\Cref{eq:qfbapa-prod} and 
\Cref{eq:qfbapa-activation}
    are proved in a way analogous to the previous case, exploiting \Cref{def:hintikkaset}.\ref{hintikka:product}, and \Cref{def:hintikkaset}.\ref{hintikka:activation}.

\end{proof}

\NEXPTIMEupperboundtrcorrectness*
\begin{proof}
\fbox{$\Rightarrow$}
    Let $G = (\setvertices, \setedges)$, and $G, u$ such that $G, u \models \phi$. 
    We construct a substitution $\semanticsQFBAPA$ and prove that it satisfies $tr(\phi)$.
    We set $\semanticsQFBAPA(\universe) := \setvertices$, $\semanticsQFBAPA(X_{\subexpression=k'}) := \set{w \mid \sem{\subexpression}{G, w} = k'}$ and $\semanticsQFBAPA(X_H) := \set{w \suchthat {G, w \models H}}$ where $G, w \models H$ means that for all $\subexpression=k \in H$, we have $\sem{\subexpression}{G, w} = k$.
    For all Hintikka sets $H$ define:
$\semanticsQFBAPA(S_H) := \set{w \mid \exists v\in \semanticsQFBAPA(X_H) \text{ and } (v,w) \in \setedges}.$

    We check that $\semanticsQFBAPA$ satisfies  $tr(\phi)$. 

    The substitution $\semanticsQFBAPA$ satisfies Formula \ref{form:reduc1}. For the first conjunct, suppose $H \not = H'$. By \Cref{def:hintikkaset}.\ref{hintikka:unique}, there is $\expression' = k_1 \in H \cap \overline{H'}$ and 
    $\expression' = k_2 \in \overline{H} \cap H'$ with $k_1 \not =_\setnumbers k_2$. By semantics of \QFBAPA, $\semanticsQFBAPA(X_H \cap X_{H'}) = \semanticsQFBAPA(X_H) \cap \semanticsQFBAPA(X_{H'}) = \{w \mid G, w \models H \text{ and } G, w \models H'\} \subseteq \{w \mid G, w \mid \expression' = k_1 \text{ and } \expression' = k_2\} = \emptyset$. 
    For the second conjunct, let $u \in \setvertices$. By semantics of \quantlogic, we have that for every $\expression' \in E(\phi)$, there is $k_{\expression'}$, 
    such that $G,u \models \expression' = k_{\expression'}$. By \Cref{def:hintikkaset}, $H = \{\expression' = k_{\expression' = k'} \mid \expression' \in E(\phi)\}$ is a Hintikka set.

    Formula \ref{form:reduc2} is satisfied by $\semanticsQFBAPA$ by definition of $\semanticsQFBAPA$ and of Hintikka sets.

    Now, $\semanticsQFBAPA$ also satisfies Formula \ref{eq:qfbapaaggregation}. Indeed, if $\agreggationfunction(\subexpression)=k \in H$, then if there is no $H$-vertex in $G$ then the implication is true. Otherwise, consider the $H$-vertex $v$. But, then by definition of $\semanticsQFBAPA(X_{\agreggationfunction(\subexpression)=k})$, we have $\sem{\agreggationfunction(\subexpression)}{G,v} = k$. But then the semantics of $\agreggationfunction$ exactly corresponds to 
    $\sum_{k' \in \setnumbers} |S_H \cap X_{\expression=k'}| \times k' = k $.
    Indeed, each $S_H\cap X_{\expression=k'}$-successor contributes with $k'$. Thus, the contribution of successors where $\expression$ is $k'$ is $|S_H \cap X_{\expression=k'}| \times k'$.

Formula \ref{eq:qfbapaaggregationglob} is also satisfied by $\semanticsQFBAPA$. Actually, let $k$ such that $\semanticsQFBAPA$ satisfies $X_{\agreggationfunctionglobalreadout(\expression)= k} = \universe$. This means that the value of $\agreggationfunctionglobalreadout(\expression)$ (which does not depend on a specific vertex $u$ but only on $G$) is $k$.  The sum $\sum_{k' \in \setnumbers} |X_{\expression=k'}| \times k' = k $ is the semantics of  $\agreggationfunctionglobalreadout(\expression) = k$.

Finally, as $G, u \models \phi$, and $\phi$ is of the form $\expression\geq k$, there is $k' \geq_\setnumbers k$ such that $\sem{\expression}{G, u} = k'$. So $\semanticsQFBAPA(X_{\expression = k'}) \neq \emptyset$.

\fbox{$\Leftarrow$}
Conversely, consider a substitution $\semanticsQFBAPA$ that satisfies $tr(\phi)$. 
We construct a graph $G = (\setvertices, \setedges, \labeling)$ as follows.
\begin{align*}
    \setvertices & := \semanticsQFBAPA(\universe) \\
    \setedges & := \set{(u, v) \mid \text{for some $H$, $u \in \semanticsQFBAPA(X_H)$ and $v \in \semanticsQFBAPA(S_H)$}} \\
    \labeling(v)_i & := k \text{ where $v \in \semanticsQFBAPA(X_{x_i = k})$}
\end{align*}
i.e., the set of vertices is the universe, and we add an edge between any $H$-vertex $u$ and a vertex $v \in \semanticsQFBAPA(S_H)$, and the labeling for features is directly given by $X_{x_i = k}$. Note that the labeling is well-defined because of \Cref{lemma:tr-entailed-equations}.\ref{hintikka:unique}.

The substitution $\semanticsQFBAPA$ satisfies $tr(\phi)$, and therefore the formula $\bigvee_{k' \geq_\setnumbers k} X_{\expression = k'} \not = \emptyset$.
So there is $k' \geq_\setnumbers k$ and $u \in \universe$ such that $u \in \semanticsQFBAPA(X_{\expression = k'})$. Let us prove that $G, u \models \phi$.

By structural induction on the syntactic complexity of $\expression$ we show that if $u \in \semanticsQFBAPA(X_{\expression = k})$ then $\sem{\expression}{G,u} = k$.

We first prove the base cases of constants and variables.
\begin{itemize}
    \item $\expression = c$. Suppose $u \in \semanticsQFBAPA(X_{c=k})$.
    From \Cref{lemma:tr-entailed-equations}.\ref{eq:qfbapa-constant}, $u \in \semanticsQFBAPA(X_{c = c})$.
    From \Cref{lemma:tr-entailed-equations}.\ref{eq:qfbapa-partition}, $c =_\setnumbers k$.
    By \quantlogic semantics, indeed $\sem{c}{G,u} = c$. 
    So $\sem{\expression}{G,u} = k$.

    \item $\expression = x_i$. Suppose $u \in \semanticsQFBAPA(X_{x_i = k})$.
    By \quantlogic semantics, $\sem{x_i}{G,u} = \labeling(u)_i$.
    By definition of $\labeling$, $\labeling(u)_i = k$. 
    So $\sem{x_i}{G,u} = k$.
\end{itemize}

Now we prove the cases of complex expressions.
Suppose for induction that 
for all $v \in \setvertices$, and for $\theta^* \in \{\theta_1, \theta_1, \theta'\}$, if $v \in X_{\theta^* = k}$ then $\sem{\theta^*}{G,v} = k$.
\begin{itemize}
    \item $\expression = \expression_1 + \expression_2$. Suppose $u \in \semanticsQFBAPA(X_{\expression_1 + \expression_2 = k})$.
    From \Cref{lemma:tr-entailed-equations}.\ref{eq:qfbapa-partition}, there are $k_1$ and $k_2$ such that $u \in \semanticsQFBAPA(X_{\expression_1 = k_1})$ and $u \in \semanticsQFBAPA(X_{\expression_2 = k_2})$.
    From \Cref{lemma:tr-entailed-equations}.\ref{eq:qfbapa-sum}, $u \in \semanticsQFBAPA(\expression_1 + \expression_2 = k_1 +_\setnumbers k_2)$.
    From \Cref{lemma:tr-entailed-equations}.\ref{eq:qfbapa-partition} (and premise), $k_1 +_\setnumbers k_2 = k$.
    By induction, $\sem{\expression_1}{G,u} = k_1$ and $\sem{\expression_2}{G,u} = k_2$.
    By \quantlogic semantics, $\sem{\expression_1 + \expression_2}{G,u} = \sem{\expression_1}{G,u} +_\setnumbers \sem{\expression_2}{G,u} = k_1 +_\setnumbers k_2 = k$.
    
    \item $\expression = c \times \expression'$. Suppose $u \in \semanticsQFBAPA(X_{c \times \expression' = k})$.
    From \Cref{lemma:tr-entailed-equations}.\ref{eq:qfbapa-partition}, there is $k'$ such that $u \in \semanticsQFBAPA(\expression' = k')$.
    From \Cref{lemma:tr-entailed-equations}.\ref{eq:qfbapa-prod}, $u \in \semanticsQFBAPA(c \times \expression' = c \times k')$.
    From \Cref{lemma:tr-entailed-equations}.\ref{eq:qfbapa-partition} (and premise), $c \times_\setnumbers k' = k$.
    By induction, $\sem{\expression'}{G,u} = k'$.
    By \quantlogic semantics, $\sem{c \times \expression'}{G,u} = c \times_\setnumbers \sem{\expression'}{G,u} = c \times_\setnumbers k' = k$.
    
    \item $\expression = \agreggationfunction(\expression')$. Suppose $u \in \semanticsQFBAPA(X_{\agreggationfunction(\expression') = k})$.
    From \Cref{lemma:tr-entailed-equations}.\ref{eq:qfbapa-partition}, for every $v \in \universe$, there is $k_v$ such that $v \in \semanticsQFBAPA(X_{\expression' = k_v})$.
    By induction, for every $v \in \universe = V$, $\sem{\expression'}{G,v} = k_v$.

    By the semantics of \quantlogic, $\sem{\agreggationfunction(\theta')}{G,u} = \sum_{v \in \setvertices \mid u\setedges v} \sem{\theta'}{G,v} = \sum_{v \in \setvertices \mid u\setedges v} k_v = \sum_{k' \in \setnumbers} |\{v \mid u\setedges v, \sem{\theta'}{G,v} = k'\}| \times_\setnumbers k'$.

    By definition of $\setedges$, and \Cref{form:reduc1}
    $(u,v) \in \setedges$ iff $v \in \semanticsQFBAPA(S_H)$.
    So $\sem{\agreggationfunction(\theta')}{G,u} =
    \sum_{k' \in \setnumbers} |\semanticsQFBAPA(S_H \cap X_{\expression = k'})| \times_\setnumbers k' = 
    \semanticsQFBAPA(\sum_{k' \in \setnumbers} |S_H \cap X_{\expression = k'}| \times k')$. 
    Therefore, by \Cref{eq:qfbapaaggregation}, we conclude that
    $\sem{\agreggationfunction(\theta')}{G,u} = k$.

    \item $\expression = \agreggationfunctionglobalreadout(\expression')$. Suppose $u \in \semanticsQFBAPA(X_{\agreggationfunctionglobalreadout(\expression') = k})$.
    From \Cref{lemma:tr-entailed-equations}.\ref{eq:qfbapa-partition}, for every $v \in \universe$, there is $k_v$ such that $v \in \semanticsQFBAPA(X_{\expression' = k_v})$.
    By \Cref{eq:qfbapaaggregationglob} (and premise), $\semanticsQFBAPA(\sum_{k' \in \setnumbers} |X_{\expression=k'}| \times k') = k$.

    By induction, for every $v \in \universe = V$, $\sem{\expression'}{G,v} = k_v$.
    
    We have successively,
    $k = \sum_{k' \in \setnumbers} |\semanticsQFBAPA(X_{\expression=k'})| \times k' =
    \sum_{k' \in \setnumbers} |\{v \mid \sem{\expression'}{G,v} = k'\}| \times k') =
    \sum_{v \in V} \sem{\expression'}{G,v} = \sem{\agreggationfunctionglobalreadout(\expression')}{G,u}.$
    
    \item $\expression = \activationfunction(\expression')$. Suppose $u \in \semanticsQFBAPA(X_{\activationfunction(\expression') = k})$. From \Cref{lemma:tr-entailed-equations}.\ref{eq:qfbapa-partition}, there is $k'$ such that $u \in \semanticsQFBAPA(X_{\expression' = k'})$. From \Cref{lemma:tr-entailed-equations}.\ref{eq:qfbapa-activation}, $u \in \semanticsQFBAPA(X_{\activationfunction(\expression') = \sem{\activationfunction}{}(k')})$.
    From \Cref{lemma:tr-entailed-equations}.\ref{eq:qfbapa-partition}, $\sem{\activationfunction}{}(k') = k$. 
    By \quantlogic semantics $\sem{\activationfunction(\expression')}{G,u} = \sem{\activationfunction}{}(\sem{\expression'}{G,u})$.
    By induction, $\sem{\expression'}{G,u} = k'$. Therefore, 
    $\sem{\activationfunction(\expression')}{G,u} = \sem{\activationfunction}{}(k') = k$.
\end{itemize}

To conclude, there is a $u \in V$, and a $k' \geq_\setnumbers k$ such that $G,u \models \expression = k'$. So $G,u \models \expression \geq k$, and therefore $\phi$ is satisfiable.
\end{proof}

\theoremlowerbound*
\begin{proof}
\newcommand{\bity}[1]{y_{#1}}
We reduce the NEXPTIME-hard problem of deciding whether a domino system $\mathcal{D} = (D,V,H)$, given an initial condition $w_0 \ldots w_{n-1} \in D^n$, can tile an exponential torus~\cite{DBLP:journals/jair/Tobies00}.
In the domino system, $D$ is the set of tile types, and $V$ and $H$ are the vertical and horizontal color compatibility relations, respectively.
We are going to write a set of modal \quantlogic formulas that characterize the torus $\mathbb{Z}^{2n+1} \times \mathbb{Z}^{2n+1}$ and the domino system.
We use $2n + 2$ features.
We use $x_0, \ldots x_{n-1}$, and $\bity0, \ldots , \bity{n-1}$, 
to hold the (binary-encoded) coordinates of vertices $\nodecoordinate x y$ in the torus.
We use the feature $x_N$ to denote a vertex $\nodeNlabel$ `on the way north' (when $x_N = 1$) and $x_E$ to denote a vertex $\nodeElabel$ `on the way east' (when $x_E = 1$), with abbreviations $\phi_N := x_N = 1$, and $\phi_E := x_E = 1$.
See Figure~\ref{fig:grid-hardness-quantlogic}.

For every $n \in \mathbb N$, we define the following set of formulas. 
\begin{equation*}
\begin{array}{llll}
T_n = \{ 
&\Box_g(x_N = 1 \lor x_N = 0),\\
&\Box_g(x_E = 1 \lor x_E = 0), \\
&\Box_g(\bigwedge_{k = 0}^{n-1} (x_i = 1 \lor x_i = 0)), \\
& \Box_g(\bigwedge_{k = 0}^{n-1} (\bity i = 1 \lor \bity i = 0)), \\ 
& \Box_g (\lnot (x_N = 1 \land x_E =1)), \\
& \Box_g (\lnot(\phi_N \lor \phi_E) \limp \agreggationfunction(1) = 2) , \\
& \Box_g (\lnot(\phi_N \lor \phi_E) \limp (\agreggationfunction(x_N) = 1)), \\
& \Box_g (\lnot(\phi_N \lor \phi_E) \limp (\agreggationfunction(x_E) = 1)), \\
& \Box_g( \phi_N \limp agg(1) = 1), \\
& \Box_g( \phi_E = 1 \limp agg(1) = 1), \\
& \Diamond_g^{=1} \phi_{(0,0)}, \\
& \Diamond_g^{=1} \phi_{(2^n-1, 2^n-1)}, \\
& \Box_g (\lnot(\phi_N \lor \phi_E) \limp \phi_{east}), \\
& \Box_g (\lnot(\phi_N \lor \phi_E) \limp \phi_{north}), \\
& \Diamond_g^{\leq  2^n \times 2^n} \lnot (\phi_N \lor \phi_E),\\
& \Diamond_g^{\leq  2^n \times 2^n} \phi_N,\\
& \Diamond_g^{\leq  2^n \times 2^n} \phi_E
\quad \}
\end{array}
\end{equation*}
where 
$\phi_{(0,0)} := \bigwedge_{k=0}^{n-1} x_i = 0 \land \bigwedge_{k=0}^{n-1} \bity i = 0$, 
and $\phi_{(2^n-1,2^n-1)} := \bigwedge_{k=0}^{n-1} x_i = 1 \land \bigwedge_{k=0}^{n-1} \bity i = 1$ represent two nodes, namely those at coordinates $(0,0)$ and $(2^n-1,2^n-1)$.
The formulas $\phi_{north}$ and $\phi_{east}$ enforce constraints on the coordinates of states, such that going north increases the coordinate encoding using the $x_i$ features by one, leaving the $\bity i$ features unchanged, and going east increases coordinate encoding using the $\bity i$ features by one, leaving the $x_i$ features unchanged. For every formula $\phi$, $\forall east. \phi$ stands for $\Box (\phi_E \limp \Box \phi)$ and 
$\forall north. \phi$ stands for $\Box (\phi_N \limp \Box \phi)$.
\begin{equation*}
\resizebox{\linewidth}{!}{$
\begin{aligned}
\phi_{north} :=
& \bigwedge_{k=0}^{n-1}(\bigwedge_{j=0}^{k-1} (x_j = 1)) \limp \\ & \left(((x_k = 1) \limp \forall north. (x_k = 0)) \land ((x_k = 0) \limp \forall north. (x_k = 1))\right) \land \\
& \bigwedge_{k=0}^{n-1}(\bigvee_{j=0}^{k-1} (x_j = 0)) \limp \\ & (((x_k = 1) \limp \forall north. (x_k = 1)) \land ((x_k = 0) \limp \forall north. (x_k = 0))) \land \\
& \bigwedge_{k=0}^{n-1}(((\bity k = 1) \limp \forall north. (\bity k = 1)) \land ((\bity k = 0) \limp \forall north. (\bity k = 0)))\\
\phi_{east} := & \bigwedge_{k=0}^{n-1}(\bigwedge_{j=0}^{k-1} (\bity j = 1)) \limp \\ & (((\bity k = 1) \limp \forall east. (\bity k = 0)) \land ((\bity k = 0) \limp \forall east. (\bity k = 1))) \land \\
& \bigwedge_{k=0}^{n-1}(\bigvee_{j=0}^{k-1} (\bity j = 0)) \limp \\ & (((\bity k = 1) \limp \forall east. (\bity k = 1)) \land ((\bity k = 0) \limp \forall east. (\bity k = 0))) \land \\
& \bigwedge_{k=0}^{n-1}(((x_k = 1) \limp \forall east. (x_k = 1)) \land ((x_k = 0) \limp \forall east. (x_k = 0)))
\end{aligned}
$}
\end{equation*}

The problem of deciding whether a domino system $\mathcal{D} = (D,V,H)$, given an initial condition $w_0 \ldots w_{n-1} \in D^n$, can tile a torus of exponential size can be reduced to the problem satisfiability in \quantlogic, checking the satisfiability of the set of formulas $T(n,\mathcal{D},w) = T_n \cup T_\mathcal{D} \cup T_w$, where $T_n$ is as above, $T_\mathcal{D}$ encodes the domino system, and $T_w$ encodes the initial condition as follows. We define 
\begin{equation*}
\scalebox{0.9}{$
\begin{aligned}
\begin{array}{ll}
   T_\mathcal{D}  = \{ &
  \Box_g(\bigwedge_{d \in D} (x_d = 1 \lor x_d = 0)), \\
  &  \Box_g( \lnot (\phi_N \lor \phi_E) \limp (\bigvee_{d \in D} \phi_d) ) ,  \\
  &  \Box_g( \lnot (\phi_N \lor \phi_E) \limp (\bigwedge_{d \in D} \bigwedge_{d' \in D \setminus \{d\}} \lnot (\phi_d \land \phi_{d'}))), \\
  & \Box_g(\bigwedge_{d \in D} (\phi_d \limp (\forall east. \bigvee_{(d,d') \in H} \phi_{d'}))),\\
  & \Box_g(\bigwedge_{d \in D} (\phi_d \limp (\forall north. \bigvee_{(d,d') \in V} \phi_{d'})))
  \quad \}
\end{array}
\end{aligned}$}
\end{equation*}
where for every $d \in D$, there is a feature $x_d$ and $\phi_d:= x_d = 1$. Finally, we define
\[
 T_w = \{
\Box_g (\phi_{(0,0)} \limp \phi_{w_0}), 
\ldots ,
\Box_g(\phi_{(n-1,0)} \limp \phi_{w_{n-1}}) 
\}
\]
The size of $T(n,\mathcal{D}, w)$ is polynomial in the size of the tiling problem instance, that is in $|D| + |H| + |V| + n$.
The rest of the proof is analogous to the proof of
\cite[Corollary~3.9]{DBLP:journals/jair/Tobies00}.
The NEXPTIME-hardness of \quantlogic follows from Lemma~\ref{lem:trFE-corr} and
\cite[Corollary~3.3]{DBLP:journals/jair/Tobies00} stating the NEXPTIME-hardness of deciding whether a domino system with initial condition can tile a torus of exponential size.

For the complexity of ACR-GNN verification tasks, we observe the following.
\begin{enumerate}
    \item We reduce the satisfiability problem in (modal) \quantlogic 
    to \verificationtask{3} in poly-time as follows. Let $\phi$ be a \quantlogic. 
    We build in poly-time an ACR-GNN $\aGNN$ that recognizes all pointed graphs. We have $\phi$ is satisfiable iff $\sem{\phi}{} \cap \sem{\aGNN}{} \neq \emptyset$. So \verificationtask{3} is NEXPTIME-hard.

    \item The validity problem of \quantlogic (dual problem of the satisfiability problem, i.e., given a formula $\phi$, is $\phi$ true in all pointed graphs $G, u$?) is coNEXPTIME-hard.
    We reduce the validity problem of \quantlogic to  \verificationtask{2}.
    Let $\phi$ be a \quantlogic formula.
    We construct an ACR-GNN $\aGNN$ that accepts all pointed graphs.
    We have $\phi$ is valid iff $\sem{\aGNN}{} \subseteq \sem{\phi}{}$. So \verificationtask{2} is coNEXPTIME-hard.


    \item  We reduce the validity problem of \quantlogic to  \verificationtask{1}.
    Let $\phi$ be a \quantlogic formula.
\verificationtaskbounded{1} can be verified with
$\sem{\top}{} \subseteq \sem{\aGNN_\phi}{}$.
    We know from \Cref{th:formulatognn} that we can construct 
    $\aGNN_\phi$ in poly-time such that $\sem{\aGNN_\phi}{} = \sem{\phi}{}$.
    Then, we have $\phi$ is valid iff $\sem{\top}{} \subseteq \sem{\aGNN_\phi}{}$. So \verificationtask{1} is coNEXPTIME-hard.
\end{enumerate}

We observe that alternative reductions can be obtained exploiting \Cref{th:formulatognn}. \verificationtask{3} can be verified with $\sem{\top}{} \cap \sem{\aGNN_\phi}{} \neq \emptyset$. \verificationtask{2} can be verified with
$\sem{\aGNN_\phi}{} \subseteq \sem{\top}{}$. 

\end{proof}


\theoremsatboundednpcomplete*
\begin{proof}
    NP upper bound is obtained by guessing a graph with at most $N$ vertices and then check that~$\phi$ holds. The obtained algorithm is non-deterministic, runs in poly-time and decides the satisfiability problem with bounded number of vertices.
    NP-hardness already holds for $\agreggationfunction$-free formulas by reduction from SAT for propositional logic (the reduction is $mod2expr$, see~Lemma~\ref{lem:trFE-corr}).

    For the complexity of the bounded ACR-GNN verification tasks, we observe the following.
\begin{enumerate}
    \item NP upper bound is also obtained by guessing a graph with at most $N$ vertices and then check. For the lower bound, we reduce (propositional) SAT to \verificationtaskbounded{3} in poly-time as follows. Let $\phi$ be a propositional formula. 
    We build in poly-time an ACR-GNN $\aGNN$ that recognizes all pointed graphs. We have $\phi$ is satisfiable iff $\sem{\phi}{} \cap \sem{\aGNN}{} \neq \emptyset$ So \verificationtaskbounded{3} is NP-hard.

    \item coNP upper bound corresponds to a NP upper bound for the dual problem: guessing a graph with at most $N$ vertices which is recognizes by $\aGNN$ but in which $\phi$ does not hold. The validity problem of propositional logic (dual problem of the satisfiability problem, i.e., given a formula $\phi$, is $\phi$ true for all valuations) is coNP-hard.
    We reduce the validity problem of propositional logic to  \verificationtaskbounded{2}.
    Let $\phi$ be a propositional formula.
    We construct an ACR-GNN $\aGNN$ that accepts all pointed graphs.
    We have $\phi$ is valid iff $\sem{\aGNN}{} \subseteq \sem{\phi}{}$. So \verificationtaskbounded{2} is coNP-hard.

    \item  coNP upper bound is obtained similarly. For the lower bound, we reduce the validity problem of propositional logic to  \verificationtaskbounded{1}.
    Let $\phi$ be a propositional formula.
  We can construct in poly-time an ACR-GNN $\aGNN$ that is equivalent to $\phi$ (by \cite{DBLP:conf/iclr/BarceloKM0RS20}).
    We have $\phi$ is valid iff $\sem{\top}{} \subseteq \sem{\aGNN}{}$. So \verificationtaskbounded{1} is coNP-hard.
    
    
\end{enumerate}

\end{proof}

\section{
The Logic
\logicKsharpglobone and ACR-GNNs Over $\mathbb{Z}$}
\label{app:ksharpone}

A \emph{(propositionally labeled) graph} $G$ is a tuple $(\setvertices, \setedges, \labeling)$ such that $\setvertices$ is 
a finite set of vertices, $\setedges \subseteq \setvertices \times \setvertices$ a set of directed edges and $\labeling$ 
is a mapping from~$\setvertices$ to a valuation over a set of atomic propositions. We write  $\labeling(u)(p) = 1$ when 
atomic proposition $p$ is true in $u$, and $\labeling(u)(p) = 0$ otherwise. 
Given a graph $G$ and vertex $u \in \setvertices$, we call $(G,u)$ a \emph{pointed graph}.

\subsection{Logic}

Consider a countable set $\Ap$ of propositions. We define the language of logic $\logicKsharpglobone$ as the set of formulas generated by the following BNF:
\begin{align*}
	\phi & ::= p \mid \lnot \phi \mid \phi \lor \phi \mid \NTexpression \geq 0 \\ 
	\NTexpression & ::= c \mid \istrue\phi \mid \modalitynumber \phi \mid \globmodalitynumber \phi \mid \NTexpression + \NTexpression \mid c\times \NTexpression 
\end{align*}
where $p$ ranges over $\Ap$, and $c$ ranges over $\mathbb Z$. We assume that all formulas $\varphi$ are represented as directed acyclic graph (DAG) and refer by \emph{the size of $\varphi$} to the size of its DAG representation.

 Atomic formulas are propositions $p$, inequalities and equalities of linear expressions. 
 We consider linear expressions over $\istrue\phi$ and $\modalitynumber \phi$ and $\globmodalitynumber \phi$. The number $\istrue\phi$ is equal to $1$ if $\phi$ holds in the current world and equal $0$ otherwise.
 The number $\modalitynumber \phi$ is the number of successors in which $\phi$ hold. The number $\globmodalitynumber \phi$ is the number of worlds in the model in which $\phi$ hold. 
 The language seems strict but we write $\NTexpression_1 \leq \NTexpression_2$ for $\NTexpression_2 - \NTexpression_1 \geq 0$, $\NTexpression = 0$ for $(\NTexpression \geq 0) \land (-\NTexpression \geq 0)$, etc.

\newcommand{\semanticsvalue}[2]{[[#1]]_{#2}}

As in modal logic, a formula $\phi$ is evaluated in a pointed graph $(G, u)$ (also known as pointed Kripke model). 
We define the truth conditions $(G,u) \models \phi$ ($\phi$ is true in $u$) by 
	\begin{center}
		\begin{tabular}{lll}
			$(G,u) \models p$ & if & $\labeling(u)(p) = 1$, \\
			$(G,u) \models \neg \phi$ & if & it is not the case that $(G,u) \models \phi$, \\
			$(G,u) \models \phi \land \psi$ & if & $(G,u) \models \phi$ and $(G,u) \models \psi$, \\
			$(G,u) \models \NTexpression \geq 0$ & if &  $\semanticsvalue{\NTexpression}{G,u} \geq 0$, \\
		\end{tabular}
	\end{center}
    and the semantics $\semanticsvalue{\NTexpression}{G,u}$ (the value of $\NTexpression$ in $u$) of an expression $\NTexpression$ by mutual induction on $\phi$ and $\NTexpression$ as follows.
	\begin{center}
		$\begin{array}{ll}
			\semanticsvalue{c}{G, u} & = c, \\
			\semanticsvalue{\NTexpression_1+\NTexpression_2}{G, u} & = \semanticsvalue{\NTexpression_1}{G,u}+\semanticsvalue{\NTexpression_2}{G,u}, \\
			\semanticsvalue{c \times \NTexpression}{G, u} & = c \times \semanticsvalue{\NTexpression}{G,u}, \\
			\semanticsvalue{\istrue\phi}{G, u} & = \begin{cases}
				1 & \text{if $(G,u) \models \phi$} \\
				0 & \text{otherwise},
			\end{cases} \\  
			\semanticsvalue{\modalitynumber\phi}{G, u} & = |\{v \in \setvertices \mid (u,v) \in \setedges \text{ and } (G,v) \models \phi\}| \\
                \semanticsvalue{\globmodalitynumber \phi}{G, u} & = |\{v \in \setvertices \mid (G,v) \models \phi\}|.
		\end{array}$
	\end{center}

A local modality $\localmodality\phi$ can be defined as $\localmodality\phi :=  (-1) \times \modalitynumber (\lnot \phi) \geq 0$. That is, to say that $\phi$ holds in all successors, we say that the number of successors in which $\lnot \phi$ holds is zero.
Similarly, a global/universal modality can be defined as $\globalmodality\phi :=  (-1) \times \globmodalitynumber (\lnot \phi) \geq 0$.

\subsection{Aggregate-Combine Graph Neural Networks}
\label{sec:gnn}

We adapt the definition of ACR-GNNs to take 
propositionally labeled graphs as input.

We define quantized Aggregate-Combine GNNs with Readout (ACR-GNN) \cite{DBLP:conf/iclr/BarceloKM0RS20}, also called message passing neural networks \cite{DBLP:conf/icml/GilmerSRVD17}.

A \emph{ACR-GNN layer} $\layer = (\COMB, \AGG, \GAGG)$ is a tuple where $\COMB : \setnumbers^{3m} \to \setnumbers^{m'}$ is a so-called \emph{combination function}, $\AGG$ is a so-called \emph{local aggregation function}, mapping multisets of vectors from $\setnumbers^m$ to a single vector from $\setnumbers^{m}$, $\GAGG$ is a so-called \emph{global aggregation function}, also mapping multisets of vectors from $\setnumbers^m$ to a single vector from $\setnumbers^m$. We call $m$ the \emph{input dimension} of layer $\layer$ and $m'$ the \emph{output dimension} of layer $\layer$.
Then, a \emph{ACR-GNN} is a tuple $(\layer^{(1)}, \dotsc, \layer^{(\nblayers)}, \CLS)$ where $\layer^{(1)}, \dotsc, \layer^{(\nblayers)}$ are 
$\nblayers$ ACR-GNN layers and $\CLS : \setnumbers^m \rightarrow \set{0, 1}$ is a \emph{classification function}. We assume that all GNNs are well-formed in the sense that
output dimension of layer $\layer^{(i)}$ matches input dimension of layer $\layer^{(i+1)}$ as well
as output dimension of $\layer^{(L)}$ matches input dimension of $\CLS$.

Let $G = (\setvertices, \setedges, \labeling)$ be a graph 
with atomic propositions $p_1, \dotsc, p_k$ 
and $\aGNN = (\layer^{(1)}, \dotsc, \layer^{(\nblayers)}, \CLS)$ an ACR-GNN. We define 
$x_0 : \setvertices \rightarrow \{0,1\}^k$, called the \emph{initial state of $G$}, as $x_0(u) := (\labeling(u)(p_1), \dots, \labeling(u)(p_k))$ for all $u \in V$.
Then, the $i$-th layer of $\aGNN$ computes an updated state of 
$G$ by
\[
\begin{aligned}
\statetv{i}u :=\;&
\COMB\Bigl(
  \statetv{i-1}u,\;
  \AGG\bigl(\multiset{\statetv{i-1}v \mid (u,v) \in \setedges}\bigr), \\
&\qquad
  \GAGG\bigl(\multiset{\statetv{i-1}v \mid v \in \setvertices}\bigr)
\Bigr)
\end{aligned}
\]
\noindent
where $\AGG$, $\GAGG$, and $\COMB$ are respectively the local aggregation, global aggregation and combination function of the $i$-th layer. Let $(G,u)$ be a pointed graph. We write $\aGNN(G,u)$ to denote the application of $\aGNN$ to $(G,u)$, which is formally defined as $\aGNN(G,u) = \CLS(x_L(u))$ where $x_L$ is the state of $G$ computed by $\aGNN$ after layer $L$. Informally, this corresponds to a binary classification of vertex $u$.

We consider the following form of ACR-GNN $\aGNN$: all local and global aggregation functions are given by the sum of all vectors in the input multiset, all combination functions are given by $\COMB(x,y,z) = \sigmabold(xC+yA_1+zA_2 +b)$ where $\sigmabold(x)$ is the componentwise application of the activation function $\activationfunction(x)$ with matrices $C$, $A_1$ and $A_2$ and vector $b$ of $\setnumbers$ parameters, and where the classification function is $\CLS(x) = \sum_i a_i x_i \geq 1$, where $a_i$ are from $\setnumbers$ as well.

We note $\sem{\aGNN}{}$ the set of pointed graphs $(G,u)$ such that 
$\aGNN(G,u) = 1$.
An ACR-GNN  $\aGNN$ is satisfiable if $\sem{\aGNN}{}$ is non-empty.
The \emph{satisfiability problem} for ACR-GNNs is:
    Given a ACR-GNN $\aGNN$, decide whether $\aGNN$ is satisfiable.

\section{Capturing GNNs with \logicKsharpglobone}
\label{app:capturinggnns}

In this section, we exclusively consider ACR-GNNs where $\setnumbers = \setZ$ and $\activationfunction$ is \emph{truncated ReLU} $\activationfunction(x) = max(0, min(1, x))$.

We demonstrate that the expressive power of ACR-GNNs over $\setZ$, with truncated ReLU activation functions, and $\logicKsharpglobone$, is equivalent. Informally, this means that for every formula $\varphi$ of $\logicKsharpglobone$, there exists an ACR-GNNs $\aGNN$ that expresses the same query, and vice-versa. To achieve this, we define a translation of one into the other and substantiate that this translation is efficient. This enables ways to employ $\logicKsharpglobone$ for reasoning about ACR-GNN.

We begin by showing that ACR-GNNs are at least as expressive as $\logicKsharpglobone$. We remark that the arguments are similar to the proof of Theorem 1 in \cite{NunnSST24}.
\begin{theorem}\label{th:ksharpglobone-to-gnn}
    Let $\varphi \in \logicKsharpglobone$ be a 
    formula. There is an ACR-GNN $\aGNN_\varphi$ such that for all pointed graphs $(G,u)$ we have $(G,u) \models \varphi$ if and only if $\aGNN_\varphi(G,u) = 1$. Furthermore, $\aGNN_\varphi$ can be built in polynomial time regarding the size of $\varphi$.
\end{theorem}
\begin{proof}[Proof sketch.]
    We construct a GNN $\aGNN_\phi$ that evaluates the semantics of a given $\logicKsharpglobone$ formula $\phi$ for some given pointed graph $(G,v)$. The network consists of $n$ layers, one for each of the $n$ subformulas $\phi_i$ of $\phi$, ordered so that the subformulas are evaluated based on subformula inclusion. The first layer evaluates atomic propositions, and each subsequent messages passing layer $l_i$ uses a fixed combination and fixed aggregation function to evaluate the semantics of $\phi_i$. 

    The correctness follows by induction on the layers: the $i$-th layer correctly evaluates $\phi_i$ at each vertex of $G$, assuming all its subformulas are correctly evaluated in previous layers. Finally, the classifying function $\CLS$ checks whether the $n$-th dimension of the vector after layer $l_n$, corresponding to the semantics of $\varphi_n$ for the respective vertex $v$, indicates that $\varphi_n = \phi$ is satisfied by $(G,v)$. The network size is polynomial in the size of $\phi$ due to the fact that the total number of layers and their width is polynomially bounded by the number of subformulas of $\varphi$. A full formal proof is given in Appendix~\ref{app:proof-th:ksharpglobone-to-gnn}.
\end{proof}

\begin{theorem}\label{th:gnn-to-ksharpglobone}
    Let $\aGNN$ be a GNN. We can compute in polynomial time wrt.\ $|\aGNN|$ a  $\logicKsharpglobone$-formula $\varphi_\aGNN$, represented as a DAG, such that $\semanticsof{\aGNN} = \semanticsof{\varphi_\aGNN}$.
\end{theorem}
\begin{proof}[Proof sketch.]
    We construct a $\logicKsharpglobone$-formula $\varphi_\aGNN$ that simulates the computation of a given GNN $\aGNN$. For each layer $l_i$ of the GNN, we define a set of formulas $\varphi_{i,j}$, one per output dimension, that encode the corresponding node features using linear threshold expressions over the formulas from the previous layer. At the base, the input features are the atomic propositions $p_1, \dotsc, p_{m_1}$.
    
    Each formula $\varphi_{i,j}$ mirrors the computation of the GNN layer, including combination, local aggregation, and global aggregation. The final classification formula $\varphi_\aGNN$ encodes the output of the linear classifier on the top layer features.
    Correctness follows from the fact that all intermediate node features remain Boolean under message passing layers with integer parameters and truncated ReLU activations. This allows expressing each output as a Boolean formula over the input propositions. The construction is efficient: by reusing shared subformulas via a DAG representation, the total size remains polynomial in the size of $\aGNN$. A more complete proof is given in Appendix~\ref{app:proof-th:gnn-to-ksharpglobone}.
\end{proof}

\subsection{Proof of Theorem~\ref{th:ksharpglobone-to-gnn}}
\label{app:proof-th:ksharpglobone-to-gnn}
\begin{proof}[Proof of Theorem~\ref{th:ksharpglobone-to-gnn}]

    Let $\phi$ be a $\logicKsharpglobone$ formula over the set of atomic propositions $p_1, \dotsc, p_m$. Let $\phi_1, \dotsc, \phi_n$ denote an enumeration of the subformulas of $\phi$ such that $\phi_i = p_i$ for $i \leq m$, $\phi_n = \phi$, and whenever $\phi_i$ is a subformula of $\phi_j$, it holds that $i \leq j$. Without loss of generality, we assume that all subformulas of the form $\NTexpression \geq 0$ are written as
    \[
    \sum_{j \in J} k_j \cdot \istrue{\phi_j} + \sum_{j' \in J'} k_{j'} \cdot \modalitynumber{\phi_{j'}}  + 
    \sum_{j'' \in J''} k_{j''} \cdot \globmodalitynumber{\phi_{j''}} - c \geq 0,
    \]
    for some index sets $J, J', J'' \subseteq \{1, \dotsc, n\}$.

    We construct the GNN $\aGNN_\phi$ in a layered manner. Note that $\aGNN_\phi$ is fully specified by defining the combination function $\COMB_i$, including its local and global aggregation, for each layer $l_i$ with $i \in \{1, \dotsc, n\}$ and the final classification function $\CLS$.
    Each $\COMB_i$ produces output vectors of dimension $n$. The first layer has input dimension $m$, and $\COMB_1$ is defined by $\COMB_1(x, y, z) = (x, 0, \dotsc, 0)$, ensuring that the first $m$ dimensions correspond to the truth values of the atomic propositions $p_1, \dotsc, p_m$, while the remaining entries are initialized to zero. Note that 
    $\COMB_1$ is easily realized by an FNN with trReLU activations.
    For $i > 1$, the combination function $\COMB_i$ is defined as
    \[
    \COMB_i(x, y, z) = \sigmabold(xC + yA_1 + zA_2 + b),
    \]
    where $C$, $A_1$, $A_2$ are $n \times n$ matrices corresponding to self, local (neighbor), and global aggregation respectively, and $b \in \mathbb{Z}^n$ is a bias vector. The parameters are defined sparsely as follows:

    \begin{itemize}
      \item $C_{ii} = 1$ for all $i \leq m$ (preserving the atomic propositions),
      \item If $\phi_i = \neg \phi_j$, then $C_{ji} = -1$ and $b_i = 1$,
      \item If $\phi_i = \phi_j \lor \phi_l$, then $C_{ji} = C_{li} = 1$, and 
      \item If $\phi_i = \sum_{j \in J} k_j \cdot 1_{\phi_j} + \sum_{j' \in J'} k_{j'} \cdot \modalitynumber{\phi_{j'}} 
      + \sum_{j'' \in J''} k_{j''} \cdot \globmodalitynumber{\phi_{j''}} - c \geq 0$, then
      \[
      C_{ji} = k_j, \quad A_{1, j'i} = k_{j'}, \quad A_{2, j'i} = k_{j''}, \quad b_i = -c + 1.
      \]
    \end{itemize}

    Note that each $\COMB_i$ has the same functional form, differing only in the non-zero entries of its parameters.
    The classification function is defined by $\CLS(x) = x_n \geq 1$.
    
    Let $l_i$ denote the $i$th layer of $\aGNN_\phi$, and fix a vertex $v$ in some input graph. We show, by induction on $i$, that the following invariant holds: for all $j \leq i$, $(x_i(v))_j = 1$ if and only if $v \models \phi_j$, and $(x_i(v))_j = 0$ otherwise.
    Assume that $i=1$. By construction, $x_1(v)$ contains the truth values of the atomic propositions $p_1, \dotsc, p_m$ in its first $m$ coordinates. Thus, the statement holds at layer $1$. 
    Next, assume the statement holds for layer $x_{i-1}$. Let $j < i$. By assumption, the semantics of $\phi_j$ are already correctly encoded in $x_{j-1}$ and preserved by $\COMB_i$ due to the fixed structure of $C$, $A_1$, $A_2$, and $b$.
    Now consider $j = i$. The semantics of all subformulas of $\phi_i$ are captured in $x_{i-1}$, either at the current vertex or its neighbors. By the design of $\COMB_i$, which depends only on the values of relevant subformulas, we conclude that $\phi_i$ is correctly evaluated. This holds regardless of whether $\phi_i$ is a negation, disjunction, or numeric threshold formula. 
    Thus, the statement holds for all $i$, and in particular for $x_n(v)$ and $\phi_n = \phi$. Finally, the classifier $\CLS$ evaluates whether $x_n(v)_n \geq 1$, which is equivalent to $G,v \models \phi$. The size claim is obvious given 
    that $n$ depends polynomial on the size of $\varphi$. We note that this assumes that the enumeration of subformulas of 
    $\varphi$ does not contain duplicates.
\end{proof}

\subsection{Proof of Theorem~\ref{th:gnn-to-ksharpglobone}}
\label{app:proof-th:gnn-to-ksharpglobone}
\begin{proof}[Proof of Theorem~\ref{th:gnn-to-ksharpglobone}]
Let $\aGNN$ be a GNN composed of layers $l_1, \dotsc, l_k$, where each $\COMB_i$ has input dimension 
$2m_i$, output dimension $n_i$, and parameters $C_i$, $A_{i,1}$, $A_{i,2}$, and $b_i$. The final
classification is defined via a linear threshold function 
$\CLS(x) = a_1 x_1 + \dotsb + a_{n_k} x_{n_k} \geq 1$. We assume that the dimensionalities 
match across layers, i.e., $m_i = n_{i-1}$ for all $i \geq 2$, so that the GNN is well-formed.

We construct a formula $\varphi_\aGNN$ over the input propositions $p_1, \dotsc, p_{m_1}$ inductively, mirroring the structure of the GNN computation.

We begin with the first layer $l_1$. For each $j \in \{1, \dotsc, n_1\}$, we define $\varphi_{1,j} \geq 1 $ or more detailed:
\begin{equation*}
    \sum_{k=1}^{m_1} (C_1)_{kj} \cdot \istrue{p_k} + (A_{1,1})_{kj} \cdot \modalitynumber{p_k} + (A_{1,2})_{kj} \cdot \globmodalitynumber{p_k} + (b_1)_j \geq 1.
\end{equation*}
Now suppose that we have already constructed formulas $\varphi_{i-1,1}, \dotsc, \varphi_{i-1,n_{i-1}}$ 
for some layer $i \geq 2$. Then, for each output index $j \in \{1, \dotsc, n_i\}$, we define $\varphi_{i,j} \geq 1$ or more detailed:
\begin{equation*}
 \begin{aligned}
 \sum_{k=1}^{m_i} (C_i)_{kj} \cdot \istrue{\varphi_{i-1,k}} + (A_{i,1})_{kj} \cdot \modalitynumber{\varphi_{i-1,k}} +\\+ (A_{i,2})_{kj} \cdot \globmodalitynumber{\varphi_{i-1,k}} + (b_i)_j \geq 1.
 \end{aligned}
\end{equation*}

Once all layers have been encoded in this way, we define the final classification formula as
\begin{displaymath}
\varphi_\aGNN = a_1 \istrue{\varphi_{k,1}} + \dotsb + a_{n_k} \istrue{\varphi_{k,n_k}} \geq 1.
\end{displaymath}

Let $G, v$ be a pointed graph. The correctness of our translation follows directly from the following observations: all weights and biases in $\aGNN$ are integers, and the input vectors $x_0(u)$ assigned to nodes $u$ in $G$ are Boolean. Moreover, each layer applies a linear transformation followed by a pointwise truncated ReLU, which preserves the Boolean nature of the node features. It follows that the intermediate representations $x_i(v)$ remain in $\{0,1\}^{n_i}$ for all $i$. Consequently, each such feature vector can be expressed via a set of Boolean $\logicKsharpglobone$-formulas as constructed above.
Taken together, this ensures that the overall formula $\varphi_\aGNN$ faithfully simulates the GNN's computation.

It remains to argue that this construction can be carried out efficiently. Throughout, we represent the (sub)formulas using a shared DAG structure, avoiding duplication of equivalent subterms. This ensures that subformulas $\varphi_{i-1,k}$ can be reused without recomputation. For each layer, constructing all $\varphi_{i,j}$ requires at most $n_i \cdot m_i$ steps, plus the same order of additional operations to account for global aggregation terms. Since the number of layers, dimensions, and parameters are bounded by $|\aGNN|$, and each operation can be performed in constant or linear time, the total construction is polynomial in the size of $\aGNN$.
\end{proof}

\section{Description Logics With Cardinality Constraints}

\subsection{\ALCQ and $T_C$Boxes Consistency}
\label{sec:ALCQ}

\ALCQ is the Description Logic adding qualified number restrictions to the standard Description Logic \ALC, analogously to how Graded Modal Logic extends standard Modal Logic with graded modalities.

Let $N_C$ and $N_R$ be two non-intersecting sets of concept names, and role names respecively.
A concept name $A \in N_C$ is an \ALCQ concept expressions of $\ALCQ$. If $C$ is an \ALCQ concept expression, so is $\lnot C$. If $C_1$ and $C_2$ are \ALCQ concept expressions, then so is $C_1 \sqcap C_2$. If $C$ is an \ALCQ concept expression, $R \in N_R$, and $n \in \mathbb{N}$, then $\geq~n~R. C$ is an \ALCQ concept expression.

A \emph{cardinality restriction} of \ALCQ is 
is an expression of the form $(\geq~n~C)$ or $(\leq~n~C)$, where $C$ an \ALCQ concept expression and $n \in \mathbb{N}$.

An \ALCQ-$T_C$Box is a finite set of cardinality restrictions.

An \emph{interpretation} is a pair $I = (\Delta^I, \cdot^I)$, where $\Delta^I$ is a non-empty set of individuals, and $\cdot^I$ is a function such that: every $A \in N_C$ is mapped to $A^I \subseteq \Delta^I$, and every $R \in N_R$ is mapped to $R^I \subseteq \Delta^I \times \Delta^I$. Given an element of $d \in \Delta^I$, we define $R^I(d) = \{ d' \mid (d,d') \in R^I  \}$. An interpretation $I$ is extended to complex concept descriptions as follows: $(\lnot C)^I = \Delta^I \setminus C^I$; $(C_1 \sqcap C_2)^I = C_1^I \cap C_2^I$; and $(\geq~n~R. C)^I = \{d \mid |R^I(d) \cap C^I| \geq n\}$.

An interpretation $I$ satisfies the cardinality restriction $(\geq~n~C)$ iff $|C^I| \geq n$ and it satisfies the cardinality restriction $(\leq~n~C)$ iff $|C^I| \leq n$.
A $T_C$Box $TC$ is \emph{consistent} if there exists an interpretation that satisfies all the cardinality restrictions in $TC$.

\begin{figure}
    \centering
\usetikzlibrary{graphs,positioning}

\begin{tikzpicture}[
xscale=2, yscale=2, font=\tiny
]

\tikzset{nodeNE/.style={draw, inner sep = 1}}

\def\maxcoord{2.5}

\node (00) at (0, 0) {$(0,0)$};
\node (01) at (0, 1) {$(0,1)$};
\node (10) at (1, 0) {$(1,0)$};
\node (11) at (1, 1) {$(1,1)$};
\node (nn) at (\maxcoord, \maxcoord) {$(2^n-1, 2^n-1)$};
\node (0n) at (0, \maxcoord) {$(0, 2^n-1)$};
\node (n0) at (\maxcoord, 0) {$(2^n-1, 0)$};

\node[nodeNE] (N1) at (0, 0.5) {$N$};
\node[nodeNE] (E1) at (0.5, 0) {$E$};
\node[nodeNE] (N2) at (1, 0.5) {$N$};
\node[nodeNE] (E2) at (0.5, 1) {$E$};

\node[nodeNE] (Nn) at (0, \maxcoord + 0.5) {$N$};
\node[nodeNE] (Nn2) at (\maxcoord, \maxcoord + 0.5) {$N$};
\node[nodeNE] (En) at (\maxcoord + 0.5, 0) {$E$};
\node[nodeNE] (En2) at (\maxcoord + 0.5, \maxcoord) {$E$};

\draw[->] (00) edge (N1);
\draw[->] (00) edge (E1);
\draw [->](01) edge (E2);
\draw[->] (10) edge (N2);

\draw[->] (N1) edge (01);
\draw[->] (E1) edge (10);
\draw[->] (N2) edge (11);
\draw[->] (E2) edge (11);

\draw[dotted] (01) edge (0n);
\draw[dotted] (10) edge (n0);
\draw[dotted] (0n) edge (nn);
\draw[dotted] (n0) edge (nn);

\draw[dotted] (11) edge (1, \maxcoord);
\draw[dotted] (11) edge (\maxcoord, 1);

\draw[->] (0n) edge (Nn);
\draw[->] (nn) edge (Nn2);
\draw[->] (nn) edge (En2);
\draw[->] (n0) edge (En);

\draw[->] (Nn) edge [bend right] (00);
\draw[->] (En) edge [bend left] (00);
\draw[->] (Nn2) edge [bend right = 40] (n0);
\draw[->] (En2) edge [bend left] (0n);

\end{tikzpicture}

    \caption{Encoding a torus of exponential size with an \ALCQ-$T_C$Box with one role.}
    \label{fig:grid-hardness-ALCQ-TC}
\end{figure}

\begin{theorem}[\protect\cite{DBLP:journals/jair/Tobies00}]\label{th:complexity-ALCQ-TCBox-hardness}
Deciding the consistency of \ALCQ-$T_C$Boxes is NEXPTIME-hard.
\end{theorem}
The proof can be slightly adapted to show that the result holds even when there is only one role.

Some abbreviations are useful. For every pair of concepts $C$ and $D$, $C \limp D$ stands for $\lnot C \sqcup D$.
For every concept $C$, role $R$, and non-negative integer $n$, we define:
$(\leq n~R.C) := \lnot (\geq (n+1)~R.C)$,
$(\forall~ R.C) := (\leq 0~R.\lnot C)$, 
$(\forall~C) := (\leq 0~\lnot C)$,
$(= n~R.C) := (\geq n~R. C) \sqcap (\leq n~R . C)$, and
$(= n~C) := (\geq n~ C) \sqcap (\leq n~ C)$.
\begin{theorem}\label{th:complexity-ALCQ-TCBox-hardness-one-role}
Deciding the consistency of \ALCQ-$T_C$Boxes is NEXPTIME-hard even if $|N_R| = 1$.
\end{theorem}
\begin{proof}
Let $next$ be the unique role in $N_R$. We use the atomic concepts $N$ to denote an individual `on the way north' and $E$ to denote an individual `on the way east'.
See Figure~\ref{fig:grid-hardness-ALCQ-TC}.

For every $n \in \mathbb N$, we define
the following \ALCQ-$T_C$Box.
\[
\begin{array}{llll}
T_n =  \{ 
& (\forall~ \lnot(N \sqcup E) \limp (= 1~ next. N)), \\
& (\forall~ \lnot(N \sqcup E) \limp (= 1~ next. E))\\
& (\forall~ N \limp (=1~ next. \top)),\\ 
& (\forall~ E \limp (=1 ~next. \top)) \\
& (= 1~ C_{(0,0)}),\\
& (= 1~ C_{(2^n-1, 2^n-1)})\\
& (\forall~ \lnot(N \sqcup E) \limp D_{east}),\\ 
& (\forall~ \lnot(N \sqcup E) \limp D_{north})\\
& (\leq~ (2^n \times 2^n)~ \lnot (N \sqcup E)), \\
& (\leq~ (2^n\times 2^n)~ N),\\
& (\leq~ (2^n \times 2^n)~ E)
\quad \}
\end{array}
\]
such that the concepts
$C_{(0,0)}$, $C_{(2^n-1,2^n-1)}$ are defined like in \cite[Figure~3]{DBLP:journals/jair/Tobies00}, and so are the concepts $D_{north}$ and $D_{east}$, except that for every concept $C$,
$\forall east. C$ now stands for $\forall next. (E \limp \forall next. C)$ and 
$\forall north. C$ now stands for $\forall next. (N \limp \forall next. C)$.


The problem of deciding whether a domino system $\mathcal{D} = (D,V,H)$, given an initial condition $w_0 \ldots w_{n-1}$, can tile a torus of exponential size can be reduced to the problem of consistency of \ALCQ-$T_C$Boxes, checking the consistency of $T(n,\mathcal{D},w) = T_n \cup T_\mathcal{D} \cup T_w$, where $T_n$ is as above, $T_\mathcal{D}$ encodes the domino system, and $T_w$ encodes the initial condition as follows.

\begin{equation*}
\scalebox{0.9}{$
\begin{aligned}
T_{\mathcal D} = \{ \;&
(\forall\, \lnot (N \sqcup E) \limp (\bigsqcup_{d \in D} C_d)), \\
&(\forall\, \lnot (N \sqcup E) \limp
  (\bigsqcap_{d \in D}
   \bigsqcap_{d' \in D \setminus \{d\}}
   \lnot (C_d \sqcap C_{d'}))), \\
&(\forall\, \bigsqcap_{d \in D}
   (C_d \limp (\forall\, east.\,
    \bigsqcup_{(d,d') \in H} C_{d'}))), \\
&(\forall\, \bigsqcap_{d \in D}
   (C_d \limp (\forall\, north.\,
    \bigsqcup_{(d,d') \in V} C_{d'})))
\;\}
\end{aligned}
$}
\end{equation*}

\[
 T_w = \{ 
(\forall ~ C_{(0,0)} \limp C_{w_0}),
\ldots ,
(\forall ~ C_{(n-1,0)} \limp C_{w_{n-1}}) 
 \}
\]
The rest of the proof remains unchanged. 
\end{proof}

\subsection{\QFBAPA}
\label{sec:qfbapa}
We assume that we have a set of \emph{set variables} and a set of \emph{integer constants}.

A \QFBAPA \cite{kuncak-rinard-QFBAPA} \emph{formula} is a Boolean combination ($\land$, $\lor$, $\lnot$) of \emph{set constraints} and \emph{cardinality constraints}.

A \emph{set term} is a Boolean combination ($\cup$, $\cap$, $\overline{\phantom{\,}\cdot\phantom{\,}}$) of \emph{set variables}, and \emph{set constants} $\universe$, and $\emptyset$.
If $S$ is a set term, then its cardinality $|S|$ is an \emph{arithmetic expressions}. Integer constants are also arithmetic expressions. If $T_1$ and $T_2$ are arithmetic expressions, so is $T_1 + T_2$. If $T$ is an arithmetic expression and $c$ is an integer constant, then $c \cdot T$ is an arithmetic expression. 

Given two set terms $B_1$ and $B_2$, the expressions $B_1 \subseteq B_2$ and $B_1 = B_2$ are \emph{set constraints}.
Given two arithmetic expressions $T_1$ and $T_2$, the expressions $T_1 < T_2$ and $T_1 = T_2$ are \emph{cardinality constraints}.
Given an integer constant $c$ and an arithmetic expression $T$, the expression $c \ dvd \ T$ is a \emph{cardinality constraint}.

A \emph{substitution} $\semanticsQFBAPA$ assigns
$\emptyset$ to the set constant $\emptyset$,
a finite set $\semanticsQFBAPA(\universe)$ to the set constant $\universe$, and a subset of $\semanticsQFBAPA(\universe)$ to every set variable.
A substitution is first extended to set terms by applying the standard set-theoretic semantics of the Boolean operations.
It is further extended to map arithmetic expressions to integers, in such that way that 
every integer constant $c$ is mapped to $c$,
for every set term $B$, the arithmetic expression $|B|$ is mapped to the cardinality of the set $\semanticsQFBAPA(B)$, and the standard semantics for addition and multiplication is applied.

The substitution $\semanticsQFBAPA$ \emph{(\QFBAPA) satisfies} the set constraint $B_1 \subseteq B_2$ if $\semanticsQFBAPA(B_1) \subseteq \sigma(B_2)$, the set constraint $B_1 = B_2$ if $\semanticsQFBAPA(B_1) = \semanticsQFBAPA(B_2)$, the cardinality constraint $T_1 < T_2$ if $\semanticsQFBAPA(T_1) < \semanticsQFBAPA(T_2)$, the cardinality constraint $T_1 = T_2$  if $\semanticsQFBAPA(T_1) = \semanticsQFBAPA(T_2)$,
and the cardinality constraint $c \ dvd \ T$  if $c$ divides $\semanticsQFBAPA(T)$.

\subsection{\ALCSCCpp: Description Logics
With Global and Local Cardinality Constraints}
\label{sec:ALCSCCpp}

The Description Logic \ALCSCCpp~\cite{DBLP:conf/ecai/BaaderBR20} extends the basic Description Logic \ALC~\cite{baader2017introduction} with concepts that capture cardinality and set constraints expressed in the quantifier-free fragment of Boolean Algebra with Presburger Arithmetic (\QFBAPA)
~\cite{kuncak-rinard-QFBAPA}.


We can now define the syntax of \ALCSCCpp concept descriptions and their semantics.
Let $N_C$ be a set of concept names, and $N_R$ be a set of role names, such that $N_C \cap N_R = \emptyset$.
Every $A \in N_C$ is a \emph{concept description} of \ALCSCCpp. Moreover, if $C$, $C_1$, $C_2$, $\ldots$ are \emph{concept descriptions} of \ALCSCCpp, then so are: $C_1 \sqcap C_2$, $C_1 \sqcup C_2$, $\lnot C$, and $\QFBAPAsat(\QFBAPAcons)$, where $\QFBAPAcons$ is a set or \QFBAPA cardinality constraints, with elements of $N_R$ and concept descriptions $C_1$, $C_2$, $\ldots$ used in place of set variables.

A \emph{finite interpretation} is a pair $I = (\Delta^I, \cdot^I)$, where $\Delta^I$ is a finite non-empty set of individuals, and $\cdot^I$ is a function such that: every $A \in N_C$ is mapped to $A^I \subseteq \Delta^I$, and every $R \in N_R$ is mapped to $R^I \subseteq \Delta^I \times \Delta^I$. Given an element of $d \in \Delta^I$, we define $R^I(d) = \{ d' \mid (d,d') \in R^I  \}$.

The semantics of the language of \ALCSCCpp makes use \QFBAPA substitutions to interpret \QFBAPA constraints in terms of \ALCSCCpp finite interpretations.
Given an element $d \in \Delta^I$, we can define the substitution $\semanticsQFBAPA_d^I$ in such a way that:
$\semanticsQFBAPA_d^I(\universe) = \Delta^I$,
$\semanticsQFBAPA_d^I(\emptyset) = \emptyset$,
and $A \in N_C$ and $R \in N_R$ are considered \QFBAPA set variables and substituted as
$\semanticsQFBAPA_d^I(A) = A^I$, and
$\semanticsQFBAPA_d^I(R) = R^I(d)$.

The finite interpretation $I$ and the \QFBAPA substitutions $\semanticsQFBAPA_d^I$ are mutually extended to complex expressions such that: $\semanticsQFBAPA^I_d(C_1 \sqcap C_2) = (C_1 \sqcap C_2)^I = C_1^I \cap C_2^I$;
$\semanticsQFBAPA^I_d(C_1 \sqcup C_2) = (C_1 \sqcup C_2)^I = C_1^I \cup C_2^I$;
$\semanticsQFBAPA^I_d(\lnot C) = (\lnot C)^I = \Delta^I \setminus C^I$; and
$\semanticsQFBAPA^I_d(\QFBAPAsat(\QFBAPAcons)) = (\QFBAPAsat(\QFBAPAcons))^I = \{ d' \in \Delta^I \mid \semanticsQFBAPA^I_{d'} \text{ (\QFBAPA) satisfies } \QFBAPAcons \}$.

\begin{definition}
The \ALCSCCpp concept description $C$ is \emph{satisfiable} if there is a finite interpretation $I$ such
that $C^I \not = \emptyset$.
\end{definition}

\begin{theorem}[\protect\cite{DBLP:conf/ecai/BaaderBR20}]\label{th:complexity-ALCSCCpp}
    The problem of deciding whether an \ALCSCCpp concept description is satisfiable is NEXPTIME-complete.
\end{theorem}

\section{Complexity of the Satisfiability of \logicKsharpglobone and Its Implications for ACR-GNN Verification}
\label{appendixsection:NEXPTIMEcompletewithintegers}

In this section, we establish the complexity of reasoning with \logicKsharpglobone.

Instrumentally, we first show that every \logicKsharpglobone formula can be translated into a \logicKsharpglobone formula that is equi-satisfiable, and has a tree representation of size at most polynomial in the size of the original formula.
An analogous result was obtained in \cite{NunnSST24} for \logicKsharpone. It can be shown using a technique reminiscent of \cite{Tseitin1983} and consisting in factorizing subformulas that are reused in the DAG by introducing a fresh proposition that is made equivalent. Instead of reusing a `possibly large' subformula, a formula then reuses the equivalent `small' atomic proposition.
\begin{lemma}\label{lem:Ktree-K}
    The satisfiability problem of \logicKsharpglobone reduces to the satisfiability of \logicKsharpglobone with tree formulas in polynomial time.
\end{lemma}
\begin{proof}
Let $\phi$ be a \logicKsharpglobone formula represented as a DAG.
%
%
For every subformula $\psi$ (i.e., for every node in the DAG representation of $\phi$), we introduce a fresh atomic proposition $p_\psi$. We can capture the meaning of these new atomic propositions with the formula $\Phi := \bigwedge_{\psi \text{ node in the DAG}} sem(\psi)$ where:
\begin{align*}
    sem(\psi \lor \chi) & := p_{\psi \lor \chi} \lequiv (p_\psi \lor p_\chi) \\
    sem(\lnot \psi) & := p_{\lnot \psi} \lequiv \lnot p_\psi\\
    sem(\NTexpression \geq 0) & := p_{\NTexpression \geq 0} \lequiv \NTexpression' \geq 0
\end{align*}
\[
\begin{array}{ccc}
    (c)'  := c  &  (\NTexpression_1 + \NTexpression_2)' := \NTexpression_1' + \NTexpression_2' & (c \times \NTexpression)' := c \times \NTexpression'\\
    (\istrue \psi)'  := \istrue p_{\psi} & (\modalitynumber \psi)'  := \modalitynumber p_{\psi} & (\globmodalitynumber \psi)' := \globmodalitynumber p_{\psi}
\end{array}
\]
%
Now, define $\phi_t := p_\phi \land \globalmodality \Phi$, where $\globalmodality\Phi :=  (-1) \times \globmodalitynumber (\lnot \Phi) \geq 0$, enforcing the truth of $\Phi$ in every vertex.
The size of its tree representation is polynomial in the size of $\phi$. Moreover, $\phi_t$ is satisfiable iff $\phi$ is satisfiable.

\end{proof}


\begin{theorem}\label{th:Ktree-NEXPTIME-complete}
    The satisfiability problem of $\logicKsharpglobone$ with tree formulas is NEXPTIME-complete.
\end{theorem}

\begin{proof}
For membership, we translate the problem into the NEXPTIME-complete problem of concept description satisfiability in the Description Logics with Global and Local Cardinality Constraints~\cite{DBLP:conf/ecai/BaaderBR20}, noted \ALCSCCpp. The Description Logic \ALCSCCpp uses the Boolean Algebra with Presburger Arithmetic~\cite{kuncak-rinard-QFBAPA}, noted \QFBAPA, to formalize cardinality constraints.
See Section~\ref{sec:ALCSCCpp} for a presentation of \ALCSCCpp and \QFBAPA.

Let $\phi_0$ be a \logicKsharpglobone formula.
%


%
For every proposition $p$ occurring in $\phi_0$, let $A_{\num{p}}$ be an \ALCSCCpp concept name. Let $R$ be an \ALCSCCpp role name. For every occurrence of $\istrue \phi$ in $\phi_0$, let $ZOO_{\num{\phi}}$ be an \ALCSCCpp role name. 
$ZOO$-roles stand for `zero or one'.
The rationale for introducing $ZOO$-roles is to be able to capture the value of $\istrue \phi$ in \ALCSCCpp making it equal to the number of successors of the role $ZOO_{\num{\phi}}$ which can then be used in \QFBAPA constraints.
A similar trick was used, in another context, in \cite{GallianiKT23}.
Here, we enforce this with the \QFBAPA constraint 
\begin{equation*}
\scalebox{0.8}{$
\begin{aligned}\QFBAPAcons_{0} = \bigwedge_{\istrue \phi \in \phi_0} 
\big ( (|ZOO_{\num{\phi}}| = 0 \lor |ZOO_{\num{\phi}}| = 1) \land  \transKA(\phi) = \QFBAPAsat(|ZOO_{\num{\phi}}| = 1) \big)
\end{aligned}
$}
\end{equation*}

which states that $ZOO_{\num{\phi}}$ has zero or one successor, and has one successor exactly when (the translation of) $\phi$ is true.
The concept descriptions $\transKA(\phi)$ and arithmetic expressions $\transKA(\NTexpression)$ are defined inductively as follows:

\[
\begin{array}{lcl}
\transKA(p) & = & A_{\num{p}} \\
\transKA(\lnot \phi) & = & \lnot \transKA(\phi)\\
\transKA(\phi \lor \psi) & = & \transKA(\phi) \sqcup \transKA(\psi)\\
\transKA(\NTexpression \geq 0) & = & \QFBAPAsat(-1 < \transKA(\NTexpression))\\
\transKA(c) & = & c\\
\transKA(\NTexpression_1 + \NTexpression_2) & = & \transKA(\NTexpression_1) + \transKA(\NTexpression_2)\\
\transKA(c \times \NTexpression) & = & \transKA(c \cdot \NTexpression)\\
\transKA(\modalitynumber \phi) & = & |R \cap \transKA(\phi)|\\
\transKA(\istrue\phi) & = & |ZOO_{\num{\phi}}|\\
\transKA(\globmodalitynumber \phi) & = & |\transKA(\phi)|\\
\end{array}
\]
Finally, we define the \ALCSCCpp concept description $C_{\phi_0} = \transKA(\phi_0) \sqcap \QFBAPAsat(\QFBAPAcons_{0})$.

\begin{claim}\label{claim:trans-Ksharpglobone-ALCSCCpp}
    The concept description $C_{\phi_0}$ is \ALCSCCpp-satisfiable iff the formula $\phi_0$ is \logicKsharpglobone-satisfiable.
    Moreover, the concept description $C_{\phi_0}$ has size polynomial in the size of $\phi_0$.
\end{claim}
\begin{proof}
    From right to left, suppose that $\phi_0$ is \logicKsharpglobone-satisfiable. It means that there is a pointed graph $(G,u)$ where $G = (V,E)$ and $u \in V$, such that $(G,u) \models \phi_0$. Let $I_0 = (\Delta^{I_0}, \cdot^{I_0})$ be the \ALCSCCpp interpretation over $N_C$ and $N_R$, such that $N_C = \{A_{\num{p}} \mid p \text{ a proposition in } \phi_0\}$, $N_R = \{R \} \cup \{ZOO_{\num{\phi}} \mid \istrue \phi \in \phi_0\}$, $\Delta^{I_0} = V$, $A^{I_0}_{\num{p}} = \{v \mid v \in \setvertices, (G,v) \models p\}$ for every $p$ in $\phi_0$,
    $R^{I_0} = \setedges$,
    $ZOO^{I_0}_{\num{\phi}} = \{(v,v) \mid v \in \setvertices, (G,v) \models \phi\}$ for every $\istrue \phi$ in $\phi_0$. We can show that $u \in C_{\phi_0}^{I_0}$.
    Basically $I^0$ is like $G$ with the addition of adequately looping $ZOO$-roles. An individual in $\Delta^{I_0}$ has exactly one $ZOO_{\num{\phi}}$-successor (itself), exactly when $\phi$ is true, and no successor otherwise; $A_p$ is true exactly where $p$ is true, and the role $R$ corresponds exactly to $\setedges$.

    From left to right, suppose that $C_{\phi_0}$ is \ALCSCCpp-satisfiable. It means that there is an \ALCSCCpp finite interpretation $I_0 = (\Delta^{I_0}, \cdot^{I_0})$ and an individual $d \in \Delta^{I_0}$ such that $d \in C_{\phi_0}^{I_0}$.
    Let $G = (\setvertices, \setedges)$ be a graph such that $\setvertices = \Delta^{I_0}$, $\setedges = R^{I_0}$, and $\labeling(d)(p) = 1$ iff $d \in A_{\num{p}}^{I_0}$. We can show that $(G,d) \models \phi_0$.
    
    \medskip

    Since there are at most $|\phi_0|$ subformulas in $\phi_0$, the representation of $ZOO_{\num{\phi}}$ for every subformula $\phi$ of $\phi_0$ can be done in size $\log_2(|\phi_0|)$. For every formula $\phi$, the size of the concept description $\transKA(\phi)$ is polynomial (at most $O(n\log(n))$).
    The overall size of $\transKA(\phi_0)$ is polynomial in the size of $\phi_0$, and so is the size of $\QFBAPAsat(\NTexpression_0)$ (at most $O(n^2(\log(n))^2$).
\end{proof}
The NEXPTIME-membership follows from Claim~\ref{claim:trans-Ksharpglobone-ALCSCCpp} and the fact that the concept satisfiability problem in \ALCSCCpp is in NEXPTIME (Theorem~\ref{th:complexity-ALCSCCpp}).

\medskip

    For the hardness, we reduce the problem of consistency of \ALCQ-$T_C$Boxes
    which is NEXPTIME-hard~\cite[Corollary~3.9]{DBLP:journals/jair/Tobies00}. See Section~\ref{sec:ALCQ} and 
    Theorem~\ref{th:complexity-ALCQ-TCBox-hardness-one-role} that slightly adapts Tobies' proof to show that the problem is hard even with only one role.

We define the translation $\transTK$ from the set of \ALCQ concept expressions and \ALCQ cardinality constraints, with only one role $R$.
\[
\begin{array}{lcl}
\transTK(A) & = & p_A\\
\transTK(\lnot C) & = & \lnot \transTK(C)\\
\transTK(C_1 \sqcup C_2) & = & \transTK(C_1) \lor \transTK(C_2)\\
\transTK(\geq~n~R. C) & = & \modalitynumber \transTK(C) + (-1) \times n \geq 0\\
\transTK(\geq~n~C) & = & \globmodalitynumber \transTK(C) + (-1) \times n \geq 0\\
\transTK(\leq~n~C) & = & (-1) \times \globmodalitynumber \transTK(C) + n \geq 0\\
\end{array}
\]

It is routine to check the following claim.
\begin{claim}
Let $TC$ be an \ALCQ-$T_C$Box. $TC$ is consistent iff $\bigwedge_{\chi \in TC} \transTK(\chi)$ is \logicKsharpglobone-satisfiable.
\end{claim}    
Moreover, the reduction is linear.
Hardness thus follows from the NEXPTIME-hardness of consistency of \ALCQ-$T_C$Boxes.
\end{proof}

Lemma~\ref{lem:Ktree-K} and Theorem~\ref{th:Ktree-NEXPTIME-complete} yield the following corollary.
\begin{corollary}\label{corr:K-NEXPTIME-complete}
    The $\logicKsharpglobone$-satisfiability problem is NEXPTIME-complete.
\end{corollary}
Furthermore, from Theorem~\ref{th:ksharpglobone-to-gnn} and Corollary~\ref{corr:K-NEXPTIME-complete},
we obtain the complexity of reasoning with ACR-GNNs with truncated ReLU and integer weights.
\begin{corollary}\label{corr:ACR-GNN-NEXPTIME}
  Satisfiability of ACR-GNN with global readout, over $\setZ$ and with truncated ReLU is NEXPTIME-complete.  
\end{corollary}
The decidability of the problem is left open in \cite{DBLP:conf/icalp/BenediktLMT24} and in the recent long version \cite{benedikt2025decidabilitygraphneuralnetworks_V4} when the weights are rational numbers. \Cref{corr:ACR-GNN-NEXPTIME} answers it positively in the case of integer weights and pinpoints the computational complexity.

\section{Experimental Data and Further Analyses}
\label{sec:Experiments}
The package is publicly available online on \href{https://github.com/francoisschwarzentruber/kr2026-Verifying-Quantized-GNNs-With-Readout-Is-Decidable-But-Highly-Intractable}{GitHub repository}
\subsection{Description of the model}
Before verifying the model, we need to explore the model we are dealing with. As mentioned in Section~\ref{section:background} of the main part of the paper, we are working with the Aggregate-Combined Graph Neural Network with the global Readout, which belongs to the GNN Message-Passing family. The theoretical model was explained earlier; in this section, we address the practical implementation. It is important to note that we are working with models for node classification. In the article with the reference architecture \cite{DBLP:conf/iclr/BarceloKM0RS20}, the authors mention binary node classification.

In Fig~.\ref{fig:acrgnnscheme}, we design the model based on the trained model. This reference model was derived from the implementation of the theoretical model presented in \cite{BarceloGit2021}.

\begin{figure}[h!]
\centering
\resizebox{\columnwidth}{!}{%
\begin{tikzpicture}[
  font=\scriptsize,
  node distance=12mm,
  box/.style={draw, rounded corners, minimum height=7mm, align=center, inner sep=3pt},
  arr/.style={-Latex, thick}
]
\node[box] (eq1) {$\alpha(xC+yA+zR+b)$};
\node[box, right=of eq1] (eq2) {$\mathrm{BatchNorm}(p)$};
\node[box, right=of eq2] (eq3) {$\mathrm{LP}(p_{\mathrm{bn}})$};

\draw[arr] (eq1) -- (eq2);
\draw[arr] (eq2) -- (eq3);
\end{tikzpicture}%
}
\caption{Scheme of the ACR-GNN one-layer.}
\label{fig:acrgnnscheme}
\end{figure}

Here is the recap of the key elements of the layer with the data type FLOAT32:
\begin{itemize}
    \item $C$,$A$,$R$ are the matrices of the weights that during the training were initialized in a random way. Separately, $C$ is the weight matrix of the input feature, $A$ is the weighted matrix of the local aggregation, and $R$ is the weighted matrix of the global aggregation.  $T \in \mathbb{R}^{h \times h}$, where $ T\in \{C,A,R\}$ and $h$ is a hidden dimension. Interesting point from architecture:
    \begin{itemize}
        \item If $h > 5$ zero padding (fixed lift to higher dimension)
        \item If $h < 5$  first convolution compresses 5 to h (learnable)
        \item If $h = 5$  no padding, no compression
    \end{itemize}
    \item $b\in \mathbb{R}^{1 \times h}$ is the bias vector. In the theoretical description, only one vector was mentioned, but in the practical implementation, there are three bias vectors, one per matrix. This fact needs to be noted for future analysis and affects the analysis and verification flow.
    \item $\alpha$ is the activation function. We have no limitation on choice by design for this part of the learning process. In the reference paper, the authors used the ReLU and trReLU activations, which are the piece-wise activation functions.
    \item $p_{\mathrm{bn}}=\mathrm{BatchNorm}(p)$ this is the Batch Normalization (BN) that was applied immediately after the activation function, used in practice to make the learned representation stable, well-scaled, and usable by the next layers. In the original paper, BN was suggested to be applied before the activation function, but the architecture shows that it was actually applied after.
    \item $\mathrm{LP}(h_{\mathrm{bn}})$ how the learned embedding is converted into the model’s output. In the theory corresponding to the classification function. In practice, it corresponds to the linear\_prediction instance in the model's state dictionary. With this layer, we can obtain the classification of the input graph. So, for the verification task, including this layer can be optional. In details, this is the mapping from $\mathbb{R}^{h \times h} \xrightarrow{} \mathbb{R}^{h \times 2}$, 2 because the task for the reference work \cite{DBLP:conf/iclr/BarceloKM0RS20} is Binary Classification.
\end{itemize}

All these instances are keys in the trained model's state dictionary. For future analyses, we will use this representation for the model's verification task. We select this way because it is not always possible to have an access to the model's reference architecture, it is more practical to work with the saved model. In this article, we use PyTorch for training and save the model in `.pth' format.

After defining the model, we specify the information that can be extracted from it. Traditionally, model evaluation focuses on predictive accuracy and the behavior of the loss function during training. A more detailed analysis of the trained model is presented in the corresponding evaluation section.

Figure~\ref{fig:pipeline} illustrates the experimental and verification pipeline adopted in this work. The model is first trained on the target dataset using PyTorch. After training, architectural and training-related parameters are collected for analysis, including activation functions, number of layers, hidden dimensions, model size, and computational costs. The trained model is then serialized as a `.pth' file, from which the state\_dict is extracted and subsequently provided as input to the verification tool.

\begin{figure}[t]
\centering
\begin{tikzpicture}[
  font=\scriptsize,
  node distance=7mm,
  box/.style={draw, rounded corners, minimum width=32mm, minimum height=6mm, align=center},
  arr/.style={-Latex, thick}
]

\node[box] (data) {Dataset};
\node[box, below=of data] (train) {Model Training (PyTorch)};

\node[box, below=of train] (save) {Saved Model (.pth)};
\node[box, right=of save] (dptq) {Quantization (PyTorch)};

\node[box, below=of save] (state) {\textit{state\_dict} Extraction};
\node[box, below=of state] (verify) {Verification Tool};

\draw[arr] (data) -- (train);
\draw[arr] (train) -- (save);
\draw[arr] (save) to[bend left=20] (dptq);
\draw[arr] (dptq) to[bend left=20] (save);

\draw[arr] (save.south) -- node[left]{FP32} (state.north);
\draw[arr, dashed] ([xshift=8pt]save.south) -- node[right]{qINT8} ([xshift=8pt]state.north);

\draw[arr] (state.south) -- node[left]{FP32} (verify.north);
\draw[arr, dashed] ([xshift=8pt]state.south) -- node[right]{qINT8} ([xshift=8pt]verify.north);

\end{tikzpicture}
\caption{Learning and verification pipeline.}
\label{fig:pipeline}
\end{figure}

After outlining the experimental workflow, we structure the analysis for each dataset as follows.
(i) Description of the dataset;
(ii) Analysis of the trained models under different architectural and hyperparameter configurations;
(iii) Application of dynamic post-training quantization;
(iv) Comparative evaluation between original and quantized models;
(v) Extraction of model parameters and verification using the ESBMC tool, where the obtained matrices are translated into the verification workflow.

Here we describe each section in detail. 

(i)\textit{ Datasets}. We employ the same datasets as in \cite{DBLP:conf/iclr/BarceloKM0RS20}, namely synthetic graphs generated using the \Erdos\ model and real-world Protein–Protein Interaction (PPI) networks\cite{zitnik2017predicting}. The synthetic dataset allows controlled variation of structural properties, while the PPI dataset represents multi-graph, real-world biological networks.
Detailed information about the structural characteristics of the graph data is provided in Tables (Table~\ref{tab:dataset_summary} and Table~\ref{tab:ppi_dataset_summary}), which summarize the properties of the \Erdos\ and PPI datasets, respectively.

(ii) \textit{Model configuration}. We train a simple Aggregate–Combine Graph Neural Network (ACR-GNN). For both the aggregation and the global readout functions, we use the sum operator. The combination function is defined as $\text{COMB}(x,y,z) = \alpha(xC+yA +zR +b)$, 
where $\alpha$ denotes the component-wise application of the activation function. Following the original work, we use a batch size of 128 and train the model for 20 epochs with the Adam optimizer~\cite{adam2014method}, employing the default PyTorch hyperparameters.
During the architectural exploration and analysis of SMT-solver compatibility, we treated the hidden dimension as a key tunable parameter. This choice was motivated by its direct impact on both model capacity and the size of the resulting SMT encoding.

In our experiments, we follow the evaluation protocol described in \cite{DBLP:conf/iclr/BarceloKM0RS20}, where accuracy is defined as the total number of correctly classified nodes divided by the total number of nodes across all graphs in the dataset.
We consider a diverse set of activation functions, grouped according to their functional properties: Piecewise linear (ReLU, ReLU6, trReLU, and LeakyReLU), Smooth unbounded (GELU and SiLU), Smooth bounded (Sigmoid, Normalized), and Smooth rectifiers (Softplus and ELU). In the ACR-GNN architecture, the activation function is applied component-wise within the combination function and therefore directly affects node-level message aggregation, feature transformation, and the resulting global graph representation obtained through the readout function.

Figure~\ref{fig:all_activation_functions} illustrates the ten non-linear activation functions considered in our experiments, using their standard implementations provided by PyTorch.

\begin{figure*}[t]
    \centering
    \includegraphics[width=0.8\linewidth]{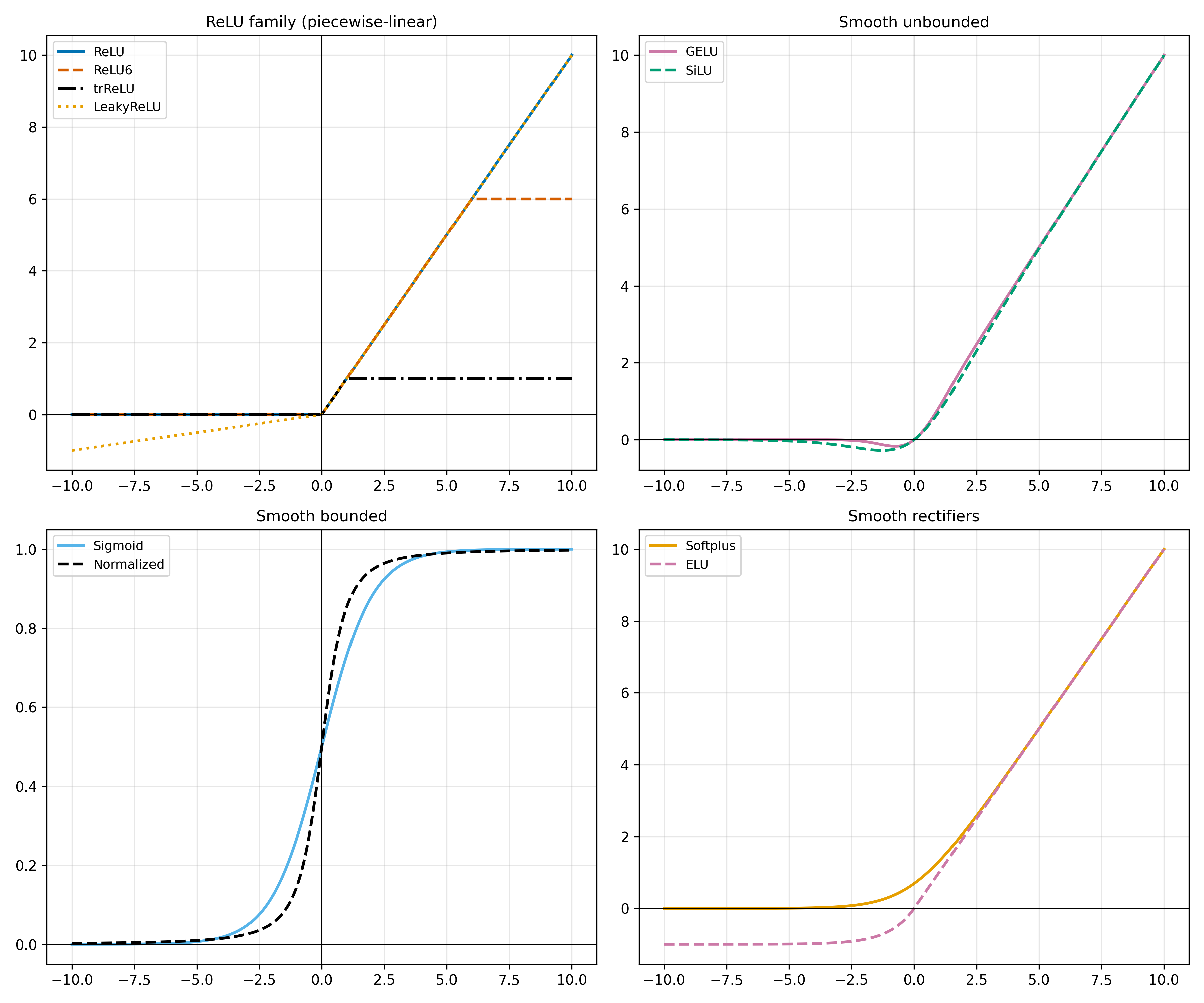}
    \caption{Non-linear activation functions whose influence were analyzed, grouped by family.}
    \label{fig:all_activation_functions}
\end{figure*}

We presented the key aspects of each activation function (A/F) in Table~\ref{tab:activation-comparison}.

\begin{table*}
\centering
\begin{tabular}{l r c l}
\toprule
Activation Function & Output Range & Smooth & Notes \\
\midrule
ReLU      
  & $[0,\infty)$     
  & No 
  & Sparse; unbounded above. \\

ReLU6     
  & $[0,6]$          
  & No 
  & Bounded ReLU; quantization-friendly. \\

trReLU    
  & $[0,1)$          
  & No  
  & Clipped ReLU; similar to HardTanh. \\

LeakyReLU ($\alpha{=}0.01$) 
  & $(-\infty,\infty)$ 
  & No 
  & Avoids dead neurons. \\

GELU      
  & $(-\infty,\infty)$ 
  & Yes 
  & Smooth ReLU variant. \\

Sigmoid   
  & $(0,1)$          
  & Yes 
  & Saturating; vanishing gradients. \\

Normalized $(0,1)$ 
  & $(0,1)$ 
  & Yes 
  & Algebraic; no exponentials. \\

SiLU      
  & $(-\infty,\infty)$ 
  & Yes 
  & Also called Swish. \\

Softplus  
  & $(0,\infty)$     
  & Yes 
  & Smooth ReLU approximation. \\

ELU       
  & $(-\alpha,\infty)$ 
  & Yes 
  & Negative outputs improve gradients. \\

\bottomrule
\end{tabular}
\caption{Mathematical properties of activation functions used in ACR-GNN experiments.}
\label{tab:activation-comparison}
\end{table*}

(iii) \textit{Quantization}. For model compression, we employ dynamic Post-Training Quantization (PTQ). According to the taxonomy presented in \cite{gholami2022survey}, quantization techniques can be broadly divided into Quantization-Aware Training (QAT) and Post-Training Quantization (PTQ). We adopt PTQ because it allows the model to be trained only once in full precision, after which quantization is applied as a separate step. This approach reduces computational cost and aligns with sustainability considerations by avoiding repeated retraining.
PTQ can be applied either statically or dynamically. In this work, we use dynamic PTQ. In dynamic quantization, model weights are quantized offline (e.g., to 8-bit integers), while activations are quantized at runtime immediately before computation. The quantization parameters for activations (e.g., scaling factors) are determined dynamically based on the observed range during inference.
We implement this approach using the PyTorch library, specifically the \textit{torch.ao.quantization.quantize\_dynamic API}. This function replaces eligible modules (e.g., \textit{nn.Linear}) with their quantized counterparts, enabling reduced model size and improved inference efficiency. Dynamic quantization is applied after training and does not require calibration data, in contrast to static PTQ.

In the above explanation, two key factors are particularly relevant: the data type used for quantization and the type of layers to which quantization is applied.
In our experiments, dynamic PTQ converts model parameters from \textit{float32} precision to \textit{qint8}. Following the architecture described in~\cite{BarceloGit2021,DBLP:conf/iclr/BarceloKM0RS20}, quantization is applied to the \textit{nn.Linear} layers, which constitute the learnable transformation components of the ACR-GNN.
In PyTorch, \textit{qint8} is a quantized tensor data type representing 8-bit signed integers. Unlike the standard \textit{int8} type, \textit{qint8} encodes quantized floating-point values together with scale and zero-point parameters. These parameters define an affine mapping between integer values and their corresponding real-valued representations:
\begin{equation}
x_{\text{real}} = \text{scale} \cdot (x_{\text{int}} - \text{zero\_point}).
\end{equation}
Internally, \textit{qint8} values are stored as two's complement signed integers, while the scale and zero-point parameters preserve the numerical range of the original floating-point weights. Consequently, quantized weights remain compatible with neural network computations while reducing memory footprint and enabling more efficient inference.

iv) \textit{Comparative analysis}. In this section, we compare the results obtained before and after quantization by applying appropriate statistical metrics to evaluate potential differences in predictive performance and model behavior.
In addition to dynamic Post-Training Quantization (PTQ), we consider the concept of \textit{fake quantization}, which simulates quantization effects during training. Fake quantization introduces rounding and clamping operations to emulate low-precision arithmetic while maintaining the data in floating-point format (e.g., \textit{FP32}). Concretely, values are quantized and immediately dequantized within the computational graph, allowing the forward pass to reflect quantization effects without permanently changing the underlying data representation.
This mechanism is particularly relevant for Quantization-Aware Training (QAT), where the model learns to compensate for quantization-induced errors during training. By anticipating discretization effects, the model can adjust its weights accordingly, typically resulting in higher post-quantization accuracy when deployed on hardware with limited numerical precision.

v) \textit{Verification Flow}. After completing the experimental pipeline, we refine the model selection procedure to account for the limitations of the verification tool.
In particular, the verification workflow imposes constraints on the admissible activation functions, as not all non-linearities can be efficiently encoded within the SMT-based framework. To systematize this restriction, we constructed Table~\ref{tab:activation-smt}, which classifies the considered activation functions according to their SMT compatibility.
These constraints directly influence the subset of models that can be subjected to formal verification.

\begin{table*}[tb]
\centering
\begin{tabular}{l r l}
\toprule
Activation Function & SMT Compatibility & Notes \\
\midrule
ReLU      
& Yes (exact) 
& Piecewise-linear encoding. \\

ReLU6     
& Yes (exact) 
& Linear constraints with upper bound. \\

trReLU    
& Yes (exact) 
& Piecewise-linear clipping. \\

LeakyReLU ($\alpha{=}0.01$) 
& Yes (exact) 
& Piecewise-linear; avoids dead neurons. \\

Normalized $(0,1)$ 
& Yes (nonlinear) 
& Algebraic encoding with auxiliary variables. \\

GELU      
& No 
& Requires transcendental functions. \\

Sigmoid   
& No (approx.) 
& Requires $\exp$; approximations needed. \\

SiLU      
& No 
& Depends on $\exp$. \\

Softplus  
& No 
& Depends on $\exp$. \\

ELU       
& No (approx.) 
& Exponential negative branch. \\

\bottomrule
\end{tabular}
\caption{SMT compatibility of activation functions in ACR-GNN verification.}
\label{tab:activation-smt}
\end{table*}

SMT solvers natively support only a restricted class of activation functions that can be encoded using conditional (if--then--else) expressions, namely piecewise-linear functions. In our verification workflow, this class includes ReLU, ReLU6, trReLU, and LeakyReLU, all of which admit a direct encoding in terms of linear constraints combined with conditional branching.
In addition, we consider the Normalized activation function. Although it is smooth and bounded, it can be encoded using algebraic (nonlinear) constraints without resorting to transcendental functions, which are typically unsupported or poorly handled by SMT solvers. This makes it suitable, albeit more computationally demanding, for inclusion in the verification setting.
We present the first implementation of this verification pipeline as a proof of concept and as a baseline for future research. For the verification task, we employ ESBMC (Efficient SMT-based Context-Bounded Model Checker)~\cite{esbmc2024}, a mature and permissively licensed open-source context-bounded model checker.
The workflow illustrated in Fig.~\ref{fig:model_to_esbmc} translates a trained PyTorch model into C code that can be analyzed by ESBMC, embedding the corresponding preconditions and postconditions for formal verification.

\begin{figure*}[!tb]
\centering
\begin{tikzpicture}[
  font=\scriptsize,
  node distance=10mm,
  box/.style={draw, rounded corners, minimum height=7mm, align=center, inner sep=3pt},
  arr/.style={-Latex, thick}
]
\node[box] (open) {Open saved model\\(\textit{.pth}, PyTorch)};
\node[box, right=of open] (sd) {Extract \\ \textit{state\_dict}};
\node[box, right=of sd] (adapt) {Adapt to\\ESBMC instances}; 
\node[box, right=of adapt] (py) {Python generator\\script};
\node[box, right=of py] (cgen) {Generate \textit{.c} file};
\node[box, right=of cgen] (esbmc) {Run verification\\(ESBMC)};

\draw[arr] (open) -- (sd);
\draw[arr] (sd) -- (adapt);
\draw[arr] (adapt) -- (py);
\draw[arr] (py) -- (cgen);
\draw[arr] (cgen) -- (esbmc);
\end{tikzpicture}
\caption{Model-to-verification workflow: converting a trained PyTorch model into ESBMC-verifiable C instances.}
\label{fig:model_to_esbmc}
\end{figure*}

For the verification-oriented analysis, we examine the model's \textit{state\_dict}, which contains the learned parameters of the network, including weight matrices, bias vectors, and batch normalization statistics. These parameters are extracted and translated into the corresponding C representation used in the verification workflow.
During the verification phase, we identified scalability limitations of ESBMC. In particular, the tool is sensitive to the number of arithmetic operations and loop constructs generated in the translated C code. As the hidden dimension increases, the number of matrix multiplications and nested iterations grows accordingly, leading to a rapid increase in the size of the SMT encoding and verification time.
To ensure tractable verification, we therefore restricted the hidden dimension to small values (e.g., from 2 to 5). The upper bound of 5 was selected to remain compatible with the number of input features in the dataset on which the models were trained, while keeping the encoding manageable for the solver.
To further enable verification in the SMT-based setting, we preprocess Batch Normalization (BN) layers by rewriting them into an equivalent affine transformation, thereby eliminating operations (e.g., square roots) that are unsupported or inefficient in SMT solvers. Originally, Batch Normalization is defined as\ref{eq:bn_original}:
\begin{equation}
BN(p_i)
= \frac{\gamma (p_i - \mu_i)}{\sqrt{\sigma_i^2 + \varepsilon}}
+ \beta_i ,
\label{eq:bn_original}
\end{equation}
where $\gamma, \beta, \mu_i, \sigma_i^2 \in \mathbb{R}$ are the learned parameters.
However, SMT solvers such as ESBMC do not directly support non-linear arithmetic involving square roots. Therefore, we algebraically rewrite~\eqref{eq:bn_original} into an affine form. Expanding~\eqref{eq:bn_original}, we obtain:
\begin{equation}
\begin{aligned}
BN(p_i)
&= \frac{\gamma}{\sqrt{\sigma_i^2 + \varepsilon}} p_i
   - \frac{\gamma}{\sqrt{\sigma_i^2 + \varepsilon}} \mu_i
   + \beta_i \\
&= a_i p_i + c_i ,
\end{aligned}
\label{eq:bn_mod}
\end{equation}
where the coefficients are defined as
\begin{equation}
a_i := \frac{\gamma}{\sqrt{\sigma_i^2 + \varepsilon}},
\qquad
c_i := \beta_i - a_i \mu_i .
\label{eq:bn_coeff}
\end{equation}
This reformulation eliminates the square root operator from the symbolic expression, as $a_i$ and $c_i$ are precomputed constants.
So for the verification we are using the $a_i$ and $c_i$ which are the two vectors with the dimension 
$\mathbb{R}^{1 \times h}$. Data for the calculation were taken from the state dictionary of the model.
In details, into the state dictionary can be found the variables for the calculations the substitution \ref{eq:bn_mod} for the equation\ref{eq:bn_original}, we listed these variables and instancies of the state dictionary for $i$th layer.
\begin{itemize}
    \item $\gamma$ - \text{batch\_norms.i.weight}, 
    \item $\beta$ - \text{batch\_norms.i.bias},
    \item $\mu$ - \text{batch\_norms.i.running\_mean},
    \item $\sigma^2$ - \text{batch\_norms.i.running\_var},
\end{itemize}
The substitution of~\eqref{eq:bn_original} by the affine form~\eqref{eq:bn_mod} ensures that the resulting verification conditions remain within linear arithmetic, which is directly supported by SMT solvers.
\subsection{Synthetic data. Description}
Graphs were generated using the dense \Erdos model, a classical method for constructing random graphs, and each graph was initialized with five node colours encoded as one-hot feature vectors.
The dataset is structured as follows, as shown in Table~\ref{tab:dataset_summary}. The training set consists of 5000 graphs, each with 40 to 50 nodes and between 560 and 700 edges. The test set is divided into two subsets. The first subset comprises 500 graphs with the same structure as the training set, featuring 40 to 50 nodes and 560 to 700 edges. The second subset contains 500 larger graphs, with 51 to 69 nodes and between 714 and 960 edges. This design allows us to evaluate the model's generalization capability to unseen graph sizes. 

\begin{table*}[t]
    \centering
    \begin{tabular}{llcccccc}
        \toprule
        & & \multicolumn{3}{c}{Node} & \multicolumn{3}{c}{Edge} \\
        \cmidrule(r){3-5} \cmidrule(r){6-8}
        Classifier & Dataset & Min & Max & Avg & Min & Max & Avg \\
        \midrule
        \multirow{3}{*}{$p_1$} 
            & Train  & 40 & 50 & 45 & 560 & 700 & 630 \\
            & Test1  & 40 & 50 & 45 & 560 & 700 & 633 \\
            & Test2  & 51 & 60 & 55 & 714 & 960 & 832 \\
        \midrule
        \multirow{3}{*}{$p_2$} 
            & Train  & 40 & 50 & 45 & 560 & 700 & 630 \\
            & Test1  & 40 & 50 & 44 & 560 & 700 & 628 \\
            & Test2  & 51 & 60 & 55 & 714 & 960 & 832 \\
        \midrule
        \multirow{3}{*}{$p_2$} 
            & Train  & 40 & 50 & 44 & 560 & 700 & 629 \\
            & Test1  & 40 & 50 & 45 & 560 & 700 & 630 \\
            & Test2  & 51 & 60 & 55 & 714 & 960 & 831 \\
        \bottomrule
    \end{tabular}
    \caption{Dataset statistics summary.}
    \label{tab:dataset_summary}
\end{table*}

We trained ACR-GNN on complex formulas $\text{FOC}_2$ for labeling. They are presented as a classifier $\alpha_i(x)$ that constructed as:
\begin{equation*}
    \alpha_0(x):= \text{Blue}(x), \alpha_{i+1}(x):=  \exists^{[N,M]}y\left(\alpha_i(y) \wedge \neg E(x,y)\right)
\end{equation*}
where $\exists^{[N,M]}$ stands for ``there exist between $N$ and $M$ nodes''. satisfying a given property.

Observe that each $\alpha_i(x)$ is in FOC$_2$, as $\exists^{[N,M]}$ can be expressed by combining $\exists^{\geq N}$ and $\neg \exists^{\geq M + 1}$.

The data set has the following specifications: \Erdos graphs and is labeled according to $\alpha_1(x)$, $\alpha_2(x)$, and $\alpha_3(x)$:
\begin{itemize}
    \item $\alpha_0(x):= \text{Blue}(x)$
    \item  $p_1: \alpha_1(x) := \exists^{[8,10]}y\left(\alpha_0(y) \wedge \neg E(x,y)\right)$
    \item  $p_2: \alpha_2(x) := \exists^{[10,30]}y\left(\alpha_1(y) \wedge \neg E(x,y)\right)$
    \item  $p_3: \alpha_3(x) := \exists^{[10,30]}y\left(\alpha_2(y) \wedge \neg E(x,y)\right)$
\end{itemize}

\subsection{Synthetic data. Analysis of model after training}
In this subsection, we analyze several model parameters: the type of activation function, the number of layers, the model size, training time, the Generalization Ratio, and the Generalization Gap.

We trained the models on the dataset and collected the training time. This data is the first step in analyzing the influence of the activation function. Based on the data obtained, we can identify the slowest and fastest activation functions. 

\begin{table*}[t!]
\centering
\begin{tabular}{l *{10}{c}}
\toprule
& ReLU & ReLU6 & trReLU & LeakyReLU & GELU & SiLU & Sigmoid & Normalized & Softplus & ELU \\
\midrule
$p_1$ & 4831.71 & 4861.24 & 4863.12 & 4860.25 & 4846.62 & 4853.88 & 4841.79 & 5037.88 & 4861.80 & 4860.28 \\
$p_2$ & 4843.09 & 4858.06 & 4860.35 & 4855.30 & 4849.57 & 4848.89 & 4851.51 & 5032.96 & 4859.86 & 4860.32 \\
$p_3$ & 4834.85 & 4856.53 & 4859.48 & 4857.82 & 4857.39 & 4853.81 & 4847.41 & 5039.38 & 4851.02 & 4855.93 \\
\bottomrule
\end{tabular}
\caption{Training time (s) per classifier and activation}
\label{tab:training-times}
\end{table*}

Table~\ref{tab:training-times} reports the total training time (in seconds) aggregated across all experimental configurations.
For each key ($p_1$, $p_2$, and $p_3$), the reported value is obtained by summing the training times over layers 1--10, where each layer was evaluated using multiple hidden dimensions and learning rates.
Thus, the table reflects the cumulative computational cost of each activation function across the full experimental grid. Across all three keys, ReLU consistently achieves the lowest total training time 
($4831.71$\,s for $p_1$, $4843.09$\,s for $p_2$, and $4834.85$\,s for $p_3$).
In contrast, the Normalized activation exhibits the highest training time 
($5037.88$\,s, $5032.96$\,s, and $5039.38$\,s, respectively), 
which corresponds to an overhead of approximately $190$--$208$ seconds (about $4\%$) compared to ReLU.
The remaining activation functions (ReLU6, trReLU, LeakyReLU, GELU, SiLU, Sigmoid, Softplus, and ELU) show very similar behavior, with differences typically below $1\%$ relative to ReLU.
This indicates that, except for the Normalized activation, the choice of activation function has only a marginal impact on overall training time.
Moreover, the relative ranking of activation functions remains stable across $p_1$, $p_2$, and $p_3$, suggesting that the observed computational trends are robust with respect to the classifier configuration.

\begin{table*}[t]
\centering
\begin{tabular}{l *{10}{r}}
\toprule
 & ReLU & ReLU6 & trReLU & LeakyReLU & GELU & SiLU & Sigmoid & Normalized & Softplus & ELU \\
\midrule
$p_1$ & 21.30 & 9.70 & 6.00 & 7.30 & 12.30 & 8.70 & 19.30 & 0.00 & 8.70 & 6.70 \\
$p_2$ & 12.70 & 7.70 & 8.70 & 8.30 & 13.30 & 15.00 & 11.70 & 0.00 & 10.70 & 12.00 \\
$p_3$ & 14.00 & 7.00 & 7.70 & 11.00 & 13.30 & 9.30 & 19.70 & 0.00 & 8.70 & 9.30 \\
\bottomrule
\end{tabular}
\caption{Percentage of configurations in which each activation function achieves the minimum training time, aggregated over layers 1--10 and all tested hidden dimensions and learning rates.}
\label{tab:best:training-times}
\end{table*}

Table~\ref{tab:best:training-times} reports the percentage of configurations in which each activation function achieves the minimum training time, aggregated over layers~1--10 and all tested hidden dimensions and learning rates.
For $p_1$, ReLU (21.3\%) and Sigmoid (19.3\%) most frequently yield the fastest training.
For $p_2$, the distribution is more balanced, with SiLU (15.0\%) and GELU (13.3\%) slightly leading.
For $p_3$, Sigmoid (19.7\%) achieves the highest win-rate, followed by ReLU (14.0\%).
Importantly, the Normalized activation never attains the minimum training time (0\% across all keys), indicating that it is consistently outperformed in terms of computational efficiency.
Overall, several activations exhibit competitive performance, suggesting that the fastest option depends on the specific classifier configuration.

\begin{table*}[t]
\centering
\begin{tabular}{l *{10}{c}}
\toprule
 & ReLU & ReLU6 & trReLU & LeakyReLU & GELU & SiLU & Sigmoid & Normalized & Softplus & ELU \\
\midrule
$p_1$ & 3.30 & 6.30 & 4.70 & 6.00 & 3.00 & 7.00 & 5.00 & 54.70 & 5.70 & 4.30 \\
$p_2$ & 3.30 & 4.30 & 7.70 & 4.00 & 4.70 & 4.70 & 5.30 & 55.70 & 5.30 & 5.00 \\
$p_3$ & 1.30 & 5.30 & 5.00 & 4.70 & 5.70 & 4.00 & 3.70 & 59.00 & 5.00 & 6.30 \\
\bottomrule
\end{tabular}
\caption{Percentage of configurations in which each activation function achieves the maximum training time, aggregated over layers 1--10 and all tested hidden dimensions and learning rates.}
\label{tab:worst:training-times}
\end{table*}

Table~\ref{tab:worst:training-times} reports the percentage of configurations in which each activation function yields the maximum training time.
Across all classifier keys, the Normalized activation overwhelmingly dominates the loss-rate (54.7\%--59.0\%), indicating that it is most frequently the slowest option.
All other activation functions exhibit relatively low and comparable loss-rates (typically below 7\%), without a clear systematic pattern.
This confirms that, while multiple activations may compete for the minimum training time, the Normalized variant consistently incurs the highest computational cost across configurations.

We evaluated 10 activation functions across 3 classifier keys.
For each classifier, we varied the hidden dimension (10 values), learning rate (3 values), and network depth (1--10 layers).
This results in a total of 9000 trained models, providing broad coverage of architectural and optimization configurations.

To ensure a fair comparison across activation functions and architectural configurations, we adopted a global hyperparameter selection strategy. The learning rate and hidden dimension were selected using a validation-based protocol, relying exclusively on the validation split (Test1) and without inspecting the final test set (Test2).

First, we evaluated all candidate learning rates across the full set of architectural configurations, including variations in activation function, number of layers, hidden dimension, and dataset key. For each learning rate, we computed the average validation accuracy across all configurations. The optimal learning rate was selected as the one maximizing the mean validation performance. To account for variability and avoid selecting unstable configurations, we additionally considered the standard deviation of validation accuracy and verified that the chosen learning rate exhibited both high mean performance and stable behavior.

After fixing the globally optimal learning rate, we repeated the same procedure to select the hidden dimension. Specifically, we aggregated validation accuracy across all activation functions, layers, and dataset keys for each candidate hidden size. The hidden dimension maximizing the average validation accuracy was selected as the global architectural parameter.

\textbf{Key finding.} The result of the selection: learning rate $lr=0.01$ with hidden dimension $h=32$.
Further analysis of accuracy before dynamic Post-Training Quantization will be done for these parameters.

\begin{figure*}[t]
    \centering
\includegraphics[width=0.7\linewidth]{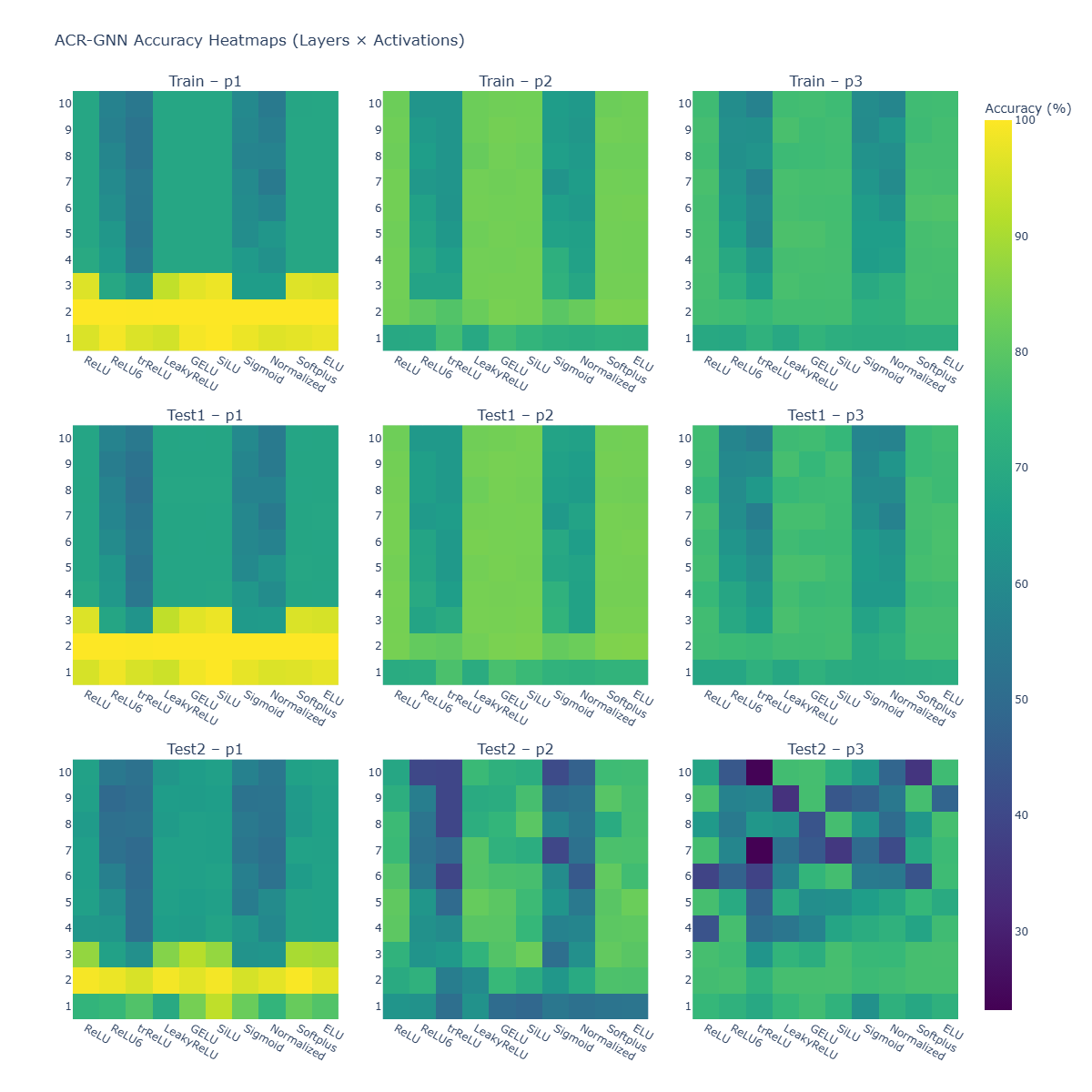}
    \caption{Heatmaps of ACR-GNN accuracy across activation functions and network depth. Each row corresponds to a metric (Train, Test1, Test2), while each column corresponds to a dataset classifier ($p_1$, $p_2$, $p_3$). Color intensity indicates classification accuracy. Hyperparameter (learning rate $lr=0.01$ with hidden dimension $h=32$)}
    \label{fig:nonqua:heatmap}
\end{figure*}
The heatmaps in Figure~\ref{fig:nonqua:heatmap} visualize how the accuracy of the ACR-GNN varies with respect to the number of layers and the choice of activation function. The figure is organized in a 3$\times$3 grid: rows represent different evaluation metrics (Train accuracy, Test 1 accuracy, and Test 2 accuracy), and columns represent the three datasets classifiers ($p_1$, $p_2$, $p_3$). Each cell encodes accuracy values as a function of the number of layers (y-axis) and activation functions (x-axis). This visualization allows for a direct comparison of performance trends, highlighting, for example, activation functions that maintain stable accuracy across increasing depth or those that degrade sharply.

Generally, the trend that is common for all models: the number of edges, has a significant influence on the final classification results, which can be an indication of how robust the model is. Generally, a robust model performs reliably across different datasets, test conditions, or noise levels (not just on the data it was trained on). A non-robust model may work very well in the training set, but its accuracy drops sharply when evaluated on slightly different or more challenging test sets. In the case of this analysis, robustness is the ability of the model to use different activation functions to keep accuracy stable on Test2 compared to training.

The vertical patterns of the heatmaps reveal a clear dependency between model performance and network depth. Across all datasets ($p_1$, $p_2$, $p_3$), shallow architectures consistently outperform deeper configurations. In particular, models with two to three layers achieve the highest validation and test accuracies, while performance degrades progressively as depth increases beyond three layers. This trend is especially pronounced for dataset p1, where deeper architectures exhibit a substantial drop in accuracy. A similar, though less severe, degradation is observed for $p_2$, while $p_3$ demonstrates the strongest instability under increased depth. This behavior is consistent with the well-known over-smoothing phenomenon in Graph Neural Networks, where repeated neighborhood aggregation leads to loss of discriminative node representations. The results therefore suggest that shallow ACR-GNN architectures provide a better trade-off between representational capacity and generalization performance. 

The horizontal structure of the heatmaps highlights systematic differences across activation functions. Smooth and non-saturating activations, such as GELU, SiLU, LeakyReLU, and ELU, consistently achieve higher and more stable accuracies across datasets and splits. In contrast, saturating nonlinearities, particularly Sigmoid, exhibit noticeable performance degradation, especially on the Test2 split. The Normalized and hard-clipped variants (e.g., truncated ReLU) also demonstrate reduced robustness under distributional shifts. These findings indicate that smooth activation functions enhance optimization stability and generalization capability in ACR-GNN models. Moreover, the observed stability across datasets suggests that such activations are less sensitive to architectural variations and dataset complexity, making them preferable candidates for subsequent quantization and verification analyses.

\textbf{Key finding.}The further analysis will be done for a number of layers from 1 to 3 if it requires some demonstrations and from 1 to 10 to have the global view.

To assess how well the models generalize beyond the training data, we report two complementary metrics: Generalization Ratio and Generalization Gap. The generalization ratio measures the relative closeness between training and test performance:
\begin{equation*}
    \text{Generalization Ratio (GR)}= \frac{\text{Test Accuracy}}{\text{Train Accuracy}}
\end{equation*}
If the ratio is close to 1, the model generalizes well (Test $\approx$ Train). If it is much less than 1, the model is overfitting (Train $\gg$ Test). 

The generalization gap quantifies the absolute drop in performance from training to test:
    \begin{equation*}
        \Delta_{gen} = \text{Train Accuracy} - \text{Test Accuracy}.
    \end{equation*}
A small gap reflects strong generalization, while larger gaps highlight overfitting. 

In our case, we compute both values separately for Test1 and Test2.  Tables across layers have the following structure as Table\ref{tab:results:nonqua:ACR-GNN:1layer}.

\begin{table*}[t]
\centering
\begin{tabular}{l ccc ccc ccc ccc}
\toprule
& \multicolumn{4}{c}{$p_1$} & \multicolumn{4}{c}{$p_2$} & \multicolumn{4}{c}{$p_3$} \\
A/F & \multicolumn{2}{c}{Test1} & \multicolumn{2}{c}{Test2} & \multicolumn{2}{c}{Test1} & \multicolumn{2}{c}{Test2} & \multicolumn{2}{c}{Test1} & \multicolumn{2}{c}{Test2} \\
 & GR & $\Delta_{gen}$ & GR & $\Delta_{gen}$ & GR & $\Delta_{gen}$ & GR & $\Delta_{gen}$ & GR & $\Delta_{gen}$ & GR & $\Delta_{gen}$ \\
\midrule
ReLU & 0.995 & 0.468 & 0.765 & 22.535 & 1.016 & -1.082 & 0.913 & 6.016 & 0.990 & 0.722 & 1.073 & -5.027 \\
ReLU6 & 0.995 & 0.527 & 0.754 & 24.375 & 1.015 & -1.015 & 0.894 & 7.375 & 0.993 & 0.478 & 1.048 & -3.339 \\
trReLU & 0.995 & 0.510 & 0.819 & 17.412 & 1.016 & -1.211 & 0.666 & 25.542 & 0.993 & 0.505 & 0.955 & 3.289 \\
LeakyReLU & 0.992 & 0.722 & 0.733 & 25.332 & 1.017 & -1.183 & 0.900 & 6.920 & 0.993 & 0.508 & 1.074 & -5.116 \\
GELU & 0.995 & 0.477 & 0.845 & 15.359 & 1.018 & -1.388 & 0.660 & 25.755 & 0.990 & 0.755 & 1.042 & -3.067 \\
SiLU & 1.000 & -0.009 & 0.927 & 7.303 & 1.022 & -1.580 & 0.664 & 24.683 & 0.989 & 0.789 & 1.020 & -1.429 \\
Sigmoid & 0.994 & 0.594 & 0.838 & 15.835 & 1.018 & -1.270 & 0.751 & 17.816 & 0.991 & 0.624 & 0.878 & 8.597 \\
Normalized & 0.995 & 0.530 & 0.762 & 22.966 & 1.016 & -1.139 & 0.729 & 19.216 & 0.994 & 0.436 & 1.020 & -1.425 \\
Softplus & 0.992 & 0.749 & 0.844 & 15.174 & 1.018 & -1.283 & 0.733 & 19.215 & 0.991 & 0.664 & 0.969 & 2.233 \\
ELU & 0.995 & 0.524 & 0.805 & 19.052 & 1.018 & -1.327 & 0.738 & 18.805 & 0.993 & 0.479 & 1.010 & -0.733 \\
\bottomrule
\end{tabular}
\caption{Generalization performance of the 1-layer ACR-GNN with different activation functions (A/F), reported as both Generalization Ratios (Test/Train) and Generalization Gaps (Train – Test accuracy) across datasets $p_1$, $p_2$, and $p_3$.}
\label{tab:results:nonqua:ACR-GNN:1layer}
\end{table*}

After analysis, we obtain the following results.
Shallow networks (1–2 layers): Softplus, ELU, SiLU dominate (with ReLU also competitive at 2 layers).
Moderate depths (3–5 layers): ELU dominates most consistently, with Softplus and SiLU frequently competitive (and ReLU strong at 5 layers).
Intermediate-deep networks (6–8 layers): ELU and SiLU dominate (GELU is occasionally competitive).
Deep networks (9–10 layers): Softplus dominates at 9 layers, while ELU dominates at 10 layers (overall: ELU/Softplus).

\subsection{Synthetic data. Applying dPTQ}

To evaluate whether the quantization time is sensitive to the choice of activation function. The dependent variable was the total quantization time, while the primary explanatory variable was the categorical factor activation function. To isolate the effect of the activation choice from architectural and hyperparameter influences, we included the following control variables: number of layers, hidden dimensions, learning rate, and model variant (classifier). This design allows the activation-specific effect to be estimated independently of model depth, width, optimization settings, and classifier type.

We applied Post-hoc Tukey HSD tests to the time of applying dPTQ and the test revealed no statistically significant differences in quantization time across activation functions (all adjusted p-values $\geq$ 0.95). This suggests that quantization cost is largely independent of activation choice and is instead primarily driven by architectural parameters such as depth and hidden dimensionality.

\begin{figure}[t]
    \centering
    \includegraphics[width=\columnwidth]{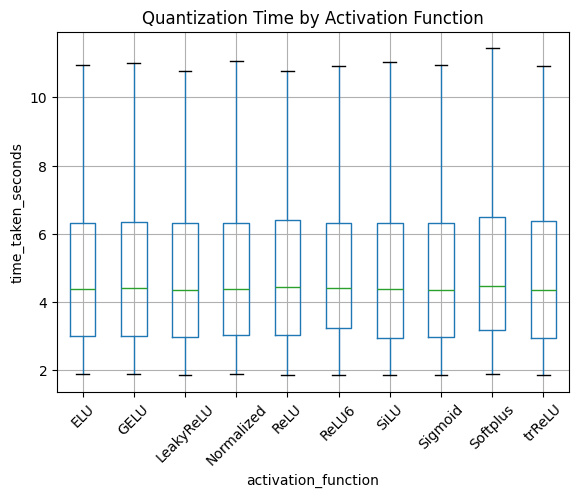}
    \caption{Time taken for quantization and evaluation of the quantized model}
    \label{fig:quant_time_by_activation}
\end{figure}

To support this conclusion visually, we present two complementary plots. The first, Quantization Time by Activation Function (Figure\ref{fig:quant_time_by_activation}), shows the empirical distribution of quantization times across activations. The second, Activation Sensitivity(Figure\ref{fig:quant_sensivity}), reports the estimated mean quantization time with confidence intervals derived from the controlled regression model. Both visualizations corroborate the statistical findings, demonstrating overlapping confidence intervals and negligible effect sizes.

\begin{figure}[t]
    \centering
    \includegraphics[width=\linewidth]{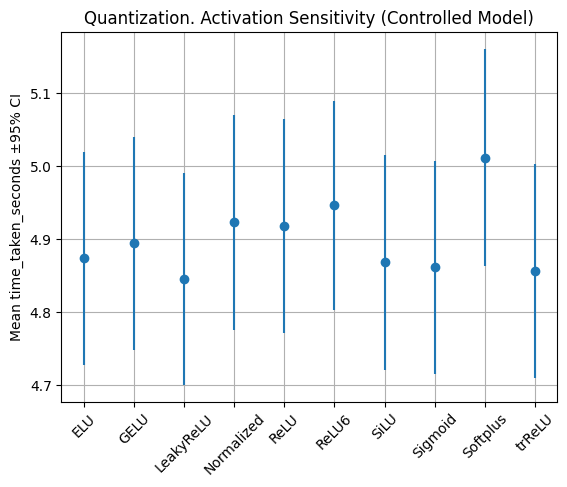}
    \caption{Activation Sensitivity (Controlled Model)}
    \label{fig:quant_sensivity}
\end{figure}

Overall, the results suggest that quantization cost is primarily driven by the model's structural characteristics (e.g., depth and hidden dimensionality) rather than by the specific nonlinearity employed. This indicates that activation function selection does not introduce measurable computational overhead during post-training quantization.

Based on the key findings from the original model, we construct the HeatMap of the Accuracy for the models after applying dynamic Post-Training Quantization.

\begin{figure*}[t]
    \centering
    \includegraphics[width=0.7\linewidth]{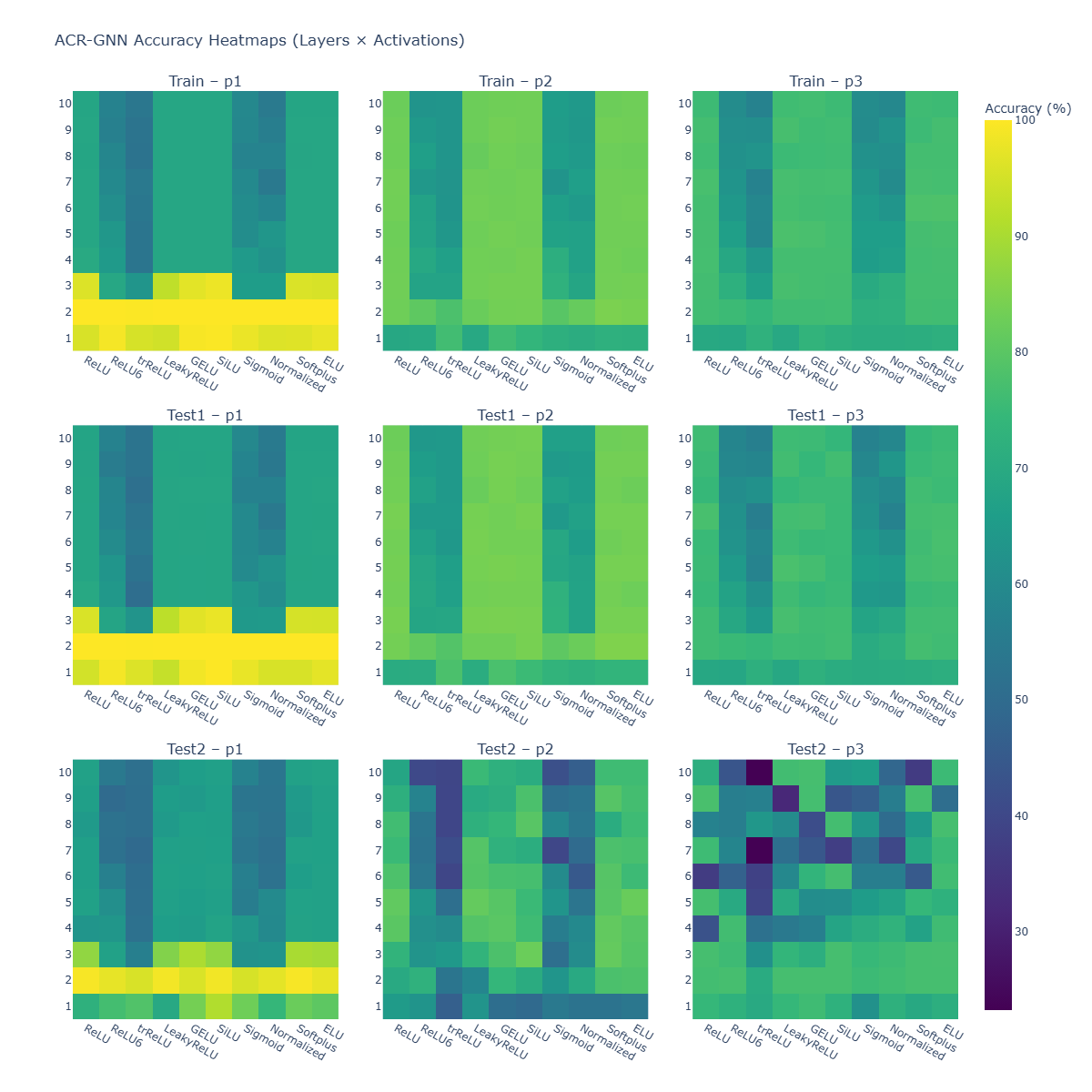}
    \caption{Heatmaps of ACR-GNN accuracy after dPTQ across activation functions and network depth. Each row corresponds to a metric (Train, Test1, Test2), while each column corresponds to a dataset classifier ($p_1$, $p_2$, $p_3$). Color intensity indicates classification accuracy. Hyperparameter (learning rate $lr=0.01$ with hidden dimension $h=32$)}
    \label{fig:dptq:heatmap_dptq}
\end{figure*}

The heatmaps in Figure~\ref{fig:dptq:heatmap_dptq} visualize how the accuracy of the ACR-GNN varies with respect to the number of layers and the choice of activation function. The figure is organized in a 3$\times$3 grid: rows represent different evaluation metrics (Train accuracy, Test 1 accuracy, and Test 2 accuracy), and columns represent the three datasets classifiers ($p_1$, $p_2$, $p_3$). Each cell encodes accuracy values as a function of the number of layers (y-axis) and activation functions (x-axis). This visualization allows for a direct comparison of performance trends, highlighting, for example, activation functions that maintain stable accuracy across increasing depth or those that degrade sharply.

Generally, the trend that is common for all models: the number of edges, has a significant influence on the final classification results, which can be an indication of how robust the model is. Generally, a robust model performs reliably across different datasets, test conditions, or noise levels (not just on the data it was trained on). A non-robust model may work very well in the training set, but its accuracy drops sharply when evaluated on slightly different or more challenging test sets. In this analysis, robustness is the model's ability to use different activation functions while maintaining Test2 accuracy comparable to training.

Table~\ref{tab:results:dptq:ACR-GNN:1layer} shows that applying dPTQ largely preserves near-distribution generalization, as indicated by GR values close to one on Test1 across all activation functions and datasets. However, the evaluation on Test2 exposes substantial activation-dependent sensitivity to quantization. For $p_1$, SiLU provides the most robust Test2 behavior (GR$=0.913$, $\Delta_{gen}=8.700$), while ReLU and LeakyReLU exhibit the largest degradation (GR$\leq 0.753$ with $\Delta_{gen}\geq 23.6$). In contrast, for $p_2$ the best Test2 performance is achieved by piecewise-linear activations, with ReLU yielding the highest GR and smallest gap (GR$=0.942$, $\Delta_{gen}=4.007$), whereas smooth activations such as GELU and SiLU degrade markedly (GR$\approx 0.67$, $\Delta_{gen}\approx 24$--$25$). For $p_3$, dPTQ generalization is largely maintained for most activations, with ELU achieving the most balanced profile (GR$=0.996$, $\Delta_{gen}=0.270$); Sigmoid remains the least stable (GR$=0.876$, $\Delta_{gen}=8.739$). Overall, these results indicate that dPTQ effects are primarily revealed under the harder Test2 setting and that the most robust activation function is dataset-dependent.

\begin{table*}[t]
\begin{tabular}{l ccc ccc ccc ccc}
\toprule
& \multicolumn{4}{c}{$p_1$} & \multicolumn{4}{c}{$p_2$} & \multicolumn{4}{c}{$p_3$} \\
A/F & \multicolumn{2}{c}{Test1} & \multicolumn{2}{c}{Test2} & \multicolumn{2}{c}{Test1} & \multicolumn{2}{c}{Test2} & \multicolumn{2}{c}{Test1} & \multicolumn{2}{c}{Test2} \\
 & GR & $\Delta_{gen}$ & GR & $\Delta_{gen}$ & GR & $\Delta_{gen}$ & GR & $\Delta_{gen}$ & GR & $\Delta_{gen}$ & GR & $\Delta_{gen}$ \\
\midrule
ReLU & 0.993 & 0.689 & 0.753 & 23.599 & 1.017 & -1.175 & 0.942 & 4.007 & 0.991 & 0.639 & 1.071 & -4.943 \\
ReLU6 & 0.999 & 0.104 & 0.774 & 22.388 & 1.013 & -0.900 & 0.904 & 6.656 & 0.995 & 0.325 & 1.046 & -3.163 \\
trReLU & 1.008 & -0.752 & 0.824 & 16.752 & 1.019 & -1.431 & 0.610 & 29.779 & 0.989 & 0.811 & 0.961 & 2.797 \\
LeakyReLU & 0.987 & 1.256 & 0.732 & 25.330 & 1.019 & -1.299 & 0.906 & 6.476 & 0.993 & 0.473 & 1.071 & -4.957 \\
GELU & 0.994 & 0.633 & 0.843 & 15.621 & 1.021 & -1.573 & 0.668 & 25.230 & 0.985 & 1.081 & 1.038 & -2.792 \\
SiLU & 1.002 & -0.211 & 0.913 & 8.700 & 1.024 & -1.730 & 0.671 & 24.208 & 0.987 & 0.925 & 1.010 & -0.688 \\
Sigmoid & 0.998 & 0.211 & 0.847 & 14.933 & 1.022 & -1.563 & 0.759 & 17.244 & 0.993 & 0.514 & 0.876 & 8.739 \\
Normalized & 0.992 & 0.821 & 0.771 & 22.037 & 1.017 & -1.196 & 0.742 & 18.302 & 0.994 & 0.395 & 1.022 & -1.515 \\
Softplus & 0.990 & 0.936 & 0.853 & 14.156 & 1.020 & -1.459 & 0.730 & 19.468 & 0.993 & 0.474 & 0.976 & 1.706 \\
ELU & 0.994 & 0.612 & 0.823 & 17.284 & 1.022 & -1.607 & 0.748 & 18.074 & 0.992 & 0.546 & 0.996 & 0.270 \\
\bottomrule
\end{tabular}
\caption{Generalization performance of the 1-layer ACR-GNN after applying dPTQ with different activation functions (A/F), reported as both Generalization Ratios (Test/Train) and Generalization Gaps (Train – Test accuracy) across datasets $p_1$, $p_2$, and $p_3$.}
\label{tab:results:dptq:ACR-GNN:1layer}
\end{table*}

\subsection{Synthetic data. Comparative analysis}
In this section, we compare the original FLOAT32 model with the dynamically quantized QINT8 model. 
The impact of dynamic PTQ was assessed by calculating the accuracy difference between the original and quantized models ($\Delta_{acc}$).

\textbf{Key finding.} We list the table for the 2-layer ACR-GNN with different activation functions (A/F) across datasets $p_1$, $p_2$, and $p_3$ as the min ($\Delta_{acc}$) across 10 layers.

\begin{table*}[t]
\centering
\begin{tabular}{l ccc ccc ccc ccc}
\toprule
 & \multicolumn{3}{c}{Train} & \multicolumn{3}{c}{Test1} & \multicolumn{3}{c}{Test2}  \\
 & $p_1$ & $p_2$ & $p_3$ & $p_1$ & $p_2$ & $p_3$ & $p_1$ & $p_2$ & $p_3$ \\
\midrule
ReLU & 0.000 & 0.193 & 0.114 & 0.000 & 0.312 & 0.107 & -0.058 & 0.310 & -0.036 \\
ReLU6 & 0.000 & 0.154 & 0.281 & 0.000 & 0.236 & 0.302 & 0.703 & -0.133 & 0.022 \\
trReLU & 0.093 & 0.900 & 0.738 & -0.013 & 1.769 & 0.435 & -0.047 & 2.572 & 2.494 \\
LeakyReLU & 0.001 & 0.101 & 0.147 & 0.000 & 0.357 & 0.280 & 0.155 & 1.452 & -0.011 \\
GELU & 0.006 & 0.631 & 0.189 & 0.000 & 1.368 & 0.226 & 1.059 & 0.782 & 0.022 \\
SiLU & -0.000 & 0.156 & 0.142 & 0.000 & 0.593 & 0.138 & 0.231 & 0.274 & 0.014 \\
Sigmoid & 0.006 & 0.281 & 0.101 & 0.009 & 0.891 & -0.351 & -0.083 & 0.717 & 0.018 \\
Normalized & 0.012 & 0.189 & 0.176 & 0.018 & 0.089 & -0.027 & 0.022 & 0.086 & -0.123 \\
Softplus& 0.000 & -0.095 & 0.140 & 0.000 & -0.004 & -0.120 & 0.288 & 0.324 & -0.004  \\
ELU & 0.001 & 0.310 & 0.247 & 0.000 & 0.165 & 0.515 & -0.735 & -0.400 & 0.025 \\
\bottomrule
\end{tabular}
\caption{Accuracy differences ($\Delta_{acc}$, \%) of the 2-layer ACR-GNN with different activation functions (A/F) across datasets $p_1$, $p_2$, and $p_3$.}
\label{tab:results:nonqua:ACR-GNN:2layer}
\end{table*}

Table~\ref{tab:results:nonqua:ACR-GNN:2layer} reports the accuracy differences ($\Delta_{acc}$) between the original and dynamically quantized models. Overall, the impact of dPTQ is limited, with most deviations remaining within $\pm 1\%$ across datasets and activation functions, indicating that dynamic quantization preserves predictive performance for shallow ACR-GNN architectures. Dataset $p_1$ exhibits the highest robustness, with nearly negligible deviations on both Test1 and Test2. In contrast, $p_2$ shows slightly higher sensitivity, particularly for trReLU and LeakyReLU, though the observed shifts remain moderate. Among activation functions, Normalized and SiLU demonstrate the most consistent stability across Train and Test splits, while trReLU exhibits the largest deviations under quantization. Importantly, no substantial divergence between training and testing accuracy is observed after quantization, suggesting that dynamic PTQ does not introduce additional overfitting or structural instability. These results indicate that 8-bit dynamic PTQ is a reliable compression strategy for 2-layer ACR-GNN models.

The next step is to see the memory reduction.
We began analyzing how dPTQ influences memory. Further computation was performed based on key findings: a hidden dimension of 32 and a two-layer model with a learning rate of 0.01.

We started with the Disk Compression, which is calculated by:
\begin{equation*}
    \text{Disk Compression}= \frac{\text{Model File size (MB) FP32}}{\text{Model File size (MB) INT8}}
\end{equation*}
For this configuration, we got the $\text{Disk Compression} = 1.45478385$, which means the INT8 checkpoint file is only about 31\% smaller than FP32.

After an analysis of the quantized model, we found that the quantized checkpoints store scale and zero-point parameters but retain FP32 tensor storage in the exported state\_dict and this comes from the implementation from the PyTorch. 

In our case for the verification this solution is valid, because as input we provide the unknown graph with the feature's values in int, so  we need quantization only of the weight matrix. One question was highlighted during the research> Should we apply quantizatio to the bias vector? To the bias vector should be applied the rule>
$b_q[i]=round(b[i]/s_w)$, where $s_w$ is the scale of the weights.

\subsection{Synthetic data. Verification}
Here we are working with the trained models in the format .pth. We select one model for the detailed analysis. Some parameters can be extracted from the model's path (more details are in the replication package).
The model is stored in: \textit{saved\_models/results\_synthetic/acrgnn\_relu/p1/}
The path components encode:
synthetic dataset, ACR-GNN architecture with ReLU activation, and training for the $\mathrm{FOC}_2$ classifier ($p_1$).

The checkpoint \textit{MODEL-acrgnn-0-aggS-readS-combT-cl1-L1-H5.pth} encodes the architectural configuration: aggregation type (aggS), readout (readS), combination strategy (combT), one convolutional layer (L1), and hidden dimension 5 (H16). We will use these parameters to pass them to the ESMBC code. 

Following the transformation flow illustrated in Figure~\ref{fig:model_to_esbmc}, the model parameters are extracted from the state dictionary.

For each convolutional component $T \in \{C, A, R\}$, the following tensors are retrieved:
\begin{itemize}
    \item \textit{convs.0.T.linear.weight} — weight matrix of $T$,
    \item \textit{convs.0.T.linear.bias} — bias vector of $T$.
\end{itemize}

The Batch Normalization layer parameters are for calculating the $a_i$ and $c_i$ for the equation\ref{eq:bn_mod}:
\begin{itemize}
    \item $\gamma$ - \textit{batch\_norms.0.weight},
    \item $\beta$ - \textit{batch\_norms.0.bias},
    \item $\mu$ - \textit{batch\_norms.0.running\_mean},
    \item $\sigma^2$ - \textit{batch\_norms.0.running\_var},
\end{itemize}
The prediction layer parameters are:
\begin{itemize}
    \item \textit{linear\_prediction.weight},
    \item \textit{linear\_prediction.bias}.
\end{itemize} 
This can be added or ignored, depending on the verification property, so for the ESBMC tool, we need to verify. The prediction layer is included in the SMT encoding, as the verified
property is defined over the final classifier output.

We performed the same analysis, but for the quantized model.
The checkpoint \textit{MODEL-acrgnn-0-aggS-readS-combT-cl1-L1-H5-8-quantized.pth} encodes the same architectural configuration as the floating-point model—aggregation type (aggS), readout (readS), combination strategy (combT), one convolutional layer (L1), and hidden dimension 5 (H5)—but with parameters stored in quantized form (8-bit). These configuration parameters are again used to initialize the ESBMC verification task.

Following the transformation flow illustrated in Figure~\ref{fig:model_to_esbmc}, the quantized model parameters are extracted from the state dictionary.

For each convolutional component $T \in \{C, A, R\}$, the following entries are retrieved:
\begin{itemize}
    \item \textit{convs.0.T.linear.scale} - quantization scale of the linear layer,
    \item \textit{convs.0.T.linear.zero\_point} - zero-point used in affine quantization,
    \item \textit{convs.0.T.linear.\_packed\_params.\_packed\_params} - packed representation containing the quantized \textit{int8} weight tensor and the associated bias.
\end{itemize}

The packed parameters are unpacked to obtain the raw integer weight matrix (via \textit{int\_repr()}) and the corresponding bias vector. 

\textbf{Key finding.} Since PyTorch stores quantized linear weights in \textit{[out, in]} format, the matrices must be transposed before being passed to ESBMC in order to match the \textit{xC+yA+zR+b} encoding, which expects matrices in \textit{[in, out]} layout.The bias is aligned with the accumulator scale and provided as a row vector consistent with the ESBMC implementation.

\subsection{Synthetic data. Performance}
\label{appendixsection:toolperformance}
We propose the first implementation to serve as a proof of concept, and a baseline for future research.
The prototype directly transforms an instance of a ACR-GNN satisfiability problem into a C program.
The C program is then verified by the model checker ESBMC~\cite{esbmc2024}.

Just to get an idea, we report in Table~\ref{tab:tool-performance} the performance of our prototype on a very small ACR-GNN $\aGNN_{test}$ (three layers of input and output dimensions of three). 

\begin{table*}[t]
    \centering
\begin{minipage}{0.45\textwidth}
\centering
\small
\begin{tabular}{crr}
\toprule
$N$ & Time (s) & \# Graphs $(2^{N^2})$ \\
\midrule
1 & 0.089 & 2 \\
2 & 0.103 & 16 \\
3 & 0.845 & 512 \\
4 & 2.576 & 65536 \\
5 & 10.406 & 33554432 \\
6 & 32.667 & 68719476736 \\
\bottomrule
\end{tabular}
\end{minipage}
\hfill
\begin{minipage}{0.45\textwidth}
\centering
\begin{tikzpicture}
  \begin{axis}[
    xlabel={Number of vertices},
    ylabel={Time (s)},
    grid=both,
    width=\linewidth,
    height=4.2cm
  ]
    \addplot[
      mark=*,
      color=blue,
      thick
    ] coordinates {
      (1,0.089)
      (2,0.103)
      (3,0.845)
      (4,2.576)
      (5,10.406)
      (6,32.667)
    };
  \end{axis}
\end{tikzpicture}
\end{minipage}
\caption{Time required to solve the ACR-GNN satisfiability problem for $\aGNN_{test}$ as the number of vertices increases. The table also reports the number of possible directed graphs with self-loops over $N$ vertices.}
\label{tab:tool-performance}
\end{table*}
Expectedly, the experimental results reveal a bad scalability.
%

We integrate matrices into the ACR-GNN simulation via ESBMC. We have loaded two state dictionaries: the original and the quantized one. The experiments were done on unknown graphs with different numbers of vertices, and the results are presented in the Table~\ref{tab:esbmc_time_relu_h2_l1}.

\begin{table}[tbp]
\centering
\begin{tabular}{clc}
\toprule
Vertices (N) & Time (s) & Time (hh:mm:ss) \\
\midrule
1 & 16.8044   & 00:00:16 \\
2 & 5421.4812 & 01:30:21 \\
\bottomrule
\end{tabular}
\caption{ESBMC verification time for ACR-GNN (H=2, L=1, ReLU).}
\label{tab:esbmc_time_relu_h2_l1}
\end{table}
The results (Table~\ref{tab:esbmc_time_relu_h2_l1}) empirically demonstrate the state-space explosion problem in SMT-based verification of GNNs. Even minimal increases in graph size result in exponential growth in verification time, confirming that naive full symbolic encoding does not scale for graph-structured neural models.

\subsection{Checking Distributivity}
\label{appendixsection:distributivity}

We provide C source code for checking distributivity. The reader may run the model checker ESBMC on it to see whether distributivity holds or not. The code for the checking is in the
supplementary material, and located in the folder
\texttt{src\_check\_distributivity}.

\subsection{Real data. Description}
To test the technique not only on synthetic data, we chose the Protein-Protein Interactions (PPI) benchmark~\cite{zitnik2017predicting} as in the reference paper of~\cite{DBLP:conf/iclr/BarceloKM0RS20}.
The PPI dataset consists of graph-level mini-batches, with separate splits for Training, Validation, and Testing.

\begin{table*}[t]
    \centering
    \begin{tabular}{lccccc}
        \toprule
        Split & \#Graphs & Node Dim & Label Dim & Avg Labels & Avg Degree \\
        \midrule
        Train      & 20 & 50 & 121 & 37.20 & 54.62 \\
        Validation &  2 & 50 & 121 & 35.64 & 61.07 \\
        Test       &  2 & 50 & 121 & 36.22 & 58.64 \\
        \bottomrule
    \end{tabular}
    \caption{Summary of the PPI dataset.}
    \label{tab:ppi_dataset_info}
\end{table*}

In Table~\ref{tab:ppi_dataset_info}, we present a summary of the PPI dataset, which consists of 20 training graphs, 2 validation graphs, and 2 test graphs. Each graph contains nodes with 50-dimensional features and supports multi-label classification with 121 possible labels. On average, each node is associated with approximately 36 labels, indicating a densely labeled dataset. The average node degree is also high, ranging from 54.6 in the training set to 61.1 in the validation set, reflecting the dense connectivity of the protein-protein interaction graphs. The dataset presents a complex multi-label classification task with consistently rich structure across all splits.

\begin{table*}[t]
    \centering
    \begin{tabular}{lcccccc}
        \toprule
        & \multicolumn{3}{c}{Node} & \multicolumn{3}{c}{Edge} \\
        \cmidrule(r){2-4} \cmidrule(r){5-7}
        Dataset & Min & Max & Avg & Min & Max & Avg \\
        \midrule
        Train      &  591 & 3480 & 2245.30 &  7708  & 106754 & 61318.40 \\
        Validation & 3230 & 3284 & 3257.00 & 97446  & 101474 & 99460.00 \\
        Test       & 2300 & 3224 & 2762.00 & 61328  & 100648 & 80988.00 \\
        \bottomrule
    \end{tabular}
    \caption{Dataset statistics summary. PPI benchmark.}
    \label{tab:ppi_dataset_summary}
\end{table*}

The statistics of the dataset presented in Table~\ref{tab:ppi_dataset_summary} contain large graphs with varying sizes between the train, the validation, and the test splits. Training graphs range from 591 to 3,480 nodes, with an average of 2,245 nodes per graph, and between 7,708 and 106,754 edges (average 61,318 edges). Validation graphs are more consistent in size, with 3,230 to 3,284 nodes and 97,446 to 101,474 edges, averaging 3,257 nodes and 99,460 edges. The test graphs have 2,300 to 3,224 nodes, averaging 2,762 nodes, and 61,328 to 100,648 edges, averaging 80,988. These statistics confirm that the dataset contains large and densely connected graphs and demonstrate a distributional shift in graph size and edge count between training and test data. This information is helpful in evaluating the model's ability to generalize to unseen and variable graph structures.

One key difference between the synthetic data and the PPI dataset is that the latter involves a multi-label classification task, rather than a binary classification task, because the PPI dataset is a common benchmark where each node (representing proteins) can have multiple labels, such as protein functions or interactions. Also, it is important to mention the key differences between the synthetic data and the real one. Here, the authors \cite{BarceloGit2021} used the code function \textit{EarlyStopping}: Utility for stopping training early if no further improvement is observed. The second difference is that the code is structured to run multiple experiments to collect statistics (mean and standard deviation) of the model performance, ensuring that the results are robust across different random initializations. In this case, we performed the experiments 10 times for each model, with a combination layer equal to 1 and a number of layers ranging from 1 to 10. The number of hidden dimensions is equal to 256. 

We applied the same eight activation functions to train the model. We also continue the experimental flow for real-world data, focusing on running time (Table~\ref{tab:ppi:nonqua:training-time}), speedup (Figure~\ref{fig:ppi:speedup}), size reduction (Table~\ref{tab:ppi:changes_of_size}), and analysis of accuracy.

\subsection{Real data. Analysis after training}
We analyze the total training time of the ACR-GNN across ten layers for different activation functions. Table~\ref{tab:ppi:nonqua:training-time} reports the total runtime in seconds and minute.

\begin{table*}[t]
\centering
\begin{tabular}{l ccc ccc cc}
\toprule
A/F & ReLU & ReLU6 & trReLU & GELU & Sigmoid & SiLU & Softplus & ELU \\
\midrule
Time (min) & 204.70 & 244.20 & 187.60 & 192.40 & 229.20 & 232.50 & 250.10 & 234.40 \\
Time (s) & 12286.30 & 14650.70 & 11261.60 & 11541.90 & 13745.00 & 13951.00 & 15007.00 & 14059.50 \\
\bottomrule
\end{tabular}
\caption{Training time per activation function}
\label{tab:ppi:nonqua:training-time}
\end{table*}

The results show substantial variability depending on the activation function. Piecewise activations such as trReLU (187.6 min) and ReLU (204.7 min) yield the fastest training times, while smooth activations such as Softplus (250.1 min), ReLU6 (244.2 min), and SiLU (232.5 min) incur significant overhead. Sigmoid also ranks among the slower functions (229.2 min).

\begin{table*}[t]
    \centering
    \begin{tabular}{@{}l*{10}{c}@{}}
    \toprule
    & 1 & 2 & 3 & 4 & 5 & 6 & 7 & 8 & 9 & 10 \\
    \midrule
    Fastest & Sigmoid & ELU & ReLU & ReLU & trReLU & ReLU & GELU & trReLU & trReLU & trReLU \\
    Slowest & ReLU6 & SiLU & ReLU6 & ReLU6 & SiLU & ReLU6 & Softplus & Softplus & Softplus & ReLU6 \\
    \bottomrule
  \end{tabular}
  \caption{Slowest and fastest activation functions across the depth of the ACR-GNN.}
    \label{tab:ppi:nonqua:minmax}
\end{table*}

As shown in Table~\ref{tab:ppi:nonqua:training-time}, the training time for all activation functions increases dramatically after the second layer. This highlights that not only the type of activation function influences performance time, but also the depth of the model.

Depth-wise analysis (Table~\ref{tab:ppi:nonqua:minmax}) confirms this pattern: trReLU frequently provides the lowest training time at deeper layers (5, 8–10), while ReLU6 consistently emerges as the slowest. These findings indicate that the choice of activation function significantly impacts computational efficiency on PPI, with piecewise functions offering faster convergence than their smooth counterparts.

We measured the size of the model (in Table~\ref{tab:ppi:size_of_model}) and obtained that the choice of activation function did not influence the size of the model.

\begin{table*}[t]
    \centering
\begin{tabular}{lcccccccccc}
    \toprule
    Layers & 1 & 2 & 3 & 4 & 5 & 6 & 7 & 8 & 9 & 10 \\
    \midrule
    Size (MB) & 
    0.92 &
    1.72 &
    2.51 &
    3.31 &
    4.11 &
    4.9 &
    5.7 &
    6.5 &
    7.29 &
    8.09 \\
    \bottomrule
    \end{tabular}
    \caption{PPI. Model size in MB as a function of the number of layers.}
    \label{tab:ppi:size_of_model}
\end{table*}

\subsection{Real data. Comparative analysis}
To assess the computational efficiency of dPTQ, we measured the elapsed time (Table~\ref{tab:ppi:elapsedtime}) of each model across different activation functions and datasets. 

\begin{table}[tbp]
\centering
\begin{tabular}{l ccc ccc cc}
\toprule
 & \multicolumn{2}{c}{Train} & \multicolumn{2}{c}{Test}& \multicolumn{2}{c}{Validation} \\
        \cmidrule(r){2-3} \cmidrule(r){4-5}\cmidrule(r){6-7}
A/F &        O & Q & O & Q & O & Q \\
\midrule
ReLU & 22.46 & 23.13 & 2.56 & 2.57 & 3.04 & 3.02 \\
ReLU6 & 21.18 & 22.44 & 2.54 & 2.69 & 2.90 & 3.17 \\
trReLU & 19.93 & 20.61 & 2.50 & 2.51 & 2.89 & 3.00 \\
GELU & 23.89 & 25.37 & 2.82 & 2.91 & 3.40 & 3.25 \\
Sigmoid & 22.52 & 24.53 & 2.72 & 2.85 & 3.17 & 3.24 \\
SiLU & 23.14 & 24.06 & 2.72 & 2.87 & 3.51 & 3.16 \\
Softplus & 21.65 & 24.97 & 2.76 & 3.04 & 3.27 & 3.23 \\
ELU & 26.13 & 26.05 & 3.31 & 3.22 & 3.90 & 3.52 \\
\bottomrule
\end{tabular}
\caption{Total elapsed time (s) per activation function and datasets before and after applying dPTQ, where O denotes the original model and Q the quantized model.}
\label{tab:ppi:elapsedtime}
\end{table}

Table~\ref{tab:ppi:elapsedtime} reports the total elapsed time (in seconds) for training, testing, and validation phases across activation functions, comparing original and quantized models. The results indicate that dynamic PTQ introduces only marginal differences in runtime across all phases and activations. In most cases, quantized models require slightly longer execution time (e.g., ReLU6 and Softplus), while in a few instances, minor improvements are observed (e.g., ELU in Test and Validation). Overall, the runtime overhead of quantization remains negligible, suggesting that the primary benefit of dynamic PTQ lies in memory and storage efficiency rather than acceleration.

We measure the speedup (Figure~\ref{fig:ppi:speedup}) of this type of quantization technique. We report the mean dynamic PTQ speedup across 10 layers, defined as the ratio of non-quantized to quantized execution time (original time / dPTQ time), which indicates whether dynamic PTQ reduces or increases runtime.

\begin{figure*}[t]
    \centering
    \includegraphics[width=\linewidth]{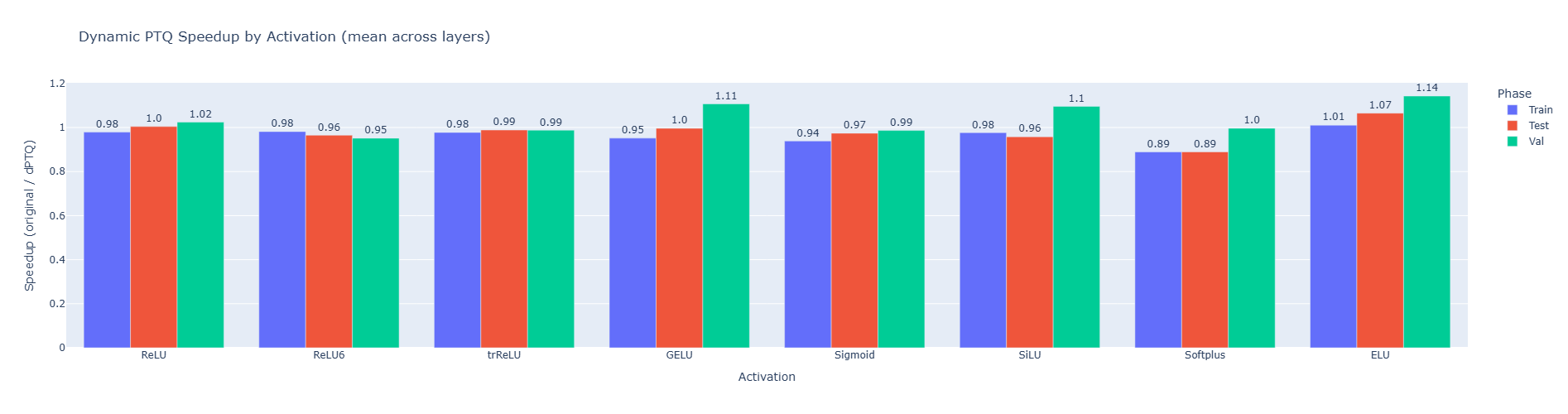}
    \caption{PPI. Dynamic PTQ Speedup by Activation (mean across layers).}
    \label{fig:ppi:speedup}
\end{figure*}

Figure~\ref{fig:ppi:speedup} reports the mean speedup values across layers for different activation functions. 
Overall, dynamic PTQ yields values close to 1, indicating only minor runtime benefits. 
ELU demonstrates the most consistent improvement, with speedup up to 1.14 in validation and above 1.07 in test, followed by SiLU and GELU, which also provide modest acceleration during validation. 
In contrast, Softplus incurs consistent slowdowns (speedup $\approx 0.89$ in training and testing), while Sigmoid and ReLU6 remain below 1, showing limited suitability for quantized execution. 
These results indicate that smooth activations such as ELU, SiLU, and GELU are better aligned with quantized computation, whereas Softplus and Sigmoid are unfavorable for efficient PTQ deployment. 

We report the results in Table~\ref{tab:ppi:changes_of_size} about the difference of the model's size. We calculated the $\Delta_{Size}$ and Reduction (\%) across the depth. The main result of this experiment is the following: the total reduction in size is $\approx$74\%. That is really good and significant, for example, for the application part of the quantization, where the model can be used on a low-power computer.

\begin{table}[tbp]
    \centering
\begin{tabular}{ccccc}
    \toprule
    L & O (MB) & Q (MB) & Delta & Reduction (\%) \\
    \midrule
    1 & 0.922108 & 0.242060 & 0.680048 & -73.7 \\
    2 & 1.718266 & 0.450790 & 1.267476 & -73.8 \\
    3 & 2.514808 & 0.659584 & 1.855224 & -73.8 \\
    4 & 3.311350 & 0.868378 & 2.442972 & -73.8 \\
    5 & 4.107892 & 1.077172 & 3.030720 & -73.8 \\
    6 & 4.904370 & 1.285972 & 3.618398 & -73.8 \\
    7 & 5.700912 & 1.494783 & 4.206129 & -73.8 \\
    8 & 6.497390 & 1.703594 & 4.793796 & -73.8 \\
    9 & 7.293933 & 1.912405 & 5.381528 & -73.8 \\
    10 & 8.090486 & 2.121216 & 5.969270 & -73.8 \\
    \bottomrule
    \end{tabular}
    \caption{PPI. Influence of the dPTQ on the size of the model, where O denotes the original model, and Q the quantized model, L is the number of layers.}
    \label{tab:ppi:changes_of_size}
\end{table}

We constructed tables with specific structural requirements to better examine the influence of dynamic PTQ on PPI data. The impact of dynamic PTQ was assessed by calculating the Generalization Ratio (GR), the Generalization Gap ($\Delta_{gen}$), and the accuracy difference between the original and quantized models ($\Delta_{acc}$). The results are presented in Tables~\ref{tab:ppi:dacc_gr_layer1} --Tables~\ref{tab:ppi:dacc_gr_layer10}. Below, we summarize the principal observations layer by layer.

\begin{table}[tbp]
\centering
\small
\begin{tabular}{l ccc ccc}
\toprule
A/F & \multicolumn{3}{c}{Test} & \multicolumn{3}{c}{Validation} \\
\cmidrule(r){2-4} \cmidrule(r){5-7}
 & GR & $\Delta_{gen}$ & $\Delta_{acc}$\% & GR & $\Delta_{gen}$ & $\Delta_{acc}$\% \\
\midrule
ReLU & 0.821 & +0.108 & -0.018 & 0.866 & +0.081 & +0.001 \\
ReLU6 & 0.835 & +0.090 & +0.007 & 0.917 & +0.045 & +0.006 \\
trReLU & 0.728 & +0.143 & +0.001 & 0.722 & +0.146 & +0.000 \\
GELU & 0.716 & +0.168 & -0.005 & 0.848 & +0.090 & -0.006 \\
Sigmoid & 0.791 & +0.109 & -0.003 & 0.741 & +0.135 & -0.001 \\
SiLU & 0.765 & +0.138 & -0.010 & 0.854 & +0.086 & -0.006 \\
Softplus & 0.667 & +0.197 & +0.025 & 0.802 & +0.118 & +0.016 \\
ELU & 0.719 & +0.156 & +0.015 & 0.774 & +0.125 & +0.007 \\
\bottomrule
\end{tabular}\caption{PPI. Accuracy differences ($\Delta_{acc}$,\% ) and generalization metrics (GR, $\Delta_{gen}$) per activation and dataset for one-layer ACR-GNN after applying the dynamic PTQ.}
\label{tab:ppi:dacc_gr_layer1}
\end{table}

For layer 1 (Table~\ref{tab:ppi:dacc_gr_layer1})  
ReLU6 achieves the strongest validation GR (0.917, $\Delta_{gen}=+0.045$) with negligible $\Delta_{acc}$. Softplus performs worst (GR 0.667/0.802 with the largest gaps). 

\begin{table}[tbp]
\centering
\small
\begin{tabular}{l ccc ccc}
\toprule
A/F & \multicolumn{3}{c}{Test} & \multicolumn{3}{c}{Validation} \\
\cmidrule(r){2-4} \cmidrule(r){5-7}
 & GR & $\Delta_{gen}$ & $\Delta_{acc}$ & GR & $\Delta_{gen}$ & $\Delta_{acc}$ \\
\midrule
ReLU & 0.606 & +0.241 & -0.001 & 0.637 & +0.222 & +0.031 \\
ReLU6 & 0.703 & +0.161 & +0.004 & 0.668 & +0.180 & -0.009 \\
trReLU & 0.694 & +0.157 & -0.001 & 0.687 & +0.161 & +0.002 \\
GELU & 0.681 & +0.195 & +0.011 & 0.723 & +0.170 & +0.009 \\
Sigmoid & 0.743 & +0.133 & +0.004 & 0.735 & +0.137 & -0.000 \\
SiLU & 0.677 & +0.197 & +0.015 & 0.635 & +0.223 & +0.006 \\
Softplus & 0.681 & +0.200 & -0.004 & 0.655 & +0.216 & +0.005 \\
ELU & 0.717 & +0.172 & -0.003 & 0.721 & +0.170 & +0.008 \\
\bottomrule
\end{tabular}
\caption{PPI. Accuracy differences ($\Delta_{acc}$, \%) and generalization metrics (GR, $\Delta_{gen}$) per activation and dataset for two-layer ACR-GNN after applying the dynamic PTQ.}
\label{tab:ppi:dacc_gr_layer2}
\end{table}

For layer 2 (Table~\ref{tab:ppi:dacc_gr_layer2})  
ELU and Sigmoid are comparatively stable ($\text{GR}\approx0.72$--0.74), whereas ReLU degrades (GR~0.606/0.637 with $\Delta_{gen}>0.22$). Accuracy changes remain within $\pm0.03\%$.

\begin{table}[tbp]
\centering
\small
\begin{tabular}{l ccc ccc}
\toprule
A/F & \multicolumn{3}{c}{Test} & \multicolumn{3}{c}{Validation} \\
\cmidrule(r){2-4} \cmidrule(r){5-7}
 & GR & $\Delta_{gen}$ & $\Delta_{acc}$ & GR & $\Delta_{gen}$ & $\Delta_{acc}$ \\
\midrule
ReLU & 0.576 & +0.261 & -0.006 & 0.629 & +0.228 & -0.004 \\
ReLU6 & 0.598 & +0.216 & +0.034 & 0.588 & +0.221 & +0.008 \\
trReLU & 0.727 & +0.143 & +0.040 & 0.691 & +0.162 & +0.007 \\
GELU & 0.730 & +0.164 & -0.027 & 0.720 & +0.170 & -0.014 \\
Sigmoid & 0.749 & +0.123 & +0.005 & 0.627 & +0.183 & +0.007 \\
SiLU & 0.646 & +0.214 & -0.000 & 0.713 & +0.173 & -0.001 \\
Softplus & 0.551 & +0.282 & +0.001 & 0.531 & +0.295 & +0.000 \\
ELU & 0.653 & +0.215 & +0.024 & 0.639 & +0.223 & -0.010 \\
\bottomrule
\end{tabular}
\caption{PPI. Accuracy differences ($\Delta_{acc}$, \%) and generalization metrics (GR, $\Delta_{gen}$) per activation and dataset for three-layer ACR-GNN after applying the dynamic PTQ.}
\label{tab:ppi:dacc_gr_layer3}
\end{table}

For layer 3 (Table~\ref{tab:ppi:dacc_gr_layer3})  
GELU and trReLU lead (GR $\approx0.73$), while Softplus is lowest (0.551/0.531). $\Delta_{acc}$ remains negligible. 

\begin{table}[tbp]
\centering
\small
\begin{tabular}{l ccc ccc}
\toprule
A/F & \multicolumn{3}{c}{Test} & \multicolumn{3}{c}{Validation} \\
\cmidrule(r){2-4} \cmidrule(r){5-7}
 & GR & $\Delta_{gen}$ & $\Delta_{acc}$ & GR & $\Delta_{gen}$ & $\Delta_{acc}$ \\
\midrule
ReLU & 0.564 & +0.269 & +0.001 & 0.570 & +0.265 & -0.002 \\
ReLU6 & 0.902 & +0.053 & -0.000 & 0.878 & +0.065 & -0.004 \\
trReLU & 0.708 & +0.150 & -0.000 & 0.707 & +0.151 & -0.004 \\
GELU & 0.589 & +0.252 & +0.006 & 0.597 & +0.247 & -0.003 \\
Sigmoid & 0.822 & +0.088 & +0.003 & 0.762 & +0.118 & +0.003 \\
SiLU & 0.612 & +0.231 & -0.001 & 0.595 & +0.241 & +0.001 \\
Softplus & 0.674 & +0.200 & -0.005 & 0.653 & +0.213 & +0.007 \\
ELU & 0.600 & +0.252 & -0.010 & 0.556 & +0.280 & -0.003 \\
\bottomrule
\end{tabular}
\caption{PPI. Accuracy differences ($\Delta_{acc}$, \%) and generalization metrics (GR, $\Delta_{gen}$) per activation and dataset for four-layer ACR-GNN after applying the dynamic PTQ.}
\label{tab:ppi:dacc_gr_layer4}
\end{table}

For layer 4 (Table~\ref{tab:ppi:dacc_gr_layer4})  
ReLU6 clearly dominates (0.902/0.878, small $\Delta_{gen}\approx0.05$--0.07). ELU is weakest ($\approx0.60/0.56$, large gaps). $\Delta_{acc}\approx0\%$.

\begin{table}[tbp]
\centering
\small
\begin{tabular}{l ccc ccc}
\toprule
A/F & \multicolumn{3}{c}{Test} & \multicolumn{3}{c}{Validation} \\
\cmidrule(r){2-4} \cmidrule(r){5-7}
 & GR & $\Delta_{gen}$ & $\Delta_{acc}$ & GR & $\Delta_{gen}$ & $\Delta_{acc}$ \\
\midrule
ReLU & 0.705 & +0.180 & +0.003 & 0.673 & +0.200 & +0.018 \\
ReLU6 & 0.786 & +0.110 & -0.006 & 0.754 & +0.126 & +0.001 \\
trReLU & 0.921 & +0.039 & +0.002 & 0.906 & +0.046 & -0.002 \\
GELU & 0.581 & +0.252 & -0.001 & 0.566 & +0.261 & +0.018 \\
Sigmoid & 0.770 & +0.112 & +0.031 & 0.755 & +0.120 & +0.006 \\
SiLU & 0.671 & +0.195 & +0.001 & 0.720 & +0.166 & -0.033 \\
Softplus & 0.598 & +0.246 & -0.003 & 0.586 & +0.253 & +0.001 \\
ELU & 0.582 & +0.251 & +0.000 & 0.574 & +0.255 & +0.000 \\
\bottomrule
\end{tabular}
\caption{PPI. Accuracy differences ($\Delta_{acc}$, \%) and generalization metrics (GR, $\Delta_{gen}$) per activation and dataset for five-layer ACR-GNN after applying the dynamic PTQ.}
\label{tab:ppi:dacc_gr_layer5}
\end{table}

For layer 5 (Table~\ref{tab:ppi:dacc_gr_layer5})  
trReLU is strongest (0.921/0.906, minimal $\Delta_{gen}$). Softplus and ELU are lowest ($\approx0.59$). Sigmoid shows a small positive $\Delta_{acc}$ on Test (+0.031\%), while SiLU has a small negative $\Delta_{acc}$ on Validation (-0.033\%).

\begin{table}[tbp]
\centering
\small
\begin{tabular}{l ccc ccc}
\toprule
A/F & \multicolumn{3}{c}{Test} & \multicolumn{3}{c}{Validation} \\
\cmidrule(r){2-4} \cmidrule(r){5-7}
 & GR & $\Delta_{gen}$ & $\Delta_{acc}$ & GR & $\Delta_{gen}$ & $\Delta_{acc}$ \\
\midrule
ReLU & 0.730 & +0.161 & -0.003 & 0.717 & +0.169 & +0.004 \\
ReLU6 & 0.672 & +0.168 & +0.000 & 0.663 & +0.173 & +0.000 \\
trReLU & 0.870 & +0.066 & +0.012 & 0.879 & +0.062 & +0.000 \\
GELU & 0.642 & +0.214 & -0.006 & 0.628 & +0.223 & -0.002 \\
Sigmoid & 0.880 & +0.057 & -0.001 & 0.910 & +0.043 & -0.001 \\
SiLU & 0.619 & +0.222 & +0.000 & 0.611 & +0.227 & -0.015 \\
Softplus & 0.641 & +0.216 & -0.025 & 0.685 & +0.190 & -0.031 \\
ELU & 0.624 & +0.232 & +0.004 & 0.626 & +0.231 & +0.010 \\
\bottomrule
\end{tabular}
\caption{PPI. Accuracy differences ($\Delta_{acc}$, \%) and generalization metrics (GR, $\Delta_{gen}$) per activation and dataset for six-layer ACR-GNN after applying the dynamic PTQ.}
\label{tab:ppi:dacc_gr_layer6}
\end{table}

For layer 6 (Table~\ref{tab:ppi:dacc_gr_layer6})  
Sigmoid emerges as best (0.880/0.910, smallest gaps), followed by trReLU (0.870/0.879). SiLU and ELU are weaker ($\approx0.62$). Softplus shows the largest negative $\Delta_{acc}$ (-0.025/-0.031\%), though still small.

\begin{table}[tbp]
\centering
\small
\begin{tabular}{l ccc ccc}
\toprule
A/F & \multicolumn{3}{c}{Test} & \multicolumn{3}{c}{Validation} \\
\cmidrule(r){2-4} \cmidrule(r){5-7}
 & GR & $\Delta_{gen}$ & $\Delta_{acc}$ & GR & $\Delta_{gen}$ & $\Delta_{acc}$ \\
\midrule
ReLU & 0.714 & +0.172 & -0.009 & 0.686 & +0.189 & +0.007 \\
ReLU6 & 0.641 & +0.182 & -0.001 & 0.615 & +0.195 & -0.002 \\
trReLU & 0.696 & +0.155 & +0.000 & 0.688 & +0.159 & +0.000 \\
GELU & 0.744 & +0.151 & +0.010 & 0.759 & +0.143 & -0.000 \\
Sigmoid & 0.866 & +0.061 & +0.014 & 0.841 & +0.072 & -0.000 \\
SiLU & 0.591 & +0.239 & +0.002 & 0.585 & +0.243 & +0.001 \\
Softplus & 0.655 & +0.210 & +0.003 & 0.665 & +0.204 & -0.003 \\
ELU & 0.713 & +0.177 & +0.003 & 0.732 & +0.165 & -0.001 \\
\bottomrule
\end{tabular}
\caption{PPI. Accuracy differences ($\Delta_{acc}$, \%) and generalization metrics (GR, $\Delta_{gen}$) per activation and dataset for seven-layer ACR-GNN after applying the dynamic PTQ.}
\label{tab:ppi:dacc_gr_layer7}
\end{table}

For layer 7 (Table~\ref{tab:ppi:dacc_gr_layer7})  
Sigmoid and GELU lead (0.866/0.841 and 0.744/0.759). SiLU is weakest ($\approx0.59$ with largest $\Delta_{gen}$). Accuracy changes remain near zero.

\begin{table}[tbp]
\centering
\small
\begin{tabular}{l ccc ccc}
\toprule
A/F & \multicolumn{3}{c}{Test} & \multicolumn{3}{c}{Validation} \\
\cmidrule(r){2-4} \cmidrule(r){5-7}
 & GR & $\Delta_{gen}$ & $\Delta_{acc}$ & GR & $\Delta_{gen}$ & $\Delta_{acc}$ \\
\midrule
ReLU & 0.590 & +0.245 & +0.002 & 0.585 & +0.247 & +0.000 \\
ReLU6 & 0.795 & +0.103 & -0.005 & 0.796 & +0.102 & -0.005 \\
trReLU & 0.855 & +0.072 & -0.001 & 0.845 & +0.077 & +0.003 \\
GELU & 0.664 & +0.199 & +0.000 & 0.656 & +0.204 & +0.000 \\
Sigmoid & 0.859 & +0.064 & -0.007 & 0.823 & +0.080 & -0.012 \\
SiLU & 0.717 & +0.162 & +0.001 & 0.714 & +0.164 & +0.000 \\
Softplus & 0.631 & +0.222 & -0.006 & 0.613 & +0.233 & +0.006 \\
ELU & 0.645 & +0.211 & +0.012 & 0.675 & +0.193 & -0.004 \\
\bottomrule
\end{tabular}
\caption{PPI. Accuracy differences ($\Delta_{acc}$, \%) and generalization metrics (GR, $\Delta_{gen}$) per activation and dataset for eight-layer ACR-GNN after applying the dynamic PTQ.}
\label{tab:ppi:dacc_gr_layer8}
\end{table}

For layer 8 (Table~\ref{tab:ppi:dacc_gr_layer8})  
trReLU and Sigmoid are strongest (0.855/0.845 and 0.859/0.823). Softplus is lowest ($\approx0.63/0.61$). $\Delta_{acc}$ values are minimal.

\begin{table}[tbp]
\centering
\small
\begin{tabular}{l ccc ccc}
\toprule
A/F & \multicolumn{3}{c}{Test} & \multicolumn{3}{c}{Validation} \\
\cmidrule(r){2-4} \cmidrule(r){5-7}
 & GR & $\Delta_{gen}$ & $\Delta_{acc}$ & GR & $\Delta_{gen}$ & $\Delta_{acc}$ \\
\midrule
ReLU & 0.614 & +0.225 & -0.002 & 0.610 & +0.228 & -0.002 \\
ReLU6 & 0.699 & +0.155 & +0.001 & 0.690 & +0.159 & -0.002 \\
trReLU & 0.864 & +0.066 & -0.003 & 0.846 & +0.075 & -0.015 \\
GELU & 0.634 & +0.216 & -0.001 & 0.627 & +0.221 & -0.000 \\
Sigmoid & 0.797 & +0.093 & +0.023 & 0.840 & +0.073 & -0.015 \\
SiLU & 0.760 & +0.144 & -0.002 & 0.752 & +0.148 & -0.003 \\
Softplus & 0.589 & +0.248 & -0.001 & 0.586 & +0.249 & -0.004 \\
ELU & 0.690 & +0.187 & +0.005 & 0.710 & +0.175 & -0.014 \\
\bottomrule
\end{tabular}
\caption{PPI. Accuracy differences ($\Delta_{acc}$, \%) and generalization metrics (GR, $\Delta_{gen}$) per activation and dataset for nine-layer ACR-GNN after applying the dynamic PTQ.}
\label{tab:ppi:dacc_gr_layer9}
\end{table}

For layer 9 (Table~\ref{tab:ppi:dacc_gr_layer9})  
trReLU again achieves best generalization (0.864/0.846). Softplus is weakest ($\approx0.59$). $\Delta_{acc}$ is small, with mixed signs for Sigmoid (+0.023\% Test, -0.015\% Validation splits). 

\begin{table}[tbp]
\centering
\small
\begin{tabular}{l ccc ccc}
\toprule
A/F & \multicolumn{3}{c}{Test} & \multicolumn{3}{c}{Validation} \\
\cmidrule(r){2-4} \cmidrule(r){5-7}
 & GR & $\Delta_{gen}$ & $\Delta_{acc}$ & GR & $\Delta_{gen}$ & $\Delta_{acc}$ \\
\midrule
ReLU & 0.554 & +0.260 & +0.008 & 0.571 & +0.250 & -0.006 \\
ReLU6 & 0.720 & +0.138 & +0.004 & 0.711 & +0.143 & +0.004 \\
trReLU & 0.732 & +0.133 & -0.008 & 0.715 & +0.141 & -0.006 \\
GELU & 0.621 & +0.225 & -0.023 & 0.588 & +0.245 & -0.008 \\
Sigmoid & 0.805 & +0.088 & +0.019 & 0.809 & +0.087 & +0.007 \\
SiLU & 0.602 & +0.232 & +0.071 & 0.586 & +0.241 & +0.079 \\
Softplus & 0.729 & +0.160 & -0.004 & 0.705 & +0.175 & +0.004 \\
ELU & 0.749 & +0.152 & -0.014 & 0.726 & +0.166 & -0.019 \\
\bottomrule
\end{tabular}
\caption{PPI. Accuracy differences ($\Delta_{acc}$, \%) and generalization metrics (GR, $\Delta_{gen}$) per activation and dataset for ten-layer ACR-GNN after applying the dynamic PTQ.}
\label{tab:ppi:dacc_gr_layer10}
\end{table}

For layer 10 (Table~\ref{tab:ppi:dacc_gr_layer10})  
Sigmoid is strongest (0.805/0.809). ReLU is weakest (0.554/0.571, $\Delta_{gen}\approx0.25$). SiLU shows the largest absolute $\Delta_{acc}$ (+0.071/+0.079\%), but still below 0.1\%.

Dynamic PTQ preserves accuracy across all layers, with $|\Delta_{acc}|<0.1\%$ in nearly every case. Generalization robustness varies by activation: trReLU and ReLU6 perform most consistently across layers, Sigmoid becomes increasingly stable in deeper layers, 
while Softplus is the weakest choice, and plain ReLU tends to degrade with depth. 
These findings confirm that quantized models retain generalization performance on PPI, with activation choice being the primary factor for robustness under PTQ.

\subsection{Real data. How accuracy varies with the quantization bit width}
After comparing the original model with its dynamically quantized variant, we analyzed how accuracy varies with the quantization bit width. To address this, we applied fake quantization to the full-precision model and evaluated performance under different bit-width configurations. The experiments were conducted with the following settings:
\begin{itemize}
    \item Number of layers from 1 to 3.
    \item Activation function: ReLU, ReLU6, trReLU, GELU, Sigmoid, SiLU, Softplus, ELU.
    \item Data split. Train, Validation, and Test for the PPI data.
    \item ``Fake Quantization". [2,4,5,6,7,8,16,32]-bit. Here, 32-bit corresponds to the original full-precision model, and 8-bit corresponds to the dynamic PTQ configuration reported before.
\end{itemize}
To isolate the most significant accuracy degradations, we filtered all bit-width transitions using two criteria: (i) both bit-widths must be greater than 4, thereby excluding the trivial drop associated with the 4 to 2-bit collapse; and (ii) the relative accuracy loss must exceed 1\%. Note that we first searched for 'lose more than 10\% of accuracy', but we did not observe a decrease in accuracy.

The resulting tables (\ref{tab:ppi-bitwidth-drops-RELU} -- \ref{tab:ppi-bitwidth-drops-ELU}) report, for each activation function, the data split (Train/Validation/Test), the layer index, the transition between bit-widths (\textit{From bits} and \textit{To bits}), and the computed accuracy difference. These entries highlight only the meaningful degradation patterns that occur before entering the low-precision regime.

\begin{table}[tbp]
\centering
\begin{tabular}{cccccc}
\toprule
Split & A/F  & Layer & From & To & Drop \\
\midrule
Test & RELU & 1 & 32 & 16 & -0.010267 \\
Test & RELU & 1 & 7 & 6 & -0.027213 \\
Test & RELU & 3 & 32 & 16 & -0.013466 \\
Test & RELU & 3 & 5 & 4 & -0.032393 \\
Validation & RELU & 2 & 16 & 8 & -0.015078 \\
\bottomrule
\end{tabular}
\caption{PPI. Accuracy drops across bit-width reductions for the RELU activation function. Only over 1\% drops are reported.}
\label{tab:ppi-bitwidth-drops-RELU}
\end{table}

According to the Table\ref{tab:ppi-bitwidth-drops-RELU}, ReLU activation exhibits only small accuracy degradations across bit-widths, all below 3.3\%. Drops occur primarily in the Test split, with notable decreases from 7 to 6 bits and from 5 to 4 bits, and mild drops from 32 to 16 bits. A single Validation drop is observed from 16 to 8 bits.

\begin{table}[tbp]
\centering
\begin{tabular}{cccccc}
\toprule
Split & A/F  & Layer & From & To & Drop \\
\midrule
Train & RELU6 & 1 & 8 & 7 & -0.011276 \\
Train & RELU6 & 2 & 5 & 4 & -0.012176 \\
\bottomrule
\end{tabular}
\caption{PPI. Accuracy drops across bit-width reductions for the RELU6 activation function. Only over 1\% drops are reported.}
\label{tab:ppi-bitwidth-drops-RELU6}
\end{table}
According to the Table\ref{tab:ppi-bitwidth-drops-RELU6}, ReLU6 shows very limited sensitivity to quantization. Only two drops exceed the 1\% threshold: one in the Train split from 8 to 7 bits, and one from 5 to 4 bits for layer 2. No drops were observed in the Validation or Test sets, indicating that ReLU6 behaves consistently and is robust under moderate bit-width reductions.

\begin{table}[tbp]
\centering
\begin{tabular}{cccccc}
\toprule
Split & A/F  & Layer & From & To & Drop \\
\midrule
Test & trRELU & 1 & 32 & 16 & -0.026129 \\
Test & trRELU & 1 & 7 & 6 & -0.023775 \\
Test & trRELU & 1 & 5 & 4 & -0.013549 \\
Test & trRELU & 2 & 16 & 8 & -0.029678 \\
Test & trRELU & 2 & 6 & 5 & -0.012522 \\
Test & trRELU & 2 & 5 & 4 & -0.014027 \\
Validation & trRELU & 1 & 32 & 16 & -0.014119 \\
Validation & trRELU & 1 & 7 & 6 & -0.014122 \\
Validation & trRELU & 1 & 5 & 4 & -0.010713 \\
Validation & trRELU & 2 & 8 & 7 & -0.011061 \\
Validation & trRELU & 3 & 16 & 8 & -0.017004 \\
\bottomrule
\end{tabular}
\caption{PPI. Accuracy drops across bit-width reductions for the trRELU activation function. Only over 1\% drops are reported.}
\label{tab:ppi-bitwidth-drops-trRELU}
\end{table}
According to the Table\ref{tab:ppi-bitwidth-drops-trRELU}, trReLU is the most sensitive activation among those tested. Drops appear across all splits and layers, with several consistent degradations from 32 to 16 bits, from 7 to 6 bits, from 6 to 5 bits, and 5 to 4 bits. Both Test and Validation splits show multiple decreases, some approaching 3\%. This pattern indicates that trReLU is more vulnerable to precision reduction and benefits from bit-widths greater that 6 bits.

\begin{table}[tbp]
\centering
\begin{tabular}{cccccc}
\toprule
Split & A/F  & Layer & From & To & Drop \\
\midrule
Train & GELU & 1 & 8 & 7 & -0.011261 \\
Test & GELU & 1 & 32 & 16 & -0.012628 \\
Test & GELU & 1 & 7 & 6 & -0.025353 \\
Test & GELU & 1 & 5 & 4 & -0.017917 \\
Test & GELU & 3 & 7 & 6 & -0.015957 \\
Test & GELU & 3 & 5 & 4 & -0.016658 \\
Validation & GELU & 1 & 32 & 16 & -0.012274 \\
Validation & GELU & 1 & 7 & 6 & -0.014351 \\
Validation & GELU & 3 & 32 & 16 & -0.012390 \\
Validation & GELU & 3 & 8 & 7 & -0.010392 \\
Validation & GELU & 3 & 5 & 4 & -0.070996 \\
\bottomrule
\end{tabular}
\caption{PPI. Accuracy drops across bit-width reductions for the GELU activation function. Only over 1\% drops are reported.}
\label{tab:ppi-bitwidth-drops-GELU}
\end{table}

According to the Table\ref{tab:ppi-bitwidth-drops-GELU}, GELU displays a mixed pattern: it is stable at higher bit-widths but shows several consistent drops once bit-width falls below 6 bits. The most substantial drop occurs in the Validation split, from 5 to 4 bits (approximately 7\%). Additional degradations appear across layers from 32 to 16 bits, from 7 to 6 bits, and from 5 to 4 bits. While GELU is reliable at 16–8 bits, its performance deteriorates noticeably at 6 bits and below.

\begin{table}[tbp]
\centering
\begin{tabular}{cccccc}
\toprule
Split & A/F  & Layer & From & To & Drop \\
\midrule
Test & Sigmoid & 1 & 32 & 16 & -0.011923 \\
Test & Sigmoid & 1 & 5 & 4 & -0.011844 \\
Test & Sigmoid & 2 & 6 & 5 & -0.010145 \\
Test & Sigmoid & 3 & 7 & 6 & -0.021927 \\
Validation & Sigmoid & 1 & 5 & 4 & -0.012812 \\
Validation & Sigmoid & 2 & 16 & 8 & -0.015463 \\
Validation & Sigmoid & 2 & 6 & 5 & -0.015374 \\
Validation & Sigmoid & 3 & 5 & 4 & -0.011994 \\
\bottomrule
\end{tabular}
\caption{PPI. Accuracy drops across bit-width reductions for the Sigmoid activation function. Only over 1\% drops are reported.}
\label{tab:ppi-bitwidth-drops-Sigmoid}
\end{table}

According to the Table\ref{tab:ppi-bitwidth-drops-Sigmoid}, Sigmoid shows moderate sensitivity, with drops distributed across the Test and Validation splits. Small decreases occur from 32 to 16, from 7 to 6, and from 6 to 5 bits, with additional declines from 5 to 4 bits in both splits. These degradations are consistent but remain relatively small ( less than 2\%), suggesting mild vulnerability to quantization, mainly below 6 bits.

\begin{table}[tbp]
\centering
\begin{tabular}{cccccc}
\toprule
Split & A/F  & Layer & From & To & Drop \\
\midrule
Test & SiLU & 1 & 32 & 16 & -0.012687 \\
Test & SiLU & 1 & 7 & 6 & -0.015057 \\
Test & SiLU & 2 & 5 & 4 & -0.022375 \\
Test & SiLU & 3 & 16 & 8 & -0.015068 \\
Validation & SiLU & 2 & 5 & 4 & -0.019180 \\
\bottomrule
\end{tabular}
\caption{PPI. Accuracy drops across bit-width reductions for the SiLU activation function. Only over 1\% drops are reported.}
\label{tab:ppi-bitwidth-drops-SiLU}
\end{table}

According to the Table\ref{tab:ppi-bitwidth-drops-SiLU}, SiLU exhibits stable accuracy at higher precisions but shows drops from 32 to 16 bits and from 7 to 6 bits in the Test split. The most notable drop occurs from 5 to 4 bits, appearing in both Test and Validation splits. This indicates that SiLU maintains robustness down to 6 bits but experiences degradation at 5 bits and below.

\begin{table}[tbp]
\centering
\begin{tabular}{cccccc}
\toprule
Split & A/F  & Layer & From & To & Drop \\
\midrule
Test & Softplus & 1 & 7 & 6 & -0.027224 \\
Test & Softplus & 2 & 32 & 16 & -0.020117 \\
Test & Softplus & 2 & 5 & 4 & -0.010378 \\
Validation & Softplus & 1 & 6 & 5 & -0.014675 \\
Validation & Softplus & 2 & 32 & 16 & -0.016789 \\
Validation & Softplus & 2 & 5 & 4 & -0.012770 \\
\bottomrule
\end{tabular}
\caption{PPI. Accuracy drops across bit-width reductions for the Softplus activation function. Only over 1\% drops are reported.}
\label{tab:ppi-bitwidth-drops-Softplus}
\end{table}

According to the Table\ref{tab:ppi-bitwidth-drops-Softplus}, Softplus shows several drops across both Test and Validation splits. The most substantial decreases occur from 32 to 16 bits and from 7 to 6 bits, with additional drops from 5 to 4 bits. These results indicate that while Softplus performs reliably at 16–8 bits, accuracy becomes unstable at lower precisions.

\begin{table}[tbp]
\centering
\begin{tabular}{cccccc}
\toprule
Split & A/F  & Layer & From & To & Drop \\
\midrule
Validation & ELU & 2 & 7 & 6 & -0.012918 \\
Validation & ELU & 2 & 6 & 5 & -0.013019 \\
\bottomrule
\end{tabular}
\caption{PPI. Accuracy drops across bit-width reductions for the ELU activation function. Only over 1\% drops are reported.}
\label{tab:ppi-bitwidth-drops-ELU}
\end{table}

According to the Table\ref{tab:ppi-bitwidth-drops-ELU}, ELU is one of the most stable activations, showing only two Validation drops, both small from 7 to 6 and from 6 to 5 bits. No degradations were observed in the Test split. This suggests ELU is highly robust to quantization down to 6 bits, with only minor sensitivity at very low precision.

As a conclusion of this analysis, we found that across all activation functions, quantization down to 8 and 6 bits generally preserved the accuracy of the PPI models, with most degradations remaining below 1–2\%. Noticeable accuracy drops emerged primarily in the transitions from 7 to 6 bits and from 6 to 5 bits, while the most substantial losses consistently appeared from 5 to 4 bits, regardless of activation function. Among the functions tested, trRELU and Softplus showed the greatest sensitivity to reduced precision, with multiple drops across splits and layers, whereas ReLU, ReLU6, and ELU demonstrated the highest robustness, exhibiting only occasional and comparatively minor decreases. Overall, these results indicate that ACR-GNNs retain stable performance at moderate bit-widths (16–8–6 bits), and significant accuracy degradation occurs mainly when entering the very low-precision regime (less that 5 bits).


\end{document}